\RequirePackage{fix-cm}
\documentclass[aos,preprint]{imsart}

\RequirePackage{amsthm,amsmath,amsfonts,amssymb,mathtools,bm}
\RequirePackage[numbers]{natbib}

\RequirePackage{graphicx}
\RequirePackage{tikz}
\usetikzlibrary{arrows.meta,positioning,calc}
\RequirePackage{booktabs}
\RequirePackage{adjustbox}
\RequirePackage{xcolor}
\RequirePackage{physics}
\RequirePackage{enumitem}
\RequirePackage{multirow}
\DeclareFontShape{T1}{ptm}{m}{scit}{<->ssub * ptm/m/sc}{}
\graphicspath{{example1/}}

\startlocaldefs

\theoremstyle{plain}
\newtheorem{theorem}{Theorem}
\newtheorem{lemma}{Lemma}
\newtheorem{proposition}{Proposition}
\newtheorem{corollary}[theorem]{Corollary}

\theoremstyle{definition}
\newtheorem{assumption}{Assumption}
\newtheorem{definition}{Definition}
\newtheorem{remark}{Remark}
\newtheorem{exam}{Example}
\newtheorem{example}{Experiment}

\newcommand{\PP}{\mathbb{P}}
\newcommand{\EE}{\mathbb{E}}
\newcommand{\E}{\mathbb E}
\newcommand{\R}{\mathbb{R}}
\newcommand{\PsiB}{\Psi_B}
\newcommand{\psiB}{\psi_B}

\newcommand{\QB}{\mathcal{Q}_B}
\newcommand{\ind}{\mathbf{1}}

\endlocaldefs

\begin{document}

\begin{frontmatter}

\title{Studentized Cheap Bootstrap: Achieving Higher-Order Coverage Accuracy with Low Computation}
\runtitle{Studentized Cheap Bootstrap}

\begin{aug}

\author[A]{\fnms{Shengyi}~\snm{He}\ead[label=e1]{sh3972@columbia.edu}}
\author[A]{\fnms{Henry}~\snm{Lam}\ead[label=e2]{henry.lam@columbia.edu}}
\author[A]{\fnms{Yunhao}~\snm{Yan}\ead[label=e3]{yy2882@columbia.edu}}

\address[A]{
Department of Industrial Engineering and Operations Research,
Columbia University,\\
\printead[presep={\ }]{e1,e2,e3}
}

\end{aug}

\begin{abstract}
The bootstrap is a versatile method for quantifying statistical uncertainty. Among its variants, a popular approach, the studentized bootstrap, provably achieves higher-order coverage error reduction compared to other benchmarks. However, its implementation typically requires an analytical form of the standard error, or otherwise an additional layer of resampling effort which can be computationally expensive. In this paper, we introduce what we call the \emph{studentized cheap bootstrap} that achieves the same higher-order coverage accuracy as the conventional studentization, but substantially thinning the computational effort in the additional resampling layer to only very few Monte Carlo replications. Intriguingly, while conventional wisdom views ``studentization'' as an informal link between the bootstrap and $t$-distribution, we provide a first recognition that this link is in fact formal, notably with a distinct insight that the degree of freedom in the $t$-distribution corresponds to the Monte Carlo computation effort in the additional resampling layer, rather than the data size as in traditional thinking. Moreover, our desirable higher-order coverage accuracy builds crucially on this insight, as well as explicit calculations and geometric analyses of higher-order terms in the Edgeworth and Cornish--Fisher expansions tailored to limiting $t$-distributions.
\end{abstract}



\end{frontmatter}

\section{Introduction}\label{sec intro}
The bootstrap is a common approach for statistical uncertainty quantification. Its main idea hinges on the resemblance between the resampling distribution and the original sampling distribution, serving as an implementable mechanism to measure the uncertainty of an empirical estimator. The theoretical foundations and diverse applications of this method have been extensively documented in several classic monographs, including \cite{efron1994introduction, davison1997bootstrap, shao2012jackknife, hall2013bootstrap}. Compared with analytical approaches such as the delta method, the bootstrap advantageously does not require structural knowledge of the problem (typically in the form of ``gradient information'' such as the influence function). On the other hand, the bootstrap is known to be computationally intensive, as its implementation requires a sufficiently large number of Monte Carlo runs---namely, the processes of resampling and re-evaluating the model---in order to obtain an accurate statistical summary of the resampling distribution.

Our focus in this paper is on attaining higher-order coverage accuracy. Namely, we are interested in generating confidence intervals (CIs) whose coverage errors are small relative to the sample size. To explain, a statistically valid CI typically means that its coverage achieves the prescribed nominal level (e.g., 95\%) asymptotically as the sample size grows. However, in any finite sample, the coverage error relative to the nominal level is nonzero and can be substantial. For ``simple'' approaches, such as the so-called basic bootstrap or the delta method, this error is typically of order $O(n^{-1})$ for two-sided intervals and $O(n^{-1/2})$ for one-sided intervals, where $n$ is the sample size. By higher-order coverage accuracy, we mean that the interval can achieve $O(n^{-2})$ for two-sided intervals and $O(n^{-1})$ for one-sided intervals. That is, they attain a coverage error that is one order smaller than that of the ``simple'' approaches. To this end, we discuss the main existing methods for attaining higher-order coverage-accurate CIs.

At a high level, there are roughly four major lines of methods in the bootstrap literature. First, the studentized bootstrap \cite{hall1986bootstrap,hall1988theoretical,beran1987prepivoting,babu1983inference} constructs confidence intervals by using quantiles of a pivotal statistic. In particular, this pivotal statistic is formed using a standard error computed from the resample estimates, rather than from the original estimate. This studentization induces a useful cancellation effect and leads to higher-order accuracy: the coverage error is of order $O(n^{-1})$ for one-sided intervals, and can be further improved to $O(n^{-2})$ for two-sided intervals when the interval is constructed symmetrically around the original estimate \cite{hall1988symmetric}. The main drawback of this approach is that it requires either an analytical or otherwise computable expression for the standard error \cite{davison1997bootstrap}; if such a formula is unavailable, one must approximate the standard error by additional resampling, which introduces another layer of computation. Second, bootstrap iteration (or the double bootstrap) \cite{hall1986bootstrap,hall1988bootstrap,lee1995asymptotic,booth1994monte} achieves higher-order accuracy by correcting the bootstrap quantiles obtained from the basic bootstrap, i.e., those based on a non-pivotal statistic, so that the resulting interval more closely matches the nominal coverage level. This correction requires a ``bootstrap on a bootstrap'', leading to a computational burden analogous to that of the studentized bootstrap when no analytical closed-form expression is available for the standard error. Third, direct Edgeworth corrections \cite{abramovitch1985edgeworth,withers1984asymptotic,hall1983inverting} use algebraic derivations to improve the calibration of intervals. These procedures require little computational effort. However, strictly speaking they are not bootstrap methods, since they are not automated and require model-specific derivations; see, for example, the linear regression setting in \cite{qumsiyeh1990edgeworth}. Finally, we mention the well-known BCa and ABC methods \cite{efron1987bca,efron2020automatic,diciccio1996bootstrap,diciccio1992more}. The ABC method is fast and accurate, but is primarily designed for parametric settings, especially exponential families, and requires additional analytical input such as local gradient information. In contrast, the BCa method applies to black-box target statistics and remains computationally attractive. Besides the basic bootstrap cost, it requires only an additional jackknife step to estimate the acceleration constant, with computational cost proportional to the sample size. Under suitable regularity conditions, standard BCa intervals are second-order accurate, with coverage error of order $O(n^{-1})$ \cite{hall1988theoretical,hall2013bootstrap}. However, unlike symmetric two-sided studentized bootstrap intervals, the standard BCa procedure does not in general achieve the sharper $O(n^{-2})$ coverage error for two-sided confidence intervals.

As the above discussion shows, existing bootstrap methods with higher-order accuracy typically require either analytical knowledge of the standard error or nested resampling with a \emph{multiplicatively large} computational cost, in the sense that the procedure requires first running an outer layer of resampling, and for each outer resample, a large number of inner resamples and re-evaluations are conducted, so that the total number of model evaluations is the product of the outer and inner numbers of resamples. The standard BCa method mentioned above is a notable exception in terms of computation, but it generally achieves only $O(n^{-1})$ coverage accuracy rather than the sharper $O(n^{-2})$ accuracy of symmetric two-sided studentized bootstrap intervals. For problems arising in modern machine learning or simulation, unfortunately neither requirement is trivial: these models can be black-box and hence may not admit a closed-form standard error (and in fact, if they do, then we may not even need to use the bootstrap in the first place), and their model evaluation can be dauntingly expensive, since it may entail running large-scale optimization routines or high-fidelity simulators. Motivated by this challenge, our main goal is to devise a higher-order coverage-accurate method that \emph{does not require a closed-form standard error or a multiplicatively large resampling effort}.

Our main methodology is what we call the \emph{studentized cheap bootstrap (SCB)}. Procedurally, this approach resembles the studentized bootstrap in settings where the standard error does not admit a closed form. In such settings, the resampled standard error in the denominator of the pivotal statistic would need to be approximated by an inner layer of resampling, thus leading to the ``bootstrap on a bootstrap'' described above. Our SCB is a subtle twist on the studentized bootstrap: it uses a very small number of resamples (e.g., 5) in this inner layer. In this way, the overall number of model re-evaluations is not the product of two large numbers, but only a small multiple of the number of outer resamples.

While the twist in SCB relative to the conventional studentized bootstrap appears minimal, the theory leading to its attainment of higher-order coverage accuracy requires a statistical argument that is very distinct from the conventional view. To understand this, in the conventional bootstrap argument, the implementation of resampling, whether in the outer or inner layer, requires a large number of resamples (or, synonymously, Monte Carlo runs). This is because the established theory of the bootstrap builds on the resemblance of the resampling distribution to the original sampling distribution, in the sense that these two distributions are approximately the same as the data size grows. This closeness is based on convergence results that exclude the consideration of Monte Carlo size. When implementing the bootstrap, Monte Carlo size enters because many resamples have to be generated in order to obtain a good Monte Carlo approximation of the summary statistics (such as the quantile and the standard error) for the resampling distribution. In other words, the conventional understanding of the bootstrap requires that the number of resamples, both in the outer and inner layers, be sufficiently large; otherwise the coverage guarantee breaks down.

For our SCB, the underlying theory that leads to higher-order coverage improvement uses the recent so-called cheap bootstrap idea \cite{lam2022,lam2023bootstrap}, but also a first formal link between ``studentization'' in the studentized bootstrap and the Student's $t$-statistic. The cheap bootstrap, which has been applied recently to problems such as stochastic gradient descent \cite{lam2026cheap}, simulation uncertainty quantification \cite{huang2023efficient,chowdhury2025efficient,lam2022cheap}, and biostatistics \cite{ohlendorff2025cheap}, allows the use of very few resamples to produce statistically valid CIs. Instead of using the sample-resample resemblance followed by a Monte Carlo approximation, the cheap bootstrap looks at the \emph{joint} distribution of the original empirical estimate and the resample estimates, in a sense taking an integrated perspective on the data noise and the Monte Carlo resampling noise. In particular, under assumptions similar to those of conventional bootstraps, this joint distribution, for a \emph{fixed} number of resamples, can be seen to be asymptotically independent as well as jointly Gaussian. This in turn allows the construction of a pivotal statistic that leads to a $t$ limit, notably with a small, fixed number of resamples used to form the standard error. The main insight of SCB is that, \emph{when we replace the $t$-quantile of the cheap bootstrap interval with a bootstrap calibration, we obtain higher-order coverage accuracy}. Intriguingly, this leads to a new view of the studentized bootstrap: while conventional wisdom claims that ``studentization'' in the studentized bootstrap has no formal link with the Student $t$-distribution, since it merely refers to pivotalization with a standard error that is procedurally analogous to the textbook $t$-statistic for the population mean, the idea behind SCB is rooted differently. More precisely, we provide an explicit link between the studentized bootstrap and the $t$-distribution, with the distinct insight that the degree of freedom in the $t$-distribution corresponds to the number of resamples in the inner resampling layer, rather than the data size. That is, if we take the degree of freedom to be a \emph{fixed} number of inner resamples, then indeed the studentized bootstrap is asymptotically $t$-distributed. This turns out to provide the foundation for a higher-order coverage correction that does not require a growing number of inner resamples.

Figure \ref{fig:bootstrap_relationships}, which we discuss in more detail in the next section, shows the relational positioning and strengths of SCB in terms of both coverage error and computational effort compared with existing bootstrap approaches. The basic bootstrap and the standard error bootstrap are ``simple'' approaches that do not have higher-order coverage accuracy, but they require only one layer of resampling effort. The studentized bootstrap, which uses pivotalization to approximate quantiles, can be viewed as an enhancement of the basic bootstrap, which does not use pivotalization, but it can also be viewed as an enhancement of the standard error bootstrap by replacing the $z$-quantiles with bootstrapped quantiles. This can yield a higher-order coverage error, but generally requires nested resampling that is computationally expensive. The cheap bootstrap can be viewed as a ``computational relaxation'' of the standard error bootstrap, by using few resamples to approximate the standard error; this turns out to be coverage-valid when using $t$-quantiles in lieu of $z$. Our SCB integrates ideas from both the studentized bootstrap and the cheap bootstrap. On the one hand, it replaces the $t$-quantiles in the cheap bootstrap with bootstrapped quantiles. On the other hand, it performs a computational relaxation of the studentized bootstrap by drastically reducing the number of inner resamples needed to approximate the standard error. This integration leads to both a higher-order coverage accuracy and a light resampling computation load that is of a similar order to that of the basic and standard error bootstraps.
\begin{figure}[!b]
\centering
\scalebox{0.78}{
\begin{tikzpicture}[
    >=Latex,
    font=\small,
    method/.style={
        draw,
        rounded corners=4pt,
        thick,
        align=left,
        minimum height=2.35cm,
        inner sep=6pt
    },
    leftmethod/.style={method, text width=5.25cm},
    rightmethod/.style={method, text width=5.95cm},
    classical/.style={leftmethod, fill=orange!10},
    se/.style={leftmethod, fill=violet!10},
    cheap/.style={leftmethod, fill=teal!10},
    advanced/.style={rightmethod, fill=blue!10},
    scb/.style={rightmethod, fill=green!12},
    arrow/.style={-{Latex[length=3mm]}, thick},
    lab/.style={font=\small, inner sep=1.5pt, align=center}
]

\node[classical] (basic) at (0,5.4) {
\makebox[\linewidth][c]{\textbf{Basic Bootstrap}}\\[1mm]
Output: $\left[
\hat\theta-\{\theta_b^*-\hat\theta\}_{\frac{1+\alpha}{2}},
\hat\theta-\{\theta_b^*-\hat\theta\}_{\frac{1-\alpha}{2}}
\right]$\\
Coverage error: $O(n^{-1})$\\
Statistic Evaluations: $1+B_1$
};

\node[se] (seboot) at (0,1.65) {
\makebox[\linewidth][c]{\textbf{Standard Error Bootstrap}}\\[1mm]
Output: $\left[
\hat\theta\pm z_{(1+\alpha)/2}\sqrt{\mathrm{Var}_*(\theta^*)}
\right]$\\
Coverage error: $O(n^{-1})$\\
Statistic Evaluations: $1+B_1$
};

\node[cheap] (cheapboot) at (0,-2.15) {
\makebox[\linewidth][c]{\textbf{Cheap Bootstrap}}\\[1mm]
Output: $\left[
\hat\theta\pm t_{B,(1+\alpha)/2}S
\right]$\\
Coverage error: $O(n^{-1})$\\
Statistic Evaluations: $1+B$
};

\node[advanced] (stud) at (9.4,5.4) {
\makebox[\linewidth][c]{\textbf{Studentized Bootstrap}}\\[1mm]
Output: $\left[
\hat\theta\pm \hat\sigma\{|T_{\mathsf N,b_1}^*|\}_{\alpha}
\right]$\\
Coverage error: $O(n^{-2})$\\
Statistic Evaluations: $(1+B_1)(1+B_2)$,\\
or $1+B_1$ if closed-form $\hat{\sigma}$.
};

\node[scb] (scbnode) at (9.4,-2.15) {
\makebox[\linewidth][c]{\textbf{Studentized Cheap Bootstrap}}\\[1mm]
Output: $\left[
\hat\theta\pm \hat S\,\hat u_{\alpha}^{\mathrm{ts}}
\right]$\\
Coverage error: $O(n^{-2})$\\
Statistic Evaluations: $(1+B_1)(1+B)$
};

\node[align=center, font=\small] at (5.35,1.55) {
$B$ small\\
$B_1,B_2\to\infty$ large
};

\draw[arrow] (basic.east) -- (stud.west);
\draw[arrow] (seboot.east) -- (stud.west);
\draw[arrow] (cheapboot.east) -- (scbnode.west);

\draw[arrow] ([xshift=-18pt]seboot.south) -- ([xshift=-18pt]cheapboot.north);
\draw[arrow] ([xshift=18pt]cheapboot.north) -- ([xshift=18pt]seboot.south);

\draw[arrow] (scbnode.north) -- (stud.south);

\node[lab] at (4.55,6.15) {Gaussian\\pivotalization};

\node[lab] at (2.75,3.40) {Replace $z$-quantiles\\with $T_{\mathsf N}^*$-quantiles};

\node[lab] at (4.55,-1.3) {$t$-pivotalization\\$B$ deg. of freedom};

\node[lab] at (-2.0,-0.25) {using $t$-quantiles};

\node[lab] at (1.25,-0.25) {$B\uparrow B_1$};

\node[lab] at (10.85,1.55) {$B\uparrow B_2$};

\end{tikzpicture}
}
\caption{Schematic comparison of five bootstrap methods. Here $\alpha$ denotes the nominal coverage level, and $\{X_b\}_p$ denotes the empirical $p$-quantile of $\{X_b\}$.}
\label{fig:bootstrap_relationships}
\end{figure}
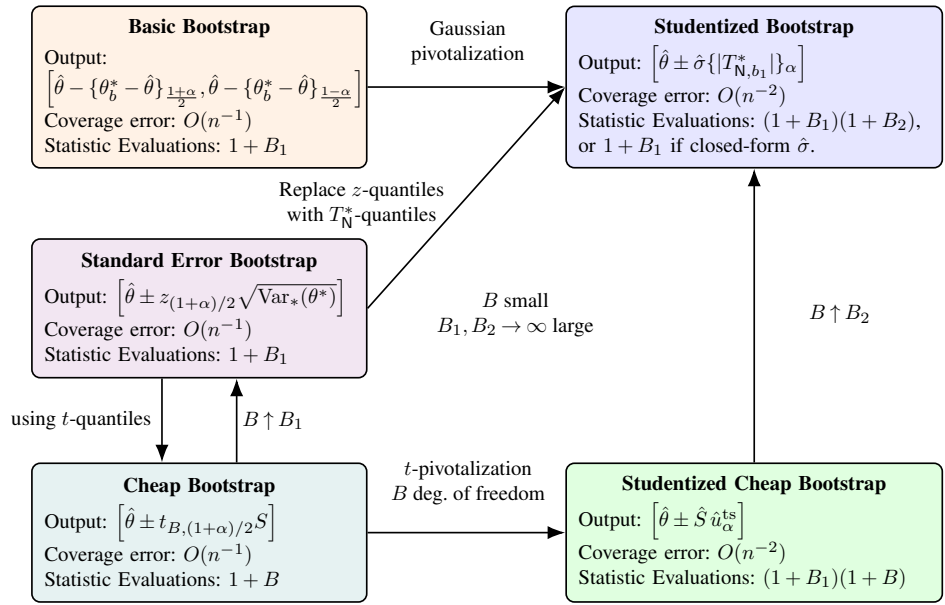

In terms of mathematical developments, we analyze and establish the higher-order coverage accuracy of SCB via Edgeworth and Cornish--Fisher expansions based on the Student's $t$ distribution. Most existing Edgeworth expansions have been developed for pivotal statistics with either Gaussian or $\chi^2$ limiting distributions \cite{barndorff1984bartlett,diciccio1991empirical}. By contrast, comparable expansions built on the Student's $t$ distribution remain much less developed, except for the recent contributions in \cite{he2021higherorder} and \cite{lam2022}. In this paper, we derive explicit closed-form expressions for the Edgeworth and Cornish--Fisher expansions associated with a $t$-distribution. In particular, we obtain explicit analogs of the classical Edgeworth polynomials which, to the best of our knowledge, are derived here for the first time in closed form for the Student's $t$ setting. These functions are quasi-polynomials in the sense that they are finite linear combinations of terms of the form $x^k/(B+x^2)^{\ell/2}$, where $B$ is the number of bootstrap resamples, and $k$ and $\ell$ are integers. They retain key structural properties of the classical Edgeworth polynomials, including infinite smoothness and the odd--even parity pattern across expansion orders, which allow us to derive the desired orders of the coverage errors of SCB intervals.

We close this introduction by positioning our work alongside other non-bootstrap methods of constructing CIs with higher-order accuracy, as well as methods to reduce resampling efforts. Besides the direct Edgeworth corrections discussed above, higher-order coverage accuracy has also been pursued through several classical analytic, non-resampling routes. Bartlett corrections improve $\chi^2$ approximations for likelihood-ratio statistics through analytic multiplicative corrections \cite{bartlett1937properties,lawley1956general}. Closely related ideas appear in empirical likelihood, the nonparametric counterpart of Wilks-type likelihood ratio, which can also admit Bartlett corrections under suitable regularity conditions \cite{owen1988empirical,owen1990empirical,diciccio1991empirical}. Another important route is saddlepoint approximation, which uses cumulant-generating functions to obtain highly accurate distributional tail approximations \cite{daniels1954saddlepoint,lugannani1980saddlepoint,jensen1995saddlepoint}. Higher-order likelihood theory provides yet another set of refinements, including modified likelihood roots, profile-likelihood adjustments, and related likelihood-based corrections for nuisance-parameter problems \cite{barndorff1983formula,barndorff1991modified,cox1987parameter,barndorffnielsen1994inference,reid2003asymptotics,severini2000likelihood,brazzale2007applied}. Our approach differs from all these methods in that we are solely resampling-based, which in a sense is easy to implement without additional analytic overhead. On the other hand, regarding the control of computational effort in bootstrap inference, a focus has been on reducing the sample size in bootstrap replications, including the classical $m$-out-of-$n$ bootstrap \cite{politis1999subsampling,bickel2012resampling}, and the more recent bags of little bootstraps \cite{kleiner2014scalable} and subsampled double bootstrap \cite{sengupta2016subsampled}. These methods reduce the effective sample size used to create each bootstrap replicate, and retain correct inference by suitable scale-back of the resulting subsampled estimates. While they can have benefits such as consistency under milder assumptions, they have also been shown to lighten resampling computation for some problems. Other approaches that save computation use analytic approximations via for instance the infinitesimal jackknife \cite{giordano2019swiss,schulam2019can,lu2020uncertainty}, including for problems such as bagging and random forests that can involve multiple layers of resampling when directly using the bootstrap \cite{wager2014confidence,lopes2019estimating}. Nonetheless, none of these studies focus on or push the performance to higher-order coverage accuracy as investigated in this paper.

The remainder of this paper is organized as follows. Section~\ref{sec 2 established methods} reviews and discusses the pros and cons of established bootstrap methods in terms of higher-order accuracy and computational cost. Section~\ref{sec scb procedure} presents our SCB procedure. Section~\ref{sec4 mainroad} develops the theoretical foundation of SCB, including higher-order Edgeworth and Cornish--Fisher expansions for Student's $t$-limiting statistics, quantile error bounds, and the resulting coverage error analysis. Section~\ref{num res} provides numerical results illustrating the empirical performance of SCB relative to existing benchmark methods.

\section{Bootstraps and Cheap Bootstraps}\label{sec 2 established methods}
We first frame the problem setting. Let $\theta\coloneqq\theta(P)\in\mathbb R$ denote a target statistical quantity, where $P$ is the unknown data distribution. Suppose we have i.i.d.\ data $\mathcal{X}=\{\mathbf{X}_1,\dots,\mathbf{X}_n\}$, where $\mathbf{X}_i\in\mathbb R^d$ and the sample size is $n$. Our goal is to construct an $\alpha$-level CI. Equivalently, $1-\alpha$ denotes the significance level in the notation of this paper. More precisely, we consider the case of a two-sided CI, in which we are interested in obtaining two quantities $\mathsf L_n$ and $\mathsf U_n$, or an interval $[\mathsf L_n,\mathsf U_n]$, such that
\begin{equation}
\mathbb P(\mathsf L_n\leq\theta\leq \mathsf U_n)=\alpha+o(1)\label{two-sided asymptotic exact}
\end{equation}
On the other hand, for a one-sided CI, we are interested in obtaining $\mathsf L_n$ such that
\begin{equation}
\mathbb P(\mathsf L_n\leq \theta)=\alpha+o(1)\label{lower asymptotic exact}
\end{equation}
which we call an upper interval, or obtaining $\mathsf U_n$ such that
\begin{equation}
\mathbb P(\theta\leq \mathsf U_n)=\alpha+o(1)\label{upper asymptotic exact}
\end{equation}
which we call a lower interval. The probability $\mathbb P$ is with respect to the data size $n$, but also potentially to any Monte Carlo noise involved in the construction procedures of $\mathsf L_n$ and $\mathsf U_n$ (whose role will be clearer in the sequel). The notation $o(1)$ means a term that goes to 0 as $n\to\infty$ (again, it is possible that Monte Carlo size would also play a role, and this will become clearer later). We call the coverage guarantees \eqref{two-sided asymptotic exact}, \eqref{lower asymptotic exact}, and \eqref{upper asymptotic exact} asymptotically exact, which is a basic validity certificate for CIs.

Our focus in this paper is higher-order coverage accuracy. In particular, we are interested in the order of the $o(1)$ terms in \eqref{two-sided asymptotic exact}, \eqref{lower asymptotic exact}, and \eqref{upper asymptotic exact}. These $o(1)$ terms are the coverage error of the CI, namely the difference between its empirical coverage and the nominal level $\alpha$. Roughly speaking, and as we will discuss in more detail below, it is relatively easy to obtain procedures that achieve $O(n^{-1})$ coverage error for the two-sided case, and $O(n^{-1/2})$ for the one-sided case. The smaller coverage error in the two-sided case is due to a cancellation effect arising from the parity structure of the coverage expansion; see Section 3.5.5 of \cite{hall2013bootstrap}. We are particularly interested in constructing CIs with higher-order coverage accuracy. In particular, relative to the ``easy'' intervals, we aim to improve the coverage error from $O(n^{-1})$ to $O(n^{-2})$ in the two-sided case, and from $O(n^{-1/2})$ to $O(n^{-1})$ in the one-sided case. These intervals perform better in small-sample situations, or in problems where the sample size is small relative to the problem dimension, as discussed in a range of works, e.g., \cite{hall1988symmetric,hall1988theoretical,diciccio1996bootstrap}.

We consider the bootstrap, which is arguably the most common approach to constructing CIs and whose main advantage is that it is data-driven and automated. Here we give a quick overview of the idea behind the bootstrap, which will be useful in paving the way to explain the so-called cheap bootstrap and subsequently SCB. The starting point is the use of resampling. To estimate $\theta$, a natural point estimator is $\hat{\theta}\coloneqq \theta(\hat P_n)$, where $\hat P_n$ is the empirical distribution constructed from the sample data $\mathcal{X}=\{\mathbf{X}_1,\dots,\mathbf{X}_n\}$, i.e.,
$\hat P_n(\cdot)=\frac{1}{n}\sum_{i=1}^n\delta_{\mathbf X_i}(\cdot)$ where $\delta_{\mathbf X_i}$ is the Dirac measure at $\mathbf X_i$. By a resample, we mean sampling with replacement from $\mathcal X$ $n$ times to obtain $\mathcal X^*=\{\mathbf X_{1}^*,\dots,\mathbf X_{n}^*\}$ and hence the resample estimate $\theta^*:=\theta(P_n^*)$, where $P_n^*$ is the empirical distribution formed from the resample, i.e., $P_n^*(\cdot)=\frac{1}{n}\sum_{i=1}^n\delta_{\mathbf X_i^*}(\cdot)$.

The bootstrap works on the principle that the sampling distribution of the original estimate $\hat\theta$, and the resampling distribution of $\theta^*$, are similar. To make this slightly more precise, asymptotically as data size $n$ grows, the original estimate $\hat\theta$, under standard regularity conditions, follows a central limit theorem
\begin{equation}
\sqrt n(\hat\theta-\theta)\Rightarrow \mathcal N\label{CLT original}
\end{equation}
where $\Rightarrow$ denotes weak convergence. On the other hand, the resample estimate $\theta^*$ satisfies a conditional weak convergence result; see, e.g., Chapter 3.6 of \cite{van1996weak}
\begin{equation}
\sqrt n(\theta^*-\hat\theta)\Rightarrow \mathcal N\text{\ \ conditional on\ \ }\mathbf X_1,\mathbf X_2,\cdots\label{CLT resample}
\end{equation}
with the limit $\mathcal N$ being the same random variable as in the central limit theorem \eqref{CLT original} for $\theta$ above. The bootstrap idea operationalizes the similar limits above in order to construct CIs. To this end, if we can approximate some summary statistics of the sampling distribution of $\hat\theta$, then we can construct CIs for $\theta$. For example, if we can find two quantiles, say $\mathsf L_n$ and $\mathsf U_n$, such that
$$\mathbb P(\mathsf L_n\leq\hat\theta-\theta\leq\mathsf U_n)=\alpha+o(1)$$
then an $\alpha$-level CI would be $[\hat\theta-\mathsf U_n,\hat\theta-\mathsf L_n]$. Using the resemblance of the resampling distribution in \eqref{CLT resample} to \eqref{CLT original}, we can use $\hat{\mathsf L}_n,\hat{\mathsf U}_n$ such that
$$\mathbb P(\hat{\mathsf L}_n\leq\theta^*-\hat\theta\leq\hat{\mathsf U}_n|\mathcal{X})=\alpha+o(1)$$
That is, $[\hat\theta-\hat{\mathsf U}_n,\hat\theta-\hat{\mathsf L}_n]$ is a valid $\alpha$-level CI. A main insight of the bootstrap is that $\hat{\mathsf L}_n,\hat{\mathsf U}_n$ are \emph{computable}, by running a large number of Monte Carlo runs to resample $P_n^*$ and re-evaluate the estimate $\theta^*=\theta(P_n^*)$. Specifically, we repeatedly resample from $\mathcal X$ to obtain $\mathcal X_b^*=\{\mathbf X_{b,1}^*,\dots,\mathbf X_{b,n}^*\}$ and compute the resample estimate $\theta_b^*=\theta(P_{b,n}^*)$, where $P_{b,n}^*$ is the empirical distribution of $\mathcal X_{b}^*$, for $b=1,\dots,B_1$. By taking the empirical quantiles of $\theta_b^*-\hat\theta$, we obtain approximations of $\hat{\mathsf L}_n$ and $\hat{\mathsf U}_n$.

In the following, we present some important bootstrap variants using the principle above, and beyond it, that collectively form the basis of Figure \ref{fig:bootstrap_relationships}.

\subsection{Basic, Standard Error, and Studentized Bootstraps}

First, the development described in the previous subsection essentially leads to what is known as the \emph{basic bootstrap}. To implement it, we run Monte Carlo replicates of $\theta_b^*$ for $b=1,\ldots,B_1$, with $B_1$ sufficiently large, and use them to approximate the $(1-\alpha)/2$ and $(1+\alpha)/2$ quantiles of the centered bootstrap statistics $\theta_b^*-\hat\theta$. The resulting CI is  $\left[
\hat\theta-\{\theta_b^*-\hat\theta\}_{\frac{1+\alpha}{2}},
\,
\hat\theta-\{\theta_b^*-\hat\theta\}_{\frac{1-\alpha}{2}}
\right]$, or equivalently
$[\,2\hat\theta-\{\theta_b^*\}_{(1+\alpha)/2},\; 2\hat\theta-\{\theta_b^*\}_{(1-\alpha)/2}\,]$, where in general we denote by $\{x_b\}_{p}$ the $p$-quantile of a set of samples $\{x_b\}$ indexed by $b$ (formed by simple order statistics). This approach gives coverage error $O(n^{-1})$ in the two-sided case and $O(n^{-1/2})$ in the one-sided case.

Using a similar principle but, instead of quantiles, operationalizing with the standard error as the summary statistic leads to the so-called \emph{standard error bootstrap} \cite{efron1981nonparametric}. To explain this, suppose the limit $\mathcal N$ in \eqref{CLT original} and \eqref{CLT resample} is the typical Gaussian case. Then we can use $[\hat\theta\pm z_{(1+\alpha)/2}\sqrt{\mathrm{Var}(\hat\theta)}]$ as the $\alpha$-level CI, where $z_{(1+\alpha)/2}$ is the standard normal $(1+\alpha)/2$-level quantile. Similar to the development of the basic bootstrap, we do not know $\mathrm{Var}(\hat\theta)$, but we can approximate it using the resampling distribution. That is, we approximate $\mathrm{Var}(\hat\theta)$ by
$\mathrm{Var}_*(\theta^*) \coloneqq \sum_{b=1}^{B_1}(\theta_b^*-\bar\theta^*)^2/(B_1-1)$,
where $\bar\theta^* \coloneqq \sum_{b=1}^{B_1}\theta_b^*/B_1$ is the mean of the bootstrap replicates $\theta_b^*$. The resulting CI is
$[\,\hat\theta \pm z_{(1+\alpha)/2}\sqrt{\mathrm{Var}_*(\theta^*)}\,]$. Note that, like the basic bootstrap, we need to run a large number of Monte Carlo replications to evaluate the $\theta_b^*$'s in the implementation of the standard error bootstrap. Moreover, its coverage error is also similar to that of the basic bootstrap in the case where the limit $\mathcal N$ is Gaussian. We summarize the key properties and implementation requirements of the basic and standard error bootstraps in the upper-left corner of Figure \ref{fig:bootstrap_relationships}.

Next, the \emph{studentized bootstrap} \cite{davison1997bootstrap}, also called the percentile-$t$ method in \cite{hall2013bootstrap}, resembles the basic bootstrap, but emphasizes the use of a pivotal statistic. Instead of approximating the quantiles of $\hat\theta-\theta$, it aims to approximate those of $T_{\mathsf N}\coloneqq (\hat\theta-\theta)/\hat\sigma$, where $\hat\sigma$ is a consistent standard error estimator of $\hat\theta$. The latter implies that $T_{\mathsf N}$, analogous to \eqref{CLT original}, converges weakly to the standard normal distribution. Moreover, the resemblance of the resampling distribution to the original sampling distribution still applies, so that $T_{\mathsf N}^*$, the resample version of $T_{\mathsf N}$, also converges to standard normal conditional on the data. This allows the execution of a bootstrap approximation. The studentized bootstrap is primarily motivated by higher-order coverage accuracy. Compared with non-pivotal bootstrap methods such as the basic bootstrap, it reduces the one-sided CI coverage error from $O(n^{-1/2})$ to $O(n^{-1})$, and the two-sided CI coverage error from $O(n^{-1})$ to $O(n^{-2})$ when combined with a symmetric construction; see \cite{hall2013bootstrap}.

While the studentized bootstrap can attain sharper coverage accuracy, it pays a nontrivial computational price. The procedure needs to generate bootstrap replicates of the pivotal statistic, $T_{\mathsf N,b_1}^*\coloneqq (\theta_{b_1}^*-\hat\theta)/\sigma_{b_1}^*$ for $b_1=1,2,\ldots,B_1$, and based on the approximation $T_{\mathsf N,b_1}^*\approx T_{\mathsf N}$ in distribution, the studentized bootstrap outputs the one-sided CIs $[\hat\theta-\hat\sigma\{T_{\mathsf N,b_1}^*\}_{\alpha}, \infty )$ and $(-\infty,\hat\theta-\hat\sigma\{T_{\mathsf N,b_1}^*\}_{1-\alpha}]$, as well as the symmetric two-sided CI
\begin{align}\label{eqsymp}
[\,\hat\theta-\hat\sigma\{\,|T_{\mathsf N,b_1}^*|\,\}_{\alpha},\hat\theta+\hat\sigma\{\,|T_{\mathsf N,b_1}^*|\,\}_{\alpha}\,].
\end{align}
Here $\hat\sigma$ is the standard error estimator based on the original data $\mathcal X$. In this pipeline, given the bootstrap datasets $\mathcal X_{b_1}^*=\{\mathbf X_{b_1,1}^*,\dots,\mathbf X_{b_1,n}^*\}$, one must compute not only $\theta_{b_1}^*$ but also an accurate estimator of $\sigma_{b_1}^*$ as a functional of the resampled empirical distribution $P_{b_1}^*$. In some simple settings, such as linear regression, this functional may admit an explicit formula. However, in general, when no efficient analytic expression is available, one needs to resort to a second layer of Monte Carlo to estimate it via the sample variance:
\begin{align}\label{eqsigmastar}
    (\sigma_{b_1}^*)^2\approx \frac{1}{B_2-1}\sum_{b_2=1}^{B_2}(\theta^{**}_{b_1,b_2}-\theta_{b_1}^*)^2.
\end{align}
Here $\theta^{**}_{b_1,b_2}$ is the estimate of the target statistic computed from the inner bootstrap dataset $\mathcal X_{b_1,b_2}^{**}=\{\mathbf X_{b_1,b_2,1}^{**},\dots,\mathbf X_{b_1,b_2,n}^{**}\}$ resampled from $\mathcal X_{b_1}^*$. To estimate the quantiles of the resampled pivotal statistics accurately, both the number of outer resamples $B_1$ and the number of inner resamples $B_2$ need to be sufficiently large. This nested Monte Carlo procedure requires a multiplicative effort of $B_1\times B_2$ model evaluations, which can be very expensive whenever there is no efficient way to calculate $\sigma^*$.

The upper-right corner of Figure \ref{fig:bootstrap_relationships} shows the higher-order coverage benefit of the studentized bootstrap, but also its multiplicative resampling effort. Moreover, suppose we view the inner resampling layer used to approximate the standard error $\sigma^*$ in the studentized bootstrap as the resampling step in the standard error bootstrap. Then the studentized bootstrap can be viewed as a higher-order coverage correction of the standard error bootstrap, by replacing the $z$-quantile in the latter with a more accurate bootstrap quantile calibration. In this sense, the studentized bootstrap is a higher-order enhancement of both the basic bootstrap and the standard error bootstrap.

\subsection{Cheap Bootstrap}
The conventional bootstrap idea, described at the beginning of this section and culminating in the use of the distributional resemblance between $\hat\theta-\theta$ and $\theta^*-\hat\theta$ depicted in \eqref{CLT original} and \eqref{CLT resample} (or their variants operating on the standard error or the pivotal statistic), is a \emph{two-stage} approach. To explain this, we first approximate the original sampling distribution by the resampling distribution, and use the summary statistics of the resample estimates, which are computable, to approximate the counterparts for the original estimate. The computation of summary statistics of the resample estimates is implemented via Monte Carlo, namely by drawing new resamples and re-evaluating the model. In this process, the approximation of the resampling distribution to the sampling distribution is separate from the Monte Carlo approximation (i.e., \eqref{CLT original} and \eqref{CLT resample} do not have Monte Carlo size explicitly involved). Because of this, we need a \emph{large} Monte Carlo size in order to ensure that the summary statistics of the resample estimates can be accurately obtained; otherwise the bootstrap pipeline would fail.

The cheap bootstrap takes a different perspective by looking at the joint distribution of the original empirical estimate and the resample estimates, in a sense providing an integrated analysis of the data noise and the Monte Carlo noise. This turns out to allow the construction of statistically valid CIs with only a very small number of Monte Carlo resamples, including even just one resample. In particular, \cite{lam2022} shows that, in the typical case where $\mathcal N$ is Gaussian, the combination of \eqref{CLT original} and \eqref{CLT resample} leads to
\begin{equation}
\sqrt n(\hat\theta-\theta,\theta_1^*-\hat\theta,\ldots,\theta_{B}^*-\hat\theta)\Rightarrow(\mathcal N_0,\ldots,\mathcal N_{B})\label{CLT joint}
\end{equation}
where $\mathcal N_0,\ldots,\mathcal N_{B}$ are i.i.d.\ Gaussian. Importantly, this holds for any fixed $B\geq 1$. With the joint convergence \eqref{CLT joint}, we can construct CIs using a pivotal statistic based on elementary Gaussian machinery. More precisely, with several resample estimates $\theta_b^*$, we construct a ``variance'' estimator $S^2 \coloneqq \sum_{b=1}^{B}(\theta_b^*-\hat\theta)^2/B$. Then the two-sided CI is $[\,\hat\theta \pm t_{B,(1+\alpha)/2}S\,]$ and the one-sided CIs are $(\hat\theta - t_{B,\alpha}S,\infty]$ and $[-\infty,\hat\theta + t_{B,\alpha}S)$, where $t_{B,\alpha}$ denotes the $\alpha$-quantile of a Student's $t$ distribution with $B$ degrees of freedom. These intervals are justified by \eqref{CLT joint}, which implies that the statistic
\begin{align}\label{eq tpivotal}
T \coloneqq \frac{\hat\theta-\theta}{S}
\end{align}
converges in distribution to a Student's $t$ distribution with $B$ degrees of freedom. In the CIs above, $B$ can be a very small number such as $1$, although \cite{lam2022} and \cite{lam2023bootstrap} suggest that it is advantageous to use $B$ above 3 to obtain a shorter interval, while on the other hand using an even larger $B$ is not necessarily beneficial, since the incremental shortening of the interval length diminishes sharply as $B$ increases.

The coverage error of the cheap bootstrap is $O(n^{-1})$ in the two-sided case and $O(n^{-1/2})$ in the one-sided case, matching those of the corresponding CIs produced by the basic bootstrap. Moreover, the cheap bootstrap can be viewed as a ``computational relaxation'' of the standard error bootstrap by using a fixed, small number of resamples $B$ to construct the standard error estimate, and correspondingly using a $t$-quantile instead of a $z$-quantile, with the degree of freedom of the $t$ being the number of resamples---a measure of the computational effort rather than the data size in conventional statistical thinking. In particular, when $B\to\infty$, the standard error estimate in the cheap bootstrap CI consistently converges to $\sqrt{\mathrm{Var}(\hat\theta)}$, and the degree of freedom in the $t$ goes to $\infty$, thus recovering the standard normal distribution, so that the resulting CI reduces to the standard error bootstrap. Nonetheless, the cheap bootstrap opens up the possibility of using far fewer resamples than the standard error bootstrap while retaining the same level of coverage accuracy. The lower-right corner of Figure \ref{fig:bootstrap_relationships} shows the statistical and relational properties of the cheap bootstrap.

\section{Studentized Cheap Bootstrap: Procedure and Main Theory}\label{sec scb procedure}
With the existing bootstrap methods and their relations discussed, we now introduce our main approach, SCB, and position it within the bootstrap framework. SCB aims to get the ``best of both worlds'' by producing higher-order coverage accuracy while at the same time avoiding a multiplicative computational effort from nested Monte Carlo resampling. The high-level idea is to integrate the cheap bootstrap with the studentized bootstrap: instead of using the $t$-quantile driven by the limit of the pivotal statistic $T$ in the cheap bootstrap, we bootstrap-calibrate the quantile of $T$. Note that the cheap bootstrap pivotal statistic $T$ itself involves running $B$ resample estimates to form the scale statistic $S$, so this entire procedure still requires two layers of resampling. This is analogous to the execution of the studentized bootstrap when the standard error does not admit an analytically computable form. However, we no longer seek a consistent estimate of $\sigma_{b_1}^*$ as in \eqref{eqsigmastar}, and the limiting quantile is no longer normal. Instead, we work with the limiting $t$ distribution of the cheap bootstrap, and the number of inner resamples is as small as what is needed in the cheap bootstrap. In this way, the overall resampling effort is $B\times B_1$, where $B$ is a fixed small number that does not inflate the computational requirement multiplicatively. The lower-right corner of Figure \ref{fig:bootstrap_relationships} positions our SCB.

\subsection{The SCB Procedure}
We now describe the SCB procedure. First, we resample from $\mathcal X$ to obtain the outer resample datasets
\[
\mathcal X_{b_1}^*=\{\mathbf X_{b_1,1}^*,\dots,\mathbf X_{b_1,n}^*\},
\qquad b_1=1,2,\ldots,B_1
\]
and compute the corresponding estimate \(\theta_{b_1}^*:=\theta(\mathcal X_{b_1}^*)\) for each $b_1$. Next, for each \(b_1\), we resample from \(\mathcal X_{b_1}^*\) to obtain the inner resample datasets
\[
\mathcal X_{b_1,b_2}^{**}
=
\{\mathbf X_{b_1,b_2,1}^{**},\dots,\mathbf X_{b_1,b_2,n}^{**}\},
\qquad b_2=1,2,\ldots,B
\]
and compute the corresponding estimate \(\theta_{b_1,b_2}^{**}:=\theta(\mathcal X_{b_1,b_2}^{**})\) for each $b_2$. Using all these \(\theta_{b_1,b_2}^{**}\) for $b_2=1,\ldots,B$, we compute the resampled scale estimate
\begin{align*}
S_{b_1}^*
\coloneqq
\left[
\frac{1}{B}\sum_{b_2=1}^{B}
(\theta_{b_1,b_2}^{**}-\theta_{b_1}^*)^2
\right]^{1/2}.
\end{align*}
Next, we compute the resampled pivotal statistics
\begin{align}\label{eq resampletpivotal}
T_{b_1}^* \coloneqq \frac{\theta_{b_1}^*-\hat\theta}{S_{b_1}^*},
\qquad b_1=1,\ldots,B_1.
\end{align}
We construct one more scale estimate, this time by resampling directly from the original dataset \(\mathcal X\) rather than from the outer resamples:
\begin{align*}
\hat S
=
\left[
\frac{1}{B}\sum_{k=1}^{B}
(\theta_k^*-\hat\theta)^2
\right]^{1/2}.
\end{align*}
Recall that \(\{x_{b_1}\}_{p}\) denotes the \(p\)-quantile of a finite sequence $x_{b_1}$ indexed by \(b_1\). Then SCB outputs the two-sided CI
\begin{align}\label{eqa11}
\mathcal I^{\mathrm{ts}}_{B_1}
=
\left[
\hat\theta-\{\,|T_{b_1}^*|\,\}_{\alpha}\times \hat S,\,
\hat\theta+\{\,|T_{b_1}^*|\,\}_{\alpha}\times \hat S
\right],
\end{align}
and the one-sided CIs
\begin{align}
\mathcal I^{\mathrm{low}}_{B_1}
&=
\left(-\infty,\hat\theta-\{\,T_{b_1}^*\,\}_{1-\alpha}\times \hat S\right],\qquad
\mathcal I^{\mathrm{up}}_{B_1}=\left[\hat\theta-\{\,T_{b_1}^*\,\}_{\alpha}\times \hat S,\infty\right).
\end{align}
Throughout, we suppress the dependence on the fixed inner resample size \(B\), and write only the dependence on \(B_1\).

Compared with the conventional studentized bootstrap, SCB still requires \(B_1\) to be sufficiently large, since \(B_1\) controls the accuracy of the Monte Carlo approximation to the conditional quantiles of \(T^*|\mathcal X\), where \(T^*\) denotes the resampled pivotal statistic (i.e., with the same distribution as \(T_{b_1}^*\) defined above given $\mathcal X$). The crucial point, however, is that the number of inner resamples \(B\) can be chosen as a small integer. Correspondingly, note that the scale estimator differs from the standard error estimator in the studentized bootstrap, namely \eqref{eqsigmastar}. In particular, it is centered at the outer resample estimates \(\theta_{b_1}^*\) and uses \(B\) in the denominator, exactly as in the cheap bootstrap. This distinction relates to the difference in the theories behind the statistical guarantees of the studentized bootstrap and SCB. Indeed, our numerical results in Section~\ref{num res} show that even very small values of \(B\) can perform well when coverage is the primary concern, and in practice we recommend \(B\approx 3\) to $5$, since this typically reduces the CI length to a level close to that of the non-pivotal methods. When \(B\) is treated as fixed, the total computational cost is linear in \(B_1\).

As discussed earlier, the relationship between CB and SCB is analogous to that between the basic bootstrap and the studentized bootstrap. SCB improves the coverage error of the two-sided CIs generated by CB from \(O(n^{-1})\) to \(O(n^{-2})\), while requiring substantially less computation than the classical studentized bootstrap. This improvement is formalized in Theorems~\ref{theorem err 1side} and \ref{theorem err 2side}, which are the main theoretical results of this paper and are presented in the next subsection.

\subsection{Theory on Coverage Accuracy}\label{sec: theoretical}
We analyze the coverage errors of SCB. As is customary in the bootstrap literature, we consider the function-of-mean model, namely $\theta=f(\bm\mu)$, where $\bm\mu=\EE[\mathbf X]$ for a $d$-dimensional random vector $\mathbf X$ and $f:\mathbb R^d\to\mathbb R$ is a function. Denote $\overline{\mathbf X}=(1/n)\sum_{i=1}^n\mathbf X_i$ as the sample mean of the i.i.d.\ data $\mathcal{X}=\{\mathbf X_1,\ldots,\mathbf X_n\}$. We define the smooth function
\begin{equation}
A(\mathbf x)=\frac{f(\mathbf x)-f(\bm\mu)}{h(\bm\mu)}\label{A def}
\end{equation}
for a function $h:\mathbb R^d\to\mathbb R$, where $h(\bm\mu)^2$ is the asymptotic variance of $\sqrt n f(\overline{\mathbf X})$ and can be written in terms of $\bm\mu$ under some regularity conditions. We also define the studentized version of $A$ by
\begin{equation}
A_s(\mathbf x)=\frac{f(\mathbf x)-f(\bm\mu)}{h(\mathbf x)}\label{A def2}
\end{equation}
For nominal coverage level \(\alpha\in(0,1)\), define the ideal cutoff values
\begin{align}\label{eq popquantile}
u_\alpha^{\mathrm{up}}
\coloneqq
\inf\{\mathsf q\in \R:\PP(T\le \mathsf q)\ge \alpha\},\qquad  u_\alpha^{\mathrm{low}}
\coloneqq
\sup\{\mathsf q\in \R:\PP(T\ge \mathsf q)\ge \alpha\}
\end{align}
and
\begin{align}\label{eq popquantile 2side}
u_\alpha^{\mathrm{ts}}
\coloneqq
\inf\{\mathsf q\ge 0:\PP(|T|\le \mathsf q)\ge \alpha\}.
\end{align}
These are the population quantiles associated with the law of the original pivotal statistic \(T\) defined in \eqref{eq tpivotal}. Recall that \(T^*\) denotes a generic resampled pivotal statistic having the same conditional distribution as \(T_{b_1}^*\) in \eqref{eq resampletpivotal} given \(\mathcal X\). We then define the corresponding bootstrap cutoff values
\begin{align}\label{eq empquantile}
\hat u_\alpha^{\mathrm{up}}
&\coloneqq
\inf\{\mathsf q\in \R:\PP(T^*\le \mathsf q\mid \mathcal X)\ge \alpha\},\qquad
\hat u_\alpha^{\mathrm{low}}
\coloneqq
\sup\{\mathsf q\in \R:\PP(T^*\ge \mathsf q\mid \mathcal X)\ge \alpha\}
\end{align}
and
\begin{align}\label{eq empquantile 2side}
\hat u_\alpha^{\mathrm{ts}}
\coloneqq
\inf\{\mathsf q\ge 0:\PP(|T^*|\le \mathsf q\mid \mathcal X)\ge \alpha\}.
\end{align}
Thus \(\hat u_\alpha^{\mathrm{up}}\), \(\hat u_\alpha^{\mathrm{low}}\) and \(\hat u_\alpha^{\mathrm{ts}}\) are data-dependent bootstrap analogues of \(u_\alpha^{\mathrm{up}}\), \(u_\alpha^{\mathrm{low}}\) and \(u_\alpha^{\mathrm{ts}}\). Since they are defined through the conditional law given \(\mathcal X\), they are random variables measurable with respect to the \(\sigma\)-algebra generated by \(\mathcal X\). In implementation, the conditional quantiles are approximated by the empirical quantiles \(\{T_{b_1}^*\}_{\alpha}\) and \(\{|T_{b_1}^*|\}_{\alpha}\) computed from the Monte Carlo replicates \(\{T_{b_1}^*\}_{b_1=1}^{B_1}\) used in \eqref{eqa11}. In the idealized regime \(B_1\to\infty\), these empirical quantiles converge to the corresponding conditional quantiles defined above. 

We now state the standard assumptions commonly used in bootstrap analysis for the smooth function model; see Chapter 5 of \cite{hall2013bootstrap} and Appendix \ref{appendix a}.
\begin{assumption}\label{assume1}
Consider the function-of-mean model in which $\theta=f(\bm\mu)$ for some function $f:\mathbb R^d\to\mathbb R$, where $\bm\mu=\E[\mathbf X]$ for a $d$-dimensional random vector $\mathbf X$. Assume that $f$ and $h$ each have $\nu+3$ bounded derivatives in a neighborhood of $\bm\mu$ for some $\nu\geq 3$, that $\E\|\mathbf X\|^l<\infty$ for a sufficiently large positive number $l$, and that the characteristic function $\chi$ of $\mathbf X$ satisfies Cram\'er's condition $\limsup_{\|\mathbf t\|\to\infty}|\chi(\mathbf t)|<1.$
\end{assumption}

The next two theorems constitute the main results of this paper: SCB attains the same higher-order coverage accuracy as the classical studentized bootstrap, namely order $O(n^{-1})$ for one-sided intervals and order $O(n^{-2})$ for symmetric two-sided intervals. In fact, these results will describe the coverage error order uniformly in the confidence level $\alpha$. The proofs of these theorems are given in Appendix \ref{appendix D}.

\begin{theorem}[One-Sided Coverage Error of SCB]\label{theorem err 1side}
Under Assumption~\ref{assume1}, for any $\varepsilon\in(0,1/2)$ and any fixed $B\geq1$,
\begin{align*}
\sup_{\varepsilon\le \alpha\le 1-\varepsilon}
\left|
\PP(T\le \hat u_\alpha^{\mathrm{up}})-\alpha
\right|
=
O(n^{-1})
\end{align*}
and
\begin{align*}
\sup_{\varepsilon\le \alpha\le 1-\varepsilon}
\left|
\PP(T\ge \hat u_\alpha^{\mathrm{low}})-\alpha
\right|
=
O(n^{-1}).
\end{align*}
\end{theorem}
\begin{theorem}[Two-Sided Coverage Error of SCB]\label{theorem err 2side}
Under Assumption~\ref{assume1}, for any $\varepsilon\in(0,1/2)$ and fixed $B\geq 2$,
\begin{align*}
\sup_{\varepsilon\le \alpha\le 1-\varepsilon}
\left|
\PP(\abs{T}\le \hat u_\alpha^{\mathrm{ts}})-\alpha
\right|
=
O(n^{-2}).
\end{align*}
\end{theorem}
Theorems~\ref{theorem err 1side} and \ref{theorem err 2side} are stated in terms of the conditional SCB pivotal quantiles. To connect these statements with the implemented SCB intervals, let
\(\mathcal I_{\infty}^{\mathrm{up}}\), \(\mathcal I_{\infty}^{\mathrm{low}}\), and \(\mathcal I_{\infty}^{\mathrm{ts}}\) denote the idealized limiting versions of
\(\mathcal I_{B_1}^{\mathrm{up}}\), \(\mathcal I_{B_1}^{\mathrm{low}}\), and \(\mathcal I_{B_1}^{\mathrm{ts}}\), respectively, as \(B_1\to\infty\). Then, by construction,
\begin{align*}
\{\theta\in\mathcal I_{\infty}^{\mathrm{up}}\}
=
\{T\le \hat u_\alpha^{\mathrm{up}}\},\qquad
\{\theta\in\mathcal I_{\infty}^{\mathrm{low}}\}
=
\{T\ge \hat u_\alpha^{\mathrm{low}}\},\qquad
\{\theta\in\mathcal I_{\infty}^{\mathrm{ts}}\}
=
\{\abs{T}\le \hat u_\alpha^{\mathrm{ts}}\}.
\end{align*}
Therefore, Theorems~\ref{theorem err 1side} and \ref{theorem err 2side} imply that, when the number of inner resamples $B$ is fixed as small as $1$ in the one-sided case and $2$ in the two-sided case, and with the idealized outer-resampling limit \(B_1\to\infty\), the implemented SCB intervals attain coverage error \(O(n^{-1})\) for one-sided intervals and \(O(n^{-2})\) for symmetric two-sided intervals. In other words, while the outer resamples need to be sizable, SCB attains higher-order coverage accuracies by using very few resamples in the inner layer, thus effectively reducing a nested Monte Carlo requirement in a ``bootstrap on a bootstrap" to a single layer of computation.

Compared with the basic and standard bootstraps, as well as other variants such as the so-called percentile bootstrap, SCB achieves higher-order coverage accuracy. The studentized bootstrap achieves the same accuracy order as SCB, namely $O(n^{-1})$ for one-sided CIs (and also two-sided equal-tailed, i.e., nonsymmetric, CIs), and $O(n^{-2})$ for two-sided symmetric CIs. However, it requires both $B_1\to\infty$ and $B_2\to\infty$, unless the standard error has an analytical form, and is therefore computationally much more expensive than SCB. On the other hand, SCB is expected to yield wider intervals because the scale estimators $S_{b_1}^*$ and $\hat S$ can vary substantially when $B$ is small, making extreme values more likely than in the studentized bootstrap. Nevertheless, the numerical results in Section~\ref{num res} show that this widening effect diminishes rapidly once \(B\) is increased to a modest value, such as \(B=5\). We also note that, compared with the basic bootstrap and other benchmarks that do not possess higher-order accuracy, the length of SCB intervals is typically larger, reflecting the correction of the undercoverage commonly observed in these benchmark bootstrap methods. This observation is also seen in our numerical experiments in Section~\ref{num res}.

We also compare the errors of SCB with two other lines of approaches discussed in the introduction. First, bootstrap iteration, which means iterative bootstrap schemes that recalibrate against the nominal coverage level $\alpha$, can also achieve higher-order accuracy; see Section 3.11 of \cite{hall2013bootstrap}. However, like the studentized bootstrap, it requires a computationally intensive nested resampling procedure. Second, we consider BCa which is a popular automated second-order correction method. In terms of coverage error, BCa intervals achieve order $O(n^{-1})$ under suitable regularity conditions. For one-sided intervals, this is on par with SCB. However, for two-sided intervals, the standard BCa construction is more analogous to an equal-tailed interval and consequently has coverage error of order $O(n^{-1})$ \cite{efron2003second,efron2020automatic}. To the best of our knowledge, no comparable $O(n^{-2})$ result is available for the standard BCa procedure. Note that both BCa and SCB are fully automatic and applicable to black-box statistics, but nonetheless they still require sufficient smoothness of the target statistic to perform well: SCB relies on stable estimation of $S^*$ and $\hat S$, while BCa relies on a stable jackknife estimate of the acceleration constant. Finally, in terms of computational cost, BCa requires $B_1+n$ statistic evaluations, where $n$ is the data size, while SCB requires $B_1B$, with $B_1$ sufficiently large. Thus, even with a small choice such as $B=5$, SCB is generally not cheaper than BCa. However, SCB is higher-order coverage accurate with error $O(n^{-2})$ for two-sided intervals, but BCa can only achieve $O(n^{-1})$. Indeed, in the two-sided case, SCB exploits symmetry to achieve a sharper coverage error order, while relying on a pivotal statistic that avoids expensive nested resampling effort. Our next section presents the roadmap on our main theoretical developments giving rise to these insights.

\subsection{Roadmap of the Coverage Error Proof}
We give a brief roadmap for the proof of the main coverage error results. The argument is most naturally understood backwards from the desired coverage statement. We discuss both the one-sided and symmetric two-sided cases, but use the two-sided notation as the representative case.

\medskip
\noindent
\emph{Step 1: Reduce the coverage error to a cutoff-comparison problem.}
For the symmetric two-sided case, we aim to prove
\(
|\PP(\abs{T}\le \hat u_\alpha^{\mathrm{ts}})-\alpha|
=O(n^{-2})
\),
uniformly in \(\alpha\). By the definition of \(u_\alpha^{\mathrm{ts}}\), the discretization error
\(|\PP(\abs{T}\le u_\alpha^{\mathrm{ts}})-\alpha|\)
is negligible because of the exponential decay under Cram\'{e}r's condition. Hence the main task is to show that
\[
\PP(\abs{T}\le \hat u_\alpha^{\mathrm{ts}})
-
\PP(\abs{T}\le u_\alpha^{\mathrm{ts}})
\]
is of the desired order. The one-sided case is reduced in the same way, with \(T\), \(u_\alpha^{\mathrm{up}}\), and \(\hat u_\alpha^{\mathrm{up}}\) replacing \(\abs{T}\), \(u_\alpha^{\mathrm{ts}}\), and \(\hat u_\alpha^{\mathrm{ts}}\).

\medskip
\noindent
\emph{Step 2: Insert Cornish--Fisher approximants as a bridge.}
Since the exact cutoff comparison in Step 1 is difficult to analyze directly, we insert the finite-order Cornish--Fisher approximants \(\hat u_{\nu,\alpha}^{\mathrm{ts}}\) and \(u_{\nu,\alpha}^{\mathrm{ts}}\), and bound it by
\begin{align*}
   \abs{ \PP(\abs{T}\le \hat u_\alpha^{\mathrm{ts}})
-
\PP(\abs{T}\le u_\alpha^{\mathrm{ts}})}&\leq\abs{\PP(\abs{T}\le \hat u_\alpha^{\mathrm{ts}})-\PP(\abs{T}\le \hat u_{\nu,\alpha}^{\mathrm{ts}})}\\
+&\abs{\PP(\abs{T}\le \hat u_{\nu,\alpha}^{\mathrm{ts}})-\PP(\abs{T}\le u_{\nu,\alpha}^{\mathrm{ts}})}+\abs{\PP(\abs{T}\le u_{\nu,\alpha}^{\mathrm{ts}})-\PP(\abs{T}\le u_\alpha^{\mathrm{ts}})}.
\end{align*}
The first and third terms are cutoff-approximation errors, while the middle term compares the two finite-order approximants. The one-sided proof follows the same structure, with \(T\), \(u_\alpha^{\mathrm{up}}\), \(\hat u_\alpha^{\mathrm{up}}\), \(u_{\nu,\alpha}^{\mathrm{up}}\), and \(\hat u_{\nu,\alpha}^{\mathrm{up}}\) replacing their two-sided counterparts.

This decomposition shifts the problem from comparing two exact probabilities to controlling finite-order cutoff approximations. Consequently, for the Student's \(t\) distribution, we need a uniform Cornish--Fisher expansion theory that includes both the population expansion and the resampled expansion with sample-based Edgeworth coefficients. The Cornish--Fisher expansions needed for the coverage proof are \(\nu=2\) in the one-sided case and \(\nu=3\) in the symmetric two-sided case:
\begin{align*}
u_{2,\alpha}^{\mathrm{up}}
&=
\mathsf q
+
n^{-1/2}\xi_1^{\mathrm{up}}(\mathsf q)
+
n^{-1}\xi_2^{\mathrm{up}}(\mathsf q),
\qquad
\mathsf q=\Psi_B^{-1}(\alpha),
\\
u_{3,\alpha}^{\mathrm{ts}}
&=
\mathsf q
+
n^{-1}\xi_2^{\mathrm{ts}}(\mathsf q),
\qquad
\mathsf q=\Psi_B^{-1}\left(\frac{1+\alpha}{2}\right).
\end{align*}
The functions \(\xi\), explicitly identified in \eqref{eq12}, have simple closed-form expressions, yet are not polynomials. This identification, together with its resampled analogue obtained by replacing population moments with sample moments, forms our Student's \(t\)-based Cornish--Fisher expansion theory, summarized in Theorems~\ref{uniform cf expansion} and \ref{uniform resample cf expansion}. 

\medskip
\noindent
\emph{Step 3: Build the Cornish--Fisher theory from uniform Edgeworth expansions.}
The Cornish--Fisher approximants in Step 2 are obtained by inverting the corresponding Edgeworth expansions. Therefore, establishing uniform Edgeworth expansions around the Student's \(t\) limit is a necessary preliminary step for deriving the cutoff expansions. In fact, we prove
\begin{align}
\PP(T\le \mathsf q)
&=
\Psi_B(\mathsf q)
+
n^{-1/2}\zeta_1^{\mathrm{up}}(\mathsf q)
+
n^{-1}\zeta_2^{\mathrm{up}}(\mathsf q)
+
O(n^{-3/2}),\notag
\\
\PP(\abs{T}\le \mathsf q)
&=
2\Psi_B(\mathsf q)-1
+
n^{-1}\zeta_2^{\mathrm{ts}}(\mathsf q)
+
O(n^{-2}),\label{eq no15}
\end{align}
where the correction terms \(\zeta_1^{\mathrm{up}}\), \(\zeta_2^{\mathrm{up}}\), and \(\zeta_2^{\mathrm{ts}}\) are given in closed form in \eqref{eq zeta2} and \eqref{eq zeta2 upper}. The same type of expansion is also needed for the resampled statistic \(T^*|\mathcal X\). The corresponding population and resampled results are Theorems~\ref{theorem uniform version} and \ref{theorem resample uniform version}. 

\medskip
\noindent
\emph{Step 4: Establish the cancellation in the symmetric two-sided case.}
For the symmetric two-sided case, a subtle issue is to justify the absence of the \(O(n^{-3/2})\) term in \eqref{eq no15}. This cancellation is not merely a formal consequence of using a symmetric interval. It requires a joint treatment of the pivotal statistic and the sample-based estimators of the Edgeworth coefficients. This is the indispensable step behind the sharper \(O(n^{-2})\) coverage error. A key technical ingredient for this analysis is Proposition~\ref{theo:expect-cancellation-good}, which establishes the required expectation-level cancellation.

\medskip

We next discuss the technical novelty of our proof. At a high level, our argument follows the guiding idea as the pioneering work \cite{hall1988symmetric}. Both analyses view $\hat u_\alpha^{\mathrm{ts}}$ as a perturbation of $u_\alpha^{\mathrm{ts}}$ and study the error expansion of
\begin{align*}
    \PP\left(\abs{T}\le u_\alpha^{\mathrm{ts}}+\Delta\right),
\end{align*}
where $\Delta\coloneqq \hat u_\alpha^{\mathrm{ts}}-u_\alpha^{\mathrm{ts}}$ is further approximated through finite-term Cornish--Fisher expansions. In both studies, the main technical work is to show that the $n^{-3/2}$ term is absent from this expansion. In \cite{hall1988symmetric}, this absence follows from the oddness of the function $r_1(x)$ in Eq.~(3.2). In our setting, the analogous cancellation is established in Proposition~\ref{theo:expect-cancellation-good}.

However, our proof method is quite different from the classical counterpart. Most significantly, since we use a Student's $t$ pivotal rather than a Gaussian one, the resulting correction terms are no longer polynomials. Instead, they are finite linear combinations of Student's $t$-type algebraic functions of the form
\begin{align*}
    \frac{x^k}{(B+x^2)^{\ell/2}},
\end{align*}
where $k$ and $\ell$ are positive integers and $B$ is the number of inner resamples. We call these functions \emph{quasi-polynomial} analogues of the classical Edgeworth polynomials. Deriving these explicit quasi-polynomial correction functions and the corresponding Cornish--Fisher expansions is one of the main technical novelties of our proof.

Moreover, to prove the absence of the $n^{-3/2}$ term, \cite{hall1988symmetric} analyzes the higher-order cumulants of the perturbed pivots $T\pm\Delta$ and then derives the corresponding perturbed Edgeworth expansion through characteristic functions and Fourier--Stieltjes inversion. This argument relies on the classical Edgeworth expansion theory for Gaussian pivotal statistics \cite{BhattacharyaGhosh1978,bhattacharya2010normal}. By contrast, our SCB analysis works in a fixed-inner-resampling regime, where the quantile perturbation must be analyzed through a high-dimensional, specifically $B$-dimensional, integral. The randomness of the perturbation affects not only the coefficients of the correction functions, but also the region of integration, and therefore requires a joint analysis of both effects. This leads to the new cancellation result in Proposition~\ref{theo:expect-cancellation-good}, together with several supporting lemmas.

Finally, in the large-$B$ limit, our expansions are consistent with the classical Gaussian-pivotal picture: $\Psi_B$ converges to $\Phi$, and the correction functions $\zeta$ and $\eta$ converge to their classical Edgeworth and Cornish--Fisher counterparts, respectively. In this sense, our fixed-$B$ framework can be viewed as a finite-inner-resampling analogue of Hall's Gaussian-pivotal analysis, with the classical picture emerging heuristically as $B$ becomes large.

\section{Developments of Main Results}\label{sec4 mainroad}
We now develop the main technical ingredients for proving the coverage error results of SCB. First, we establish uniform Edgeworth expansions for the Student's \(t\)-type pivotal statistics \(T\) and \(T^*\). Next, we derive the corresponding Cornish--Fisher expansions for their quantiles. These two ingredients allow us to compare the ideal cutoff values with their bootstrap counterparts and to control the resulting quantile errors uniformly in \(\alpha\). We then use these quantile approximations to prove the one-sided coverage result. For the symmetric two-sided case, the same strategy forms the basis of the argument, together with additional cancellation coming from the symmetry structure of the expansion.

\subsection{Uniform Edgeworth Expansion for t-Distributions}
In this and the next subsection, we present the main technical ingredients that support our major Theorems \ref{theorem err 1side} and \ref{theorem err 2side}. Throughout the paper, $\Phi$ and $\phi$ denote the cdf and pdf of the standard normal distribution, respectively, while $\Psi_B$ and $\psi_B$ denote the cdf and pdf of the Student's $t$ distribution with $B$ degrees of freedom, respectively. 

Under suitable regularity conditions (see, e.g., \cite{hall2013bootstrap}), or more simply under Assumption~\ref{assume1}, both $\sqrt{n}A(\bar{\mathbf X})$ and $\sqrt{n}A_s(\bar{\mathbf X})$ admit Edgeworth expansions of the form
\[
\mathbb{P}\bigl(\sqrt{n}\tilde A(\bar{\mathbf X}) \le \mathsf q\bigr)
=
\Phi(\mathsf q)
+
\sum_{k=1}^{\nu} n^{-k/2}\pi_k(\mathsf q)\phi(\mathsf q)
+
O(n^{-(\nu+1)/2}),
\]
uniformly in $\mathsf q\in\mathbb{R}$, where each $\pi_k$ is an Edgeworth polynomial. Here $\Phi$ and $\phi$ denote the cumulative distribution function and probability density function of the standard normal distribution, respectively. By convention, we write $\pi_k=p_k$ when $\tilde A=A$, and $\pi_k=q_k$ when $\tilde A=A_s$. Moreover, it is known that $q_1$, $p_2$, and $q_2$ admit the representations
\begin{align}\label{Eq edgepoly}
    q_1(x) = \sum_{k=0,2} b_{[1,k]} x^k,\qquad
    p_2(x) = \sum_{k=1,3,5} a_{[2,k]} x^k,\qquad
    q_2(x) = \sum_{k=1,3,5} b_{[2,k]} x^k.
\end{align}
Equation~\eqref{Eq edgepoly} serves as the definition of the coefficients $b_{[1,k]}$, $b_{[2,k]}$, and $a_{[2,k]}$. The coefficients $a_{[2,k]}$ and $b_{[2,k]}$ are polynomial functions of moments and mixed moments up to order $4$, associated with the smooth function models $A(\mathbf X)$ and $A_s(\mathbf X)$, respectively. Similarly, the coefficients $b_{[1,k]}$ are polynomial functions of moments and mixed moments up to order $3$ associated with $A_s(\mathbf X)$. We state these coefficients explicitly because they also appear in the quasi-polynomial terms, namely the $\zeta_k$ functions in Theorem~\ref{theorem uniform version}, arising in the Edgeworth expansions for the Student's $t$-type statistic.

Precise statements of this Edgeworth expansion result, together with its bootstrap and resampled counterparts, are given in Appendix~\ref{appendix a}. Although these three theorems, adapted from \cite{hall2013bootstrap}, do not appear explicitly in the statements of our main results, they play a foundational role in the proofs and will be invoked repeatedly throughout the appendix.

After reviewing the usual Edgeworth expansion for Gaussian-limit pivotal statistics, we now present the corresponding expansion for Student's $t$-type statistics in Theorem~\ref{theorem uniform version}. The detailed proof for the symmetric two-sided case is given in Appendix~\ref{seca1}, while that for the one-sided case is given in Appendix~\ref{sec theorem1B2}.

\begin{theorem}[Uniform Higher-order Coverage Errors for Cheap Bootstrap]\label{theorem uniform version}
Under Assumption~\ref{assume1},
\begin{itemize}
    \item When $B\geq 2$, for any $\varepsilon_0>0$, there exists a constant $\mathsf C_{\mathrm{ts}}$ such that
\begin{align}\label{Eq cheap boots uniform convergence}
    \sup_{\mathsf{q}\geq \varepsilon_0}\abs{\mathbb P(\abs{T}\leq \mathsf{q}) - 2\Psi_B(\mathsf{q})+1-\frac{\zeta^\mathrm{ts}_2(\mathsf{q})}{n} }\leq \frac{1}{n^2}\times \mathsf C_{\mathrm{ts}}.
\end{align}
where 
\begin{align}
\zeta_2^\mathrm{ts}(\mathsf q)
\coloneqq
\sum_{k\in\{1,3,5\}}
\Biggl[
&
b_{[2,k]}\,
\frac{
2^{(k+1)/2}B^{B/2}\Gamma\!\left(\frac{B+k}{2}\right)
}{
\sqrt{\pi}\Gamma(B/2)
}
\cdot
\frac{
\mathsf q^k
}{
(B+\mathsf q^2)^{(B+k)/2}
}
\notag\\
&\quad
-
a_{[2,k]}\,
\frac{
2^{(k+1)/2}B^{(B+k+1)/2}
\Gamma\!\left(\frac{k}{2}+1\right)
\Gamma\!\left(\frac{B+k}{2}\right)
}{
\pi\Gamma\!\left(\frac{B+k+1}{2}\right)
}
\cdot
\frac{
\mathsf q
}{
(B+\mathsf q^2)^{(B+k)/2}
}
\Biggr].
\label{eq zeta2}
\end{align}

\item When $B\geq 1$, there exists a constant $\mathsf C_{\mathrm{up}}$ such that 
\begin{align}
    \sup_{\mathsf{q}\in\mathbb{R}}\abs{\mathbb P(T\leq \mathsf{q})-\Psi_B(\mathsf{q})-\frac{\zeta_{1}^{\mathrm{up}}(\mathsf{q})}{\sqrt n}-\frac{\zeta_{2}^\mathrm{up}(\mathsf{q})}{n}}\leq \frac{1}{n^{3/2}}\times \mathsf C_{\mathrm{up}},
\end{align}
where
\begin{equation}\label{eq zeta2 upper}
\begin{split}
\zeta_{1}^{\mathrm{up}}(\mathsf{q})
&\coloneqq
\sum_{k=0,2}b_{[1,k]}
\frac{
B^{B/2+k}
}{
\sqrt{2\pi}
}
\cdot
\frac{
\mathsf{q}^k
}{
(B+\mathsf{q}^2)^{B/2+k}
},
\\
\zeta_{2}^\mathrm{up}(\mathsf{q})
&\coloneqq
\frac{1}{2}\times \zeta_{2}^\mathrm{ts}(\mathsf{q}).
\end{split}
\end{equation}
\end{itemize}
\end{theorem}

Theorem~\ref{theorem uniform version} has three notable features. First, the remainder bounds hold uniformly in the argument $\mathsf q$: over all $\mathsf q\in\mathbb R$ in the one-sided case, and over all $\mathsf q\ge \varepsilon_0$ in the symmetric two-sided case. This uniformity is essential for the coverage error analysis in Theorems~\ref{theorem err 1side} and \ref{theorem err 2side}, since those arguments require control of the distributional approximation at random or varying quantile locations rather than at a single fixed $\mathsf q$.

Second, in the symmetric two-sided case, the theorem gives the sharper approximation error of order $O(n^{-2})$. This order should be distinguished from the more standard $O(n^{-3/2})$ bound, which can be obtained by a comparatively direct higher-order expansion argument and does not require excluding the case $B=1$. The improvement from $O(n^{-3/2})$ to $O(n^{-2})$ is not a formal consequence of symmetrization alone. Rather, it relies on the fact that the entire $O(n^{-3/2})$ contribution cancels in the symmetric two-sided expansion. Establishing this cancellation is the technically delicate part of the analysis, but it is indispensable for the $O(n^{-2})$ coverage-error result. It requires treating jointly the pivotal statistic $\sqrt n\,\tilde A(\bar{\mathbf X})$ and the sample-based estimators of the Edgeworth coefficients, namely $\hat a_{[2,k]}$ and $\hat b_{[2,k]}$ for $k=1,3,5$. The key technical ingredient is the following expectation-level cancellation result, whose proof is given in Appendix~\ref{secprop1}.

\begin{proposition}[Expectation Cancellation of Monomial Moment Errors under an Even Weight]\label{theo:expect-cancellation-good}
Under Assumption~\ref{assume1}, with the moment order chosen sufficiently large, let $W_n=\sqrt n \tilde A(\bar{\mathbf X}_n)$ be the pivotal statistic. Let \(M\) be a given monomial in the marginal and mixed moments of \(\mathbf X_1\), and let \(M_n\) be its sample plug-in estimator. The precise definition of \(M\) and \(M_n\) is given in Section~\ref{secProp1lemma1}, in particular~\eqref{eq MMM}. Let \(g\in C^2(\mathbb R)\) be an \emph{even} function such that \(g''\in L^1(\mathbb R)\). Suppose further that \(g\), \(g'\), and \(g''\) are uniformly bounded over \(\mathbb R\). Then
\begin{align}\label{eq cancel baby version}
\mathbb E\left[g(W_n)(M_n-M)\right]=O\left(\frac{1}{n}\right).
\end{align}
\end{proposition}
Na\"ively, one would expect the quantity inside the expectation in~\eqref{eq cancel baby version} to be of order $O(n^{-1/2})$, due to the CLT-scale fluctuation of $\sqrt n(M_n-M)$. The sharper $O(n^{-1})$ bound instead arises from a nontrivial cancellation at the expectation level induced by the evenness of $g$. Except for the evenness condition, the smoothness and boundedness assumptions imposed on $g$ are much stronger than necessary, but they are sufficient for deriving our main theorem. Example~\ref{exeven} below provides some intuition for this phenomenon and illustrates concretely that evenness is essential.
\begin{exam}\label{exeven}
Consider the simplest case where $X_i\in\mathbb R$ are iid with mean $\mu$ and unit variance, $M$ is the true mean $\mu$, $M_n=\bar{X}$, and $\tilde A(x)=x-\mu$. Let $\kappa\coloneqq \mathbb E[(X_1-\mu)^3]$ denote the third central moment. We consider the function $g_k(x)=x^k$ for positive integers $k$. Then
\[
\mathbb E\left[g_k\left(\sqrt n\tilde A(\bar X)\right)(\bar X-\mu)\right]
=
n^{k/2}\mathbb E[(\bar X-\mu)^{k+1}].
\]
It is easy to verify that
\[
n^{k/2}\mathbb E[(\bar X-\mu)^{k+1}]
=
\begin{cases}
n^{-1/2}, & k=1,\\[0.3em]
n^{-1}\kappa, & k=2,\\[0.3em]
3n^{-1/2}+n^{-3/2}\bigl(\mathbb E[(X_1-\mu)^4]-3\bigr), & k=3.
\end{cases}
\]
More generally,
\[
n^{k/2}\mathbb E[(\bar X-\mu)^{k+1}]
=
\begin{cases}
\displaystyle
\frac{1}{n}
\frac{(2r+1)!}{3!2^{r-1}(r-1)!}
\kappa
+
O(n^{-2}),
& k=2r,\\[1.2em]
\displaystyle
\frac{(2r+1)!!}{\sqrt n}
+
O(n^{-3/2}),
& k=2r+1.
\end{cases}
\]
Thus, when $k$ is even, the expectation is of order $O(n^{-1})$, whereas when $k$ is odd, it is generally only of order $O(n^{-1/2})$. This calculation illustrates the cancellation mechanism induced by evenness.
\end{exam}

A somewhat peculiar feature of Theorem~\ref{theorem err 2side} is the technical restriction $B\geq 2$, which excludes the case $B=1$. This restriction arises from the regularity requirements on the weight functions in order to apply Proposition~\ref{theo:expect-cancellation-good}. In particular, the relevant functions include terms of the form
\begin{align}
   g_{[B,k]}(x)\coloneqq
   \abs{\frac{x}{\mathsf q}}^{B-1+k}
   \exp\left(-\frac{B x^2}{2\mathsf q^2}\right),
   \qquad k\in\{1,3,5\}.
\end{align}
When $B=1$ and $k=1$, the function $g_{[1,1]}$ has a kink at $x=0$ and is not differentiable there. Hence $g_{[1,1]}\notin C^2(\mathbb R)$. On the other hand, Proposition~\ref{theo:expect-cancellation-good} is proved under the condition $g\in C^2(\mathbb R)$, since its proof uses a first-order Taylor expansion with a Lagrange remainder. We believe that this smoothness requirement, and hence the additional restriction $B\geq 2$, is not essential and can be removed through an alternative proof.

Third, the correction terms $\zeta_k^{\mathrm{ts}}$ and $\zeta_k^{\mathrm{up}}$ admit explicit closed-form expressions. By contrast, in the earlier cheap-bootstrap development of \cite{lam2022}, the corresponding $\zeta_k$ functions were represented as high-dimensional integrals. This makes the present expansions substantially more transparent. The simplification is possible because the integral structure naturally gives rise to $\chi^2$ random variables through sums of squared Gaussian coordinates, allowing the high-dimensional integrals to be evaluated in closed form.

Furthermore, the closed-form expressions of $\zeta_k^{\mathrm{ts}}$ and $\zeta_k^{\mathrm{up}}$ can be simplified in the limit $B\to\infty$.

\begin{corollary}[Large-$B$ limits of $\zeta_k^{\mathrm{ts}}$ and $\zeta_k^{\mathrm{up}}$] The following limits hold:
\begin{align*}
\lim_{B\to\infty}\zeta_{1}^{\mathrm{up}}(\mathsf q)
&=
\phi(\mathsf q)
\left(
b_{[1,0]}+b_{[1,2]}\mathsf q^2
\right),\\
\lim_{B\to\infty}\zeta_2^{\mathrm{ts}}(\mathsf q)
&=
2\phi(\mathsf q)
\left[
b_{[2,1]}\mathsf q
+
b_{[2,3]}\mathsf q^3
+
b_{[2,5]}\mathsf q^5
-
\left(
a_{[2,1]}+3a_{[2,3]}+15a_{[2,5]}
\right)\mathsf q
\right].
\end{align*}
\end{corollary}
It is known that the cheap bootstrap reduces to the standard error bootstrap as $B\to\infty$. The explicit limiting formulas above may help clarify the relationship between Edgeworth expansions and the standard error bootstrap. To the best of our knowledge, this connection has not been established in the existing literature.

We next state the corresponding uniform Edgeworth expansion for the resampled cheap-bootstrap statistic $T^*$. This result is the conditional analogue of Theorem~\ref{theorem uniform version}: the deterministic coefficients in the population expansion are replaced by their sample-based counterparts, and the remainder bound holds with high probability. The proof is given in Appendix~\ref{sec-proof-resample-ts} for the symmetric two-sided case and in Appendix~\ref{resample-one-side} for the one-sided case.

\begin{theorem}[Uniform Higher-order Coverage Errors for the Resampled Cheap Bootstrap]\label{theorem resample uniform version}
Under Assumption~\ref{assume1},
\begin{itemize}
    \item When $B\geq 2$, for any $\varepsilon_0>0$, there exists a constant $\hat{\mathsf C}_{\mathrm{ts}}$ such that
    \begin{align}\label{Eq resample ts uniform convergence}
        \mathbb{P}\left(
        \sup_{\mathsf q\geq \varepsilon_0}
        \abs{
        \mathbb P(\abs{T^*}\leq \mathsf q\mid \mathcal X)
        -2\Psi_B(\mathsf q)+1
        -\frac{\hat{\zeta}_2^\mathrm{ts}(\mathsf q)}{n}
        }
        \geq \frac{1}{n^2}\times \hat{\mathsf C}_{\mathrm{ts}}
        \right)
        =
        O(n^{-\lambda})
    \end{align}
    where $\hat{\zeta}^\mathrm{ts}_2$ is defined in the same way as $\zeta^\mathrm{ts}_2$ in \eqref{eq zeta2}, except that the coefficients $a_{[2,k]}$ and $b_{[2,k]}$ are replaced by $\hat a_{[2,k]}$ and $\hat b_{[2,k]}$, respectively.

    \item When $B\geq 1$, there exists a constant $\hat{\mathsf C}_{\mathrm{up}}$ such that
    \begin{align}\label{Eq resample up uniform convergence}
        \mathbb{P}\left(
        \sup_{\mathsf q\in\mathbb{R}}
        \abs{
        \mathbb P(T^*\leq \mathsf q\mid \mathcal X)
        -\Psi_B(\mathsf q)
        -\frac{\hat{\zeta}_{1}^{\mathrm{up}}(\mathsf q)}{\sqrt n}
        -\frac{\hat{\zeta}_{2}^{\mathrm{up}}(\mathsf q)}{n}
        }
        \geq \frac{1}{n^{3/2}}\times \hat{\mathsf C}_{\mathrm{up}}
        \right)
        =
        O(n^{-\lambda})
    \end{align}
    where $\hat{\zeta}_{1}^{\mathrm{up}}$ and $\hat{\zeta}_{2}^{\mathrm{up}}$ are defined in the same way as $\zeta_{1}^{\mathrm{up}}$ and $\zeta_{2}^{\mathrm{up}}$ in \eqref{eq zeta2 upper}, except that all coefficients $a_{[2,k]}$ and $b_{[i,k]}$ are replaced by $\hat a_{[2,k]}$ and $\hat b_{[i,k]}$ for $i=1,2$, respectively.
\end{itemize}
\end{theorem}
Recall that $\psi_B$ is the pdf of the Student's $t$ distribution with $B$ degrees of freedom. For the expansion coefficients in Theorem~\ref{theorem uniform version}, define
\begin{align*}
    \eta_2^{\mathrm{ts}}(\mathsf q)
    \coloneqq
    \frac{\zeta_2^{\mathrm{ts}}(\mathsf q)}{\psi_B(\mathsf q)},
    \qquad
    \eta_k^{\mathrm{up}}(\mathsf q)
    \coloneqq
    \frac{\zeta_k^{\mathrm{up}}(\mathsf q)}{\psi_B(\mathsf q)},
    \quad k=1,2.
\end{align*}
Then the expansions in Theorem~\ref{theorem uniform version} can be rewritten as
\begin{equation}\label{eq11}
\begin{split}
G_{n}^{\mathrm{up}}(\mathsf q)
&\coloneqq 
\PsiB(\mathsf q)
+
n^{-1/2}\eta_1^{\mathrm{up}}(\mathsf q)\psiB(\mathsf q)
+
n^{-1}\eta_2^{\mathrm{up}}(\mathsf q)\psiB(\mathsf q),
\\
G_{n}^{\mathrm{ts}}(\mathsf q)
&\coloneqq
2\PsiB(\mathsf q)-1
+
\frac{1}{n}
\eta_2^{\mathrm{ts}}(\mathsf q)\psiB(\mathsf q).
\end{split}
\end{equation}
This reformulation is more convenient for developing the corresponding Cornish--Fisher expansions in the next section. The functions $\eta_k$ are given explicitly by
\begin{align*}
\eta_2^{\mathrm{ts}}(\mathsf q)
=
&
\,2(b_{[2,1]}-a_{[2,1]})\mathsf q
-a_{[2,3]}\frac{6B(B+1)}{B+2}\frac{\mathsf q}{B+\mathsf q^2}
+2b_{[2,3]}(B+1)\frac{\mathsf q^3}{B+\mathsf q^2}
\\
&
-a_{[2,5]}\frac{30B^2(B+1)(B+3)}{(B+2)(B+4)}\frac{\mathsf q}{(B+\mathsf q^2)^2}
+2b_{[2,5]}(B+1)(B+3)\frac{\mathsf q^5}{(B+\mathsf q^2)^2},
\\
\eta_1^{\mathrm{up}}(\mathsf q)
=
&
\frac{\Gamma(\frac{B}{2})}{\sqrt{2}\,\Gamma\!\left(\frac{B+1}{2}\right)}
\left[
b_{[1,0]}(B+\mathsf q^2)^{1/2}
+
b_{[1,2]}\frac{B^2\mathsf q^2}{(B+\mathsf q^2)^{3/2}}
\right],
\qquad
\eta_2^{\mathrm{up}}(\mathsf q)
=
\frac{1}{2}\eta_2^{\mathrm{ts}}(\mathsf q).
\end{align*}
Although their explicit forms are complicated, the functions $\eta_k$ have two useful properties: they are infinitely differentiable on $\mathbb R$, and they follow the same odd--even parity pattern as classical Edgeworth polynomials. More precisely, $\eta_k$ is odd when $k$ is even and even when $k$ is odd. 

We therefore refer to them as Edgeworth quasi-polynomials for the Student's $t$ distribution; this structure motivates the general quasi-polynomial framework developed in Appendix~\ref{conjectural framework}, where the expansions are denoted by $G_{n,\nu}^{\mathrm{up}}(\mathsf q)$ and $G_{n,\nu}^{\mathrm{ts}}(\mathsf q)$. The expansions established here correspond to the one-sided case with $\nu=2$ and the symmetric two-sided case with $\nu=3$, as displayed in \eqref{eq11}.

\subsection{Uniform Cornish--Fisher Expansion for t-Distributions}

We now pass from the Edgeworth expansions to their Cornish--Fisher counterparts. We first introduce the population quantiles and their formal approximants, and then state the corresponding existence and uniqueness result. 
We expect these quantiles to admit the following finite Cornish--Fisher expansions for the Student's $t$-type statistic:
\begin{align*}
u_{2,\alpha}^{\mathrm{up}}
&\coloneqq
\mathsf q+\sum_{k=1}^{2} n^{-k/2}\xi_k^{\mathrm{up}}(\mathsf q),
\quad \text{where }\mathsf q\coloneqq\PsiB^{-1}(\alpha),
\\
u_{3,\alpha}^{\mathrm{ts}}
&\coloneqq
\mathsf q+\frac{1}{n}\xi_2^{\mathrm{ts}}(\mathsf q),
\quad \text{where }\mathsf q\coloneqq\PsiB^{-1}\!\left(\frac{1+\alpha}{2}\right).
\end{align*}
with
\begin{equation}\label{eq12}
\begin{split}
\xi_1^{\mathrm{up}}(\mathsf q)
&=
-\eta_1^{\mathrm{up}}(\mathsf q),
\qquad
\xi_2^{\mathrm{ts}}(\mathsf q)
=
-\frac{1}{2}\eta_2^{\mathrm{ts}}(\mathsf q),
\\
\xi_2^{\mathrm{up}}(\mathsf q)
&=
-\eta_2^{\mathrm{up}}(\mathsf q)
+\eta_1^{\mathrm{up}}(\mathsf q)\bigl(\eta_1^{\mathrm{up}}\bigr)'(\mathsf q)
-\frac{(B+1)\mathsf q}{2(B+\mathsf q^2)}\bigl(\eta_1^{\mathrm{up}}(\mathsf q)\bigr)^2.
\end{split}
\end{equation}
Here each $\xi_k$, although not a polynomial, belongs to $C^\infty(\mathbb R)$, is even when $k$ is odd, and odd when $k$ is even. Moreover, the coefficients of $\xi_k$ are polynomials in moments up to order $k+2$. The subscripts $\nu=2,3$ in $u_{2,\alpha}^{\mathrm{up}}$ and $u_{3,\alpha}^{\mathrm{ts}}$ are consistent with the approximation order $O(n^{-(\nu+1)/2})$.

Proposition \ref{thpp1} below establishes the existence and uniqueness of these formal Cornish--Fisher approximants, together with their uniform convergence properties. A general version of Proposition \ref{thpp1}, namely Proposition \ref{prop:CF-existence}, is stated in Appendix \ref{conjectural framework} under the general quasi-polynomial framework. Its proof is given in Appendix \ref{sec51} and Appendix  \ref{sec52}.
\begin{proposition}[Existence and uniqueness of the Cornish--Fisher expansions on the Student's $t$ distribution]\label{thpp1}
Under Assumption~\ref{assume1}, for any fixed $B\geq1$ in the one-sided statement and any fixed $B\geq2$ in the two-sided statement, there exist unique functions $\xi_1^{\mathrm{up}},\xi_2^{\mathrm{up}},\xi_2^{\mathrm{ts}}$, given in \eqref{eq12}, such that for every compact interval $K\subset\mathbb R$ and $\tilde K\subset(0,\infty)$, there exist finite constants $\mathsf A_K^{\mathrm{up}}$ and $\mathsf A_{\tilde K}^{\mathrm{ts}}$ satisfying
\begin{align*}
\sup_{\mathsf q\in K}
\left|
G_n^{\mathrm{up}}\!\left(\mathsf q+\sum_{k=1}^{2} n^{-k/2}\xi_k^{\mathrm{up}}(\mathsf q)\right)-\PsiB(\mathsf q)
\right|
&\le
\mathsf A_K^{\mathrm{up}} n^{-3/2},
\\
\sup_{\mathsf q\in \tilde K}
\left|
G_n^{\mathrm{ts}}\!\left(\mathsf q+\frac{1}{n}\xi_2^{\mathrm{ts}}(\mathsf q)\right)-\bigl(2\PsiB(\mathsf q)-1\bigr)
\right|
&\le
\mathsf A_{\tilde K}^{\mathrm{ts}} n^{-2}.
\end{align*}
\end{proposition}

Motivated by the expansion forms of $u_{2,\alpha}^{\mathrm{up}}$ and $u_{3,\alpha}^{\mathrm{ts}}$, we define the formal resampled Cornish--Fisher approximants by
\begin{align*}
\hat u_{2,\alpha}^{\mathrm{up}}
&\coloneqq
\mathsf q+\sum_{k=1}^{2} n^{-k/2}\hat\xi_k^{\mathrm{up}}(\mathsf q),
\quad \text{where }\mathsf q\coloneqq\PsiB^{-1}(\alpha),
\\
\hat u_{3,\alpha}^{\mathrm{ts}}
&\coloneqq
\mathsf q+\frac{1}{n}\hat\xi_2^{\mathrm{ts}}(\mathsf q),
\quad \text{where }\mathsf q\coloneqq\PsiB^{-1}\!\left(\frac{1+\alpha}{2}\right).
\end{align*}
Here $\hat\xi_1^{\mathrm{up}},\hat\xi_2^{\mathrm{up}}$, and $\hat\xi_2^{\mathrm{ts}}$ are obtained from \eqref{eq12} by replacing the population moments in the coefficients with the corresponding sample moments. In this way, we have defined the four quantile approximants needed for the coverage error analysis, namely
\[
u_{2,\alpha}^{\mathrm{up}},\qquad
u_{3,\alpha}^{\mathrm{ts}},\qquad
\hat u_{2,\alpha}^{\mathrm{up}},\qquad
\hat u_{3,\alpha}^{\mathrm{ts}}.
\]

We now explain their roles. Recall the exact quantiles $u_\alpha^{\mathrm{up}}$ and $u_\alpha^{\mathrm{ts}}$ of the finite-sample distribution of $T$, defined in \eqref{eq popquantile} and \eqref{eq popquantile 2side}, and the resampled quantiles $\hat u_\alpha^{\mathrm{up}}$ and $\hat u_\alpha^{\mathrm{ts}}$, defined in \eqref{eq empquantile} and \eqref{eq empquantile 2side}. The approximants $u_{2,\alpha}^{\mathrm{up}},u_{3,\alpha}^{\mathrm{ts}},\hat u_{2,\alpha}^{\mathrm{up}},\hat u_{3,\alpha}^{\mathrm{ts}}$ serve as explicit higher-order proxies for these exact quantiles. Their main advantage is that they are given by finite sums of smooth $\xi_k$ functions, which makes it much easier to control differences such as $|u_{2,\alpha}^{\mathrm{up}}-\hat u_{2,\alpha}^{\mathrm{up}}|$ and $|u_{3,\alpha}^{\mathrm{ts}}-\hat u_{3,\alpha}^{\mathrm{ts}}|$ than if one worked directly with the original quantile definitions. At the same time, these approximants remain sufficiently accurate: as the next theorems show, $u_{2,\alpha}^{\mathrm{up}}$ and $u_{3,\alpha}^{\mathrm{ts}}$ are already higher-order approximations to $u_\alpha^{\mathrm{up}}$ and $u_\alpha^{\mathrm{ts}}$, while $\hat u_{2,\alpha}^{\mathrm{up}}$ and $\hat u_{3,\alpha}^{\mathrm{ts}}$ are the corresponding higher-order approximations to $\hat u_\alpha^{\mathrm{up}}$ and $\hat u_\alpha^{\mathrm{ts}}$. Therefore, these four explicit approximants provide the convenient intermediate objects through which the final coverage error bounds are established.
\begin{theorem}[Uniform error for the Cornish--Fisher expansions]\label{uniform cf expansion}
Let $\varepsilon\in(0,1/2)$. Under Assumption~\ref{assume1}, for any fixed $B\geq1$ in the one-sided statement and any fixed $B\geq2$ in the two-sided statement, there exist finite constants $\mathsf C_{\mathrm{c.f.}}^{\mathrm{up}}$ and $\mathsf C_{\mathrm{c.f.}}^{\mathrm{ts}}$ such that
\begin{align*}
\sup_{\varepsilon\le\alpha\le1-\varepsilon}
\left|u_{2,\alpha}^{\mathrm{up}}-u_\alpha^{\mathrm{up}}\right|
\le
\mathsf C_{\mathrm{c.f.}}^{\mathrm{up}} n^{-3/2},
\qquad
\sup_{\varepsilon\le\alpha\le1-\varepsilon}
\left|u_{3,\alpha}^{\mathrm{ts}}-u_\alpha^{\mathrm{ts}}\right|
\le
\mathsf C_{\mathrm{c.f.}}^{\mathrm{ts}} n^{-2}.
\end{align*}
\end{theorem}

\begin{theorem}[Uniform error for the resampled Cornish--Fisher expansions]\label{uniform resample cf expansion}
Let $\varepsilon\in(0,1/2)$. Under Assumption~\ref{assume1}, for any fixed $B\geq1$ in the one-sided statement and any fixed $B\geq2$ in the two-sided statement, there exist finite constants $\hat{\mathsf C}_{\mathrm{c.f.}}^{\mathrm{up}}$ and $\hat{\mathsf C}_{\mathrm{c.f.}}^{\mathrm{ts}}$ such that
\begin{align*}
&\PP\left(
\sup_{\varepsilon\le\alpha\le1-\varepsilon}
\left|\hat u_{2,\alpha}^{\mathrm{up}}-\hat u_\alpha^{\mathrm{up}}\right|
\ge
\hat{\mathsf C}_{\mathrm{c.f.}}^{\mathrm{up}} n^{-3/2}
\right)
=
O(n^{-\lambda}),\\
&\PP\left(
\sup_{\varepsilon\le\alpha\le1-\varepsilon}
\left|\hat u_{3,\alpha}^{\mathrm{ts}}-\hat u_\alpha^{\mathrm{ts}}\right|
\ge
\hat{\mathsf C}_{\mathrm{c.f.}}^{\mathrm{ts}} n^{-2}
\right)
=
O(n^{-\lambda}).
\end{align*}
\end{theorem}
General versions of the two theorems above are stated at the end of Appendix~\ref{conjectural framework}, with proofs given in Appendices~\ref{sec theo13} and \ref{sec theo14}. These higher-order quantile error bounds imply the quantile error of SCB, stated in the following corollary, whose proof is given in Appendix~\ref{appendixcor5}.

\begin{corollary}[Quantile error of SCB]\label{cor1}
Under Assumption~\ref{assume1}, for any fixed $B\geq1$ in the one-sided statement and any fixed $B\geq2$ in the two-sided statement,
\begin{align*}
\sup_{\varepsilon\le \alpha\le 1-\varepsilon}
\left|
\hat u_{\alpha}^{\mathrm{up}}-u_{\alpha}^{\mathrm{up}}
\right|
=
O_p(n^{-1}),
\qquad
\sup_{\varepsilon\le \alpha\le 1-\varepsilon}
\left|
\hat u_{\alpha}^{\mathrm{ts}}-u_{\alpha}^{\mathrm{ts}}
\right|
=
O_p(n^{-3/2}).
\end{align*}
\end{corollary}

No restriction $B\geq 2$ is needed in Corollary~\ref{cor1} for the symmetric two-sided quantile error. The reason is that an $O(n^{-3/2})$ approximation in \eqref{Eq cheap boots uniform convergence} already suffices for this conclusion. The classical studentized bootstrap has quantile error order $O_p(n^{-3/2})$ in the symmetric two-sided case and $O_p(n^{-1})$ in the one-sided upper case; see Eqs.~(3.23) and (3.53) in \cite{hall2013bootstrap}. Thus, SCB achieves the same quantile error order as the studentized bootstrap with substantially less computation.

\section{Numerical Results}\label{num res}
We study the finite-sample performance of SCB through three function-of-mean examples, one V-statistic example, and two input uncertainty quantification examples for \(M/M/k\) queueing systems. All CIs are constructed at the nominal coverage level \(\alpha=0.95\). We compare SCB with the basic bootstrap (BB), standard error bootstrap (SE), and cheap bootstrap (CB) discussed in Section~\ref{sec 2 established methods}, and additionally include the percentile bootstrap (Per) and symmetric bootstrap (Sym) as two further benchmarks, defined below.

The percentile bootstrap takes the empirical \((1-\alpha)/2\) and \((1+\alpha)/2\) quantiles of the resample estimates \(\theta_b^*\) as the CI endpoints:
\[
\left[
\{\theta_b^*\}_{(1-\alpha)/2},
\{\theta_b^*\}_{(1+\alpha)/2}
\right].
\]
The symmetric bootstrap is included to separate the effect of symmetry from the effect of studentization. It uses the empirical \(\alpha\)-quantile of \(|\theta_b^*-\hat\theta|\) and constructs
\[
\left[
\hat\theta-\{|\theta_b^*-\hat\theta|\}_{\alpha},
\hat\theta+\{|\theta_b^*-\hat\theta|\}_{\alpha}
\right].
\]
Unlike the symmetric studentized bootstrap of \cite{hall1988symmetric}, this procedure does not use a pivotal statistic and therefore should not be regarded as a higher-order coverage-accurate method.

For SCB, let \(B\) denote the inner bootstrap size, \(B_1\) the outer bootstrap size, and \(n\) the sample size. Throughout the experiments, we take \(B_1\in\{39,79\}\), and use the \(38\)th or \(77\)th order statistic of \(\{|T_{b_1}^*|\}_{b_1=1}^{B_1}\), respectively, when computing empirical quantiles. Measured by the number of evaluations of the target statistic, the computational cost of SCB is
\[
1+B_1+B_1B+B=(B_1+1)(B+1).
\]
Here the four terms correspond respectively to the original estimate \(\hat\theta\), the \(B_1\) outer estimates \(\theta_{b_1}^*\), the \(B_1B\) inner estimates \(\theta_{b_1,b_2}^{**}\), and the \(B\) additional resample estimates used to construct \(\hat S\). For each benchmark method without an inner resampling layer, we set its single-layer resample size to
\[
B_{\mathrm{bench}}=(B_1+1)(B+1),
\]
so that all methods are compared under the same computational budget.

A minor degeneracy can arise when \(B=1\), since the scale estimates \(\hat S\) or \(S_{b_1}^*\) may become zero or nearly zero when the corresponding resample estimates coincide. To avoid degenerate intervals, we use a restart threshold \(\epsilon\), restarting the procedure whenever \(\hat S<\epsilon\) or \(S_{b_1}^*<\epsilon\). In our experiments, this safeguard is essentially irrelevant for \(B\geq 2\). The threshold values used below are reported with the corresponding experiments.

\subsection{Function-of-Mean Type Experiments}

We first consider three function-of-mean experiments, where the target parameter has the form $\theta(P)=f(\mathbb E_P[\mathbf X])$. Let $\mu_1,\mu_2,\ldots,\mu_d$ denote the components of the mean vector. Throughout this subsection, we take $B\in\{1,2,5,10,20,30,50\}$, $B_1\in\{39,79\}$, and $n\in\{10,25,50,100,200\}$. The true target values in all three examples can be computed analytically.

\begin{example}\label{ex1}
Let $d=2$ and $\mathbf X=(X_1,X_2)$, where $X_1$ and $X_2$ are independent and each follows a $\chi^2(5)$ distribution. The function is
\[
f(\mu_1,\mu_2)
=
\frac{(5-4\mu_1)\mu_2}{20}
+
\frac{\mu_1(4-\mu_1)}{40}.
\]
\end{example}

The results for Experiment~1 are shown in Figure~\ref{figure example1}. This example illustrates the main pattern observed in the function-of-mean experiments. SCB is consistently closer to the nominal \(95\%\) coverage level than the benchmark methods, especially in small and moderate samples where the benchmarks exhibit visible undercoverage. This coverage improvement comes with wider intervals when \(B\) is very small, but the widths decrease quickly as \(B\) increases, while the coverage remains stable. Increasing \(B_1\) from \(39\) to \(79\) has a comparatively mild effect: it slightly stabilizes the empirical pivotal quantiles and can reduce the interval widths, but this reduction is much less pronounced than the effect of increasing \(B\).

\begin{figure}[!tbp]
  \centering
  \begin{adjustbox}{max width=\linewidth}
    \begin{minipage}[t]{0.475\linewidth}\vspace{0pt}
      \centering
      \includegraphics[width=0.98\linewidth]{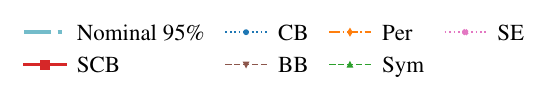}\\[0.2cm]
      \includegraphics[width=0.98\linewidth]{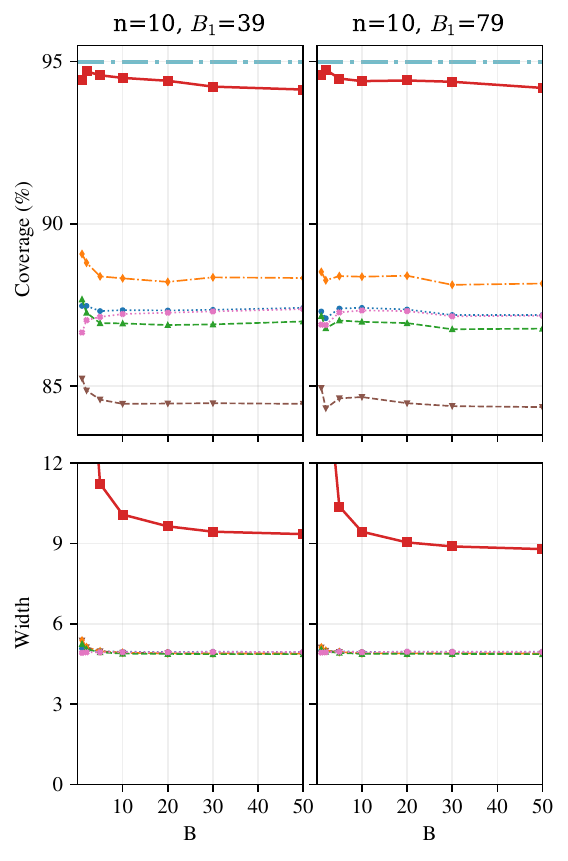}
      \vfill
      \includegraphics[width=0.98\linewidth]{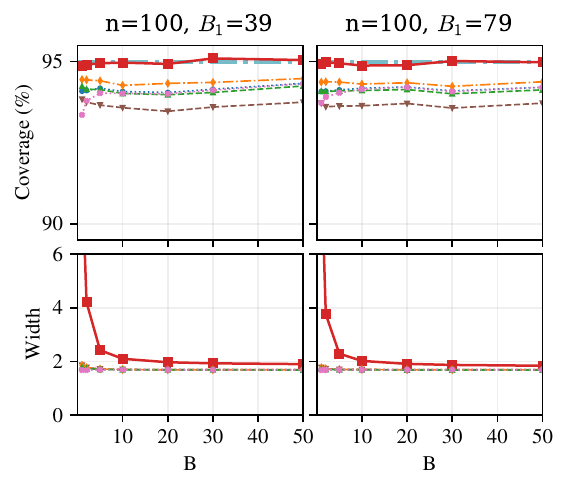}
    \end{minipage}
    \hfill
    \begin{minipage}[t]{0.475\linewidth}\vspace{0pt}
      \centering
      \includegraphics[width=0.98\linewidth]{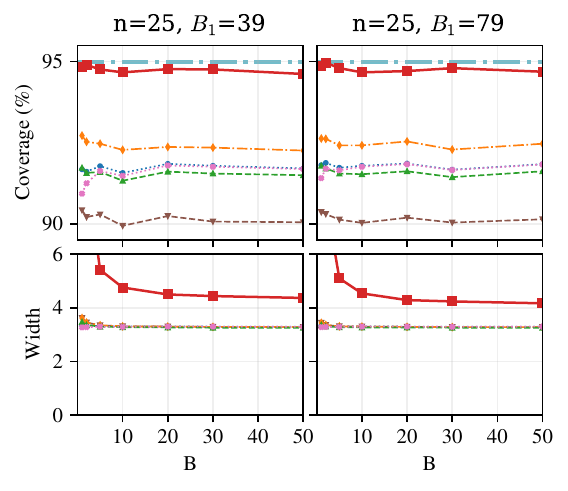}\\
      \includegraphics[width=0.98\linewidth]{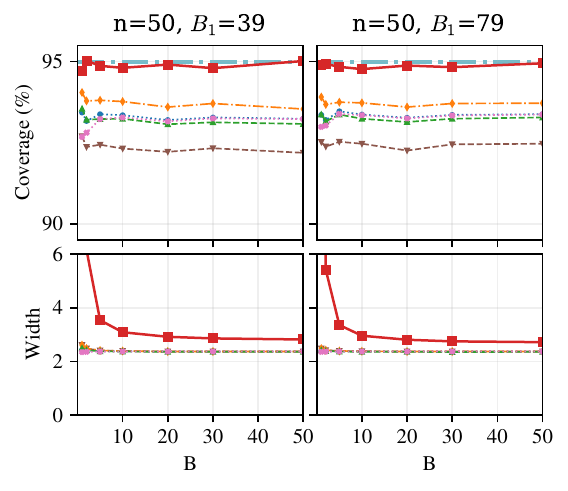}
      \vfill
      \includegraphics[width=0.98\linewidth]{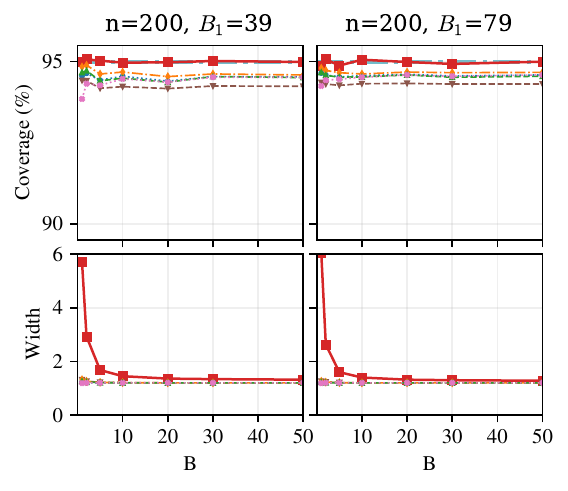}
    \end{minipage}
  \end{adjustbox}
  \caption{Coverage probabilities and interval widths for SCB and benchmark methods in Experiment~1. Benchmark methods use \(B_{\mathrm{bench}}=(B_1+1)(B+1)\), matching the same computational budget as SCB. The SCB restart threshold is \(\epsilon=0.01\) for \(B_1=39\) and \(\epsilon=0.05\) for \(B_1=79\). Each result is based on \(100{,}000\) independent experiments.}
  \label{figure example1}
\end{figure}

\begin{example}
The setup is the same as in Example~\ref{ex1}, except that $X_1$ and $X_2$ are independent and each follows the uniform distribution on $(0,10)$.
\end{example}

The results for Experiment~2 are shown in Figure~\ref{figure example2}. The same qualitative behavior appears in this bounded-input setting, although the undercoverage of the benchmark methods is less severe than in Experiment~1. SCB remains close to the nominal level across the sample sizes considered, and its additional width for very small \(B\) again diminishes rapidly as \(B\) grows.

\begin{figure}[!tbp]
  \centering
  \begin{adjustbox}{max width=\linewidth}
    \begin{minipage}[t]{0.475\linewidth}\vspace{0pt}
      \centering
      \includegraphics[width=0.98\linewidth]{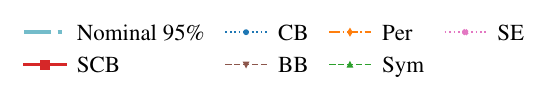}\\[0.2cm]
      \includegraphics[width=0.98\linewidth]{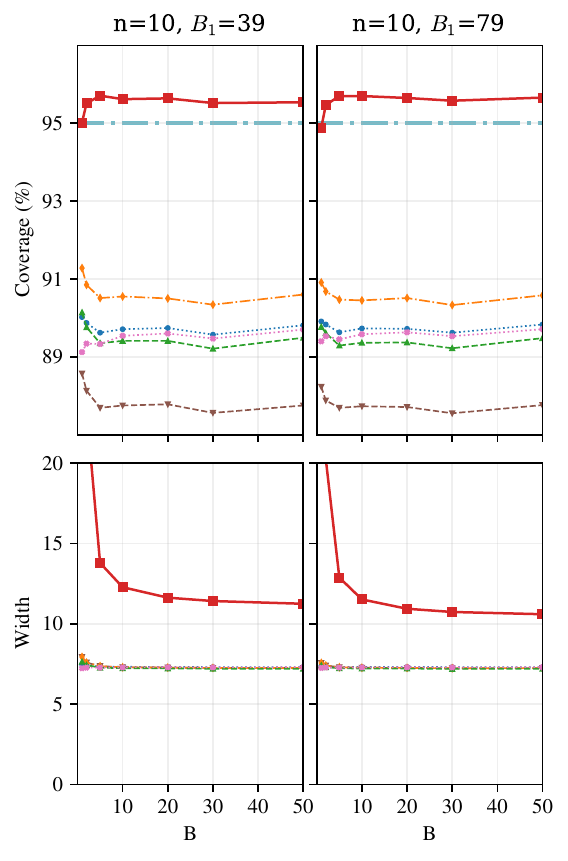}
      \vfill
      \includegraphics[width=0.98\linewidth]{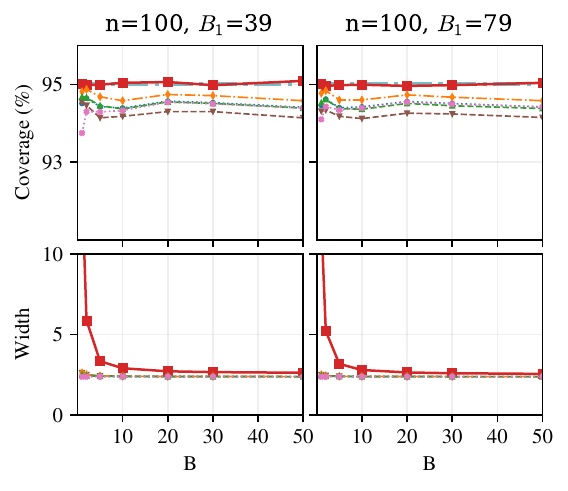}
    \end{minipage}
    \hfill
    \begin{minipage}[t]{0.475\linewidth}\vspace{0pt}
      \centering
      \includegraphics[width=0.98\linewidth]{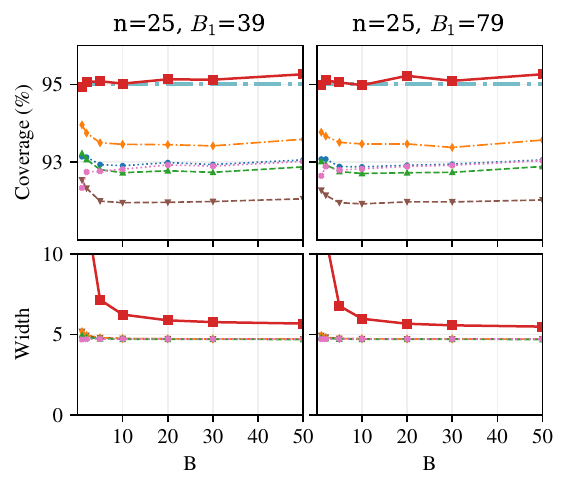}\\
      \includegraphics[width=0.98\linewidth]{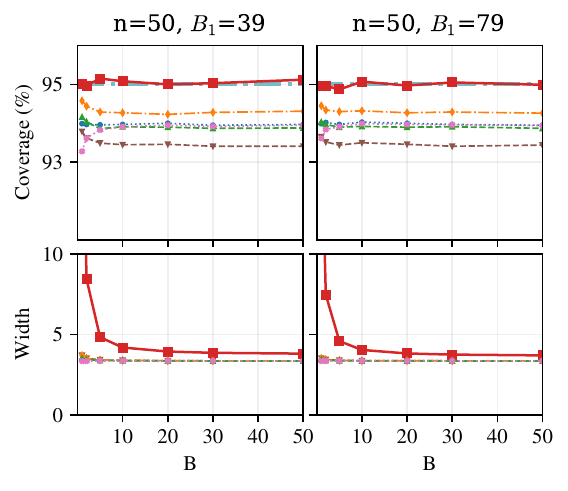}
      \vfill
      \includegraphics[width=0.98\linewidth]{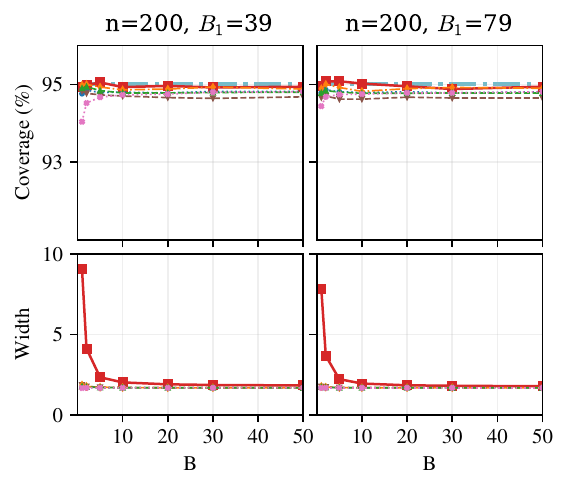}
    \end{minipage}
  \end{adjustbox}
  \caption{Coverage probabilities and interval widths for SCB and benchmark methods in Experiment~2. Benchmark methods use \(B_{\mathrm{bench}}=(B_1+1)(B+1)\), matching the same computational budget as SCB. The SCB restart threshold is \(\epsilon=0.01\). Each result is based on \(100{,}000\) independent experiments.}
  \label{figure example2}
\end{figure}

\begin{example}
Let $d=3$, and let $X_1,X_2,X_3$ be independent exponential random variables with mean $2$. The function is
\[
f(\mu_1,\mu_2,\mu_3)
=
\log\left(
0.1+
(\mu_1^2+4\mu_1\mu_2\mu_3+\mu_2^2\mu_3+0.4\mu_1\mu_2\mu_3^2)^2
\right).
\]
\end{example}

The results for Experiment~3 are shown in Figure~\ref{figure example3}. The same coverage--width tradeoff is observed, now in a more nonlinear function-of-mean example with skewed exponential inputs. The benchmark methods, particularly the percentile bootstrap, show more pronounced undercoverage, while SCB stays close to the nominal level. As before, the extra width of SCB is mainly concentrated at the smallest values of \(B\) and decreases substantially for moderate \(B\).

\begin{figure}[!tbp]
  \centering
  \begin{adjustbox}{max width=\linewidth}
    \begin{minipage}[t]{0.475\linewidth}\vspace{0pt}
      \centering
      \includegraphics[width=0.98\linewidth]{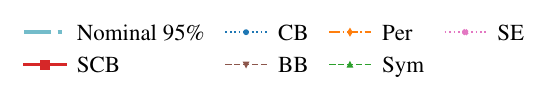}\\[0.2cm]
      \includegraphics[width=0.98\linewidth]{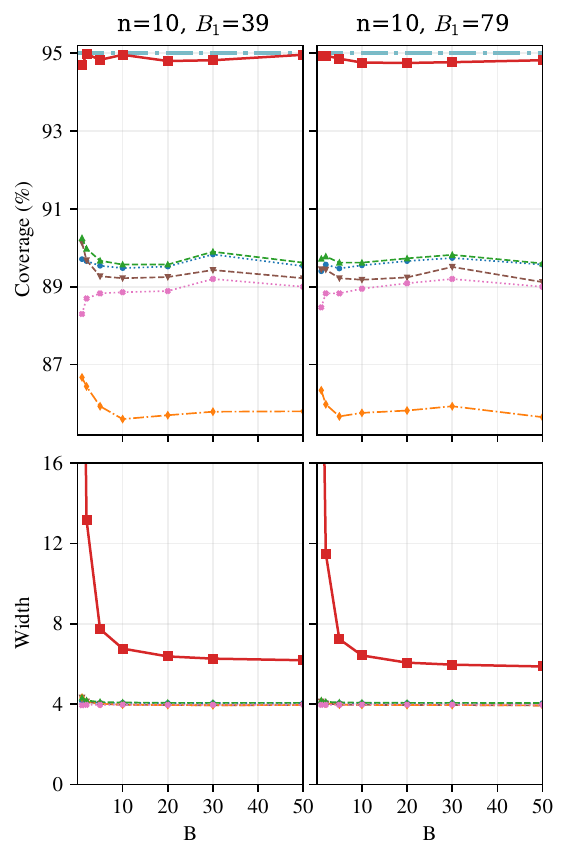}
      \vfill
      \includegraphics[width=0.98\linewidth]{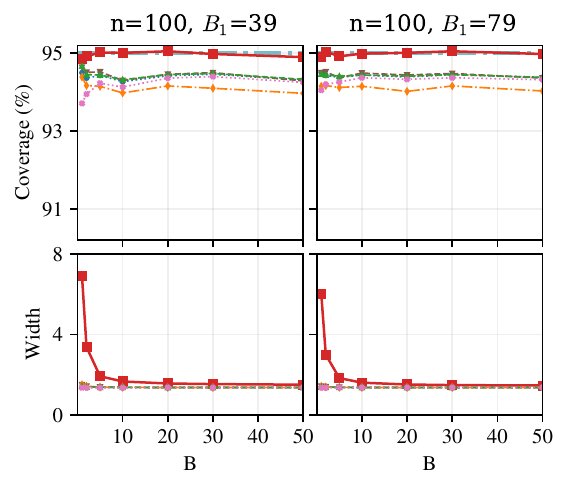}
    \end{minipage}
    \hfill
    \begin{minipage}[t]{0.475\linewidth}\vspace{0pt}
      \centering
      \includegraphics[width=0.98\linewidth]{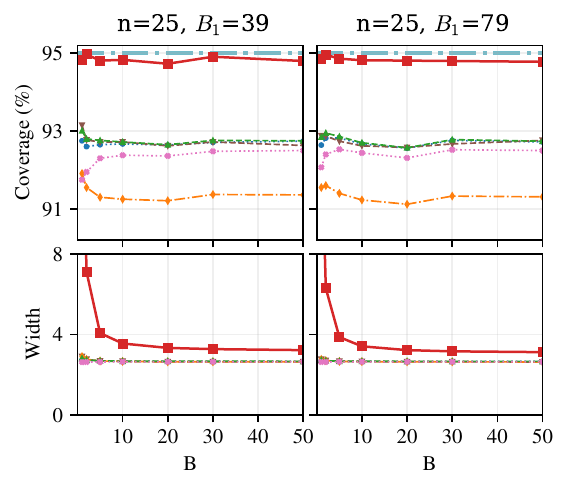}\\
      \includegraphics[width=0.98\linewidth]{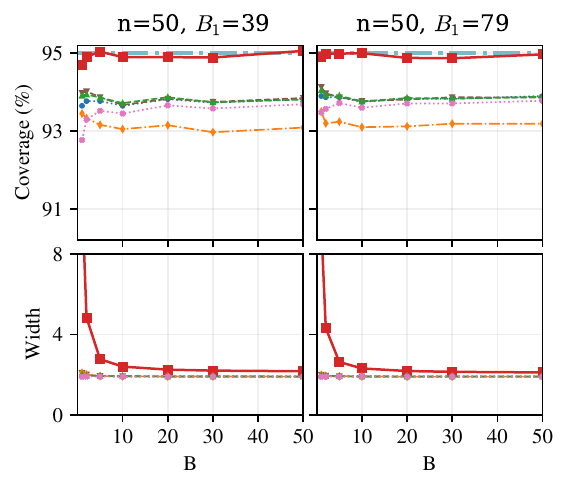}
      \vfill
      \includegraphics[width=0.98\linewidth]{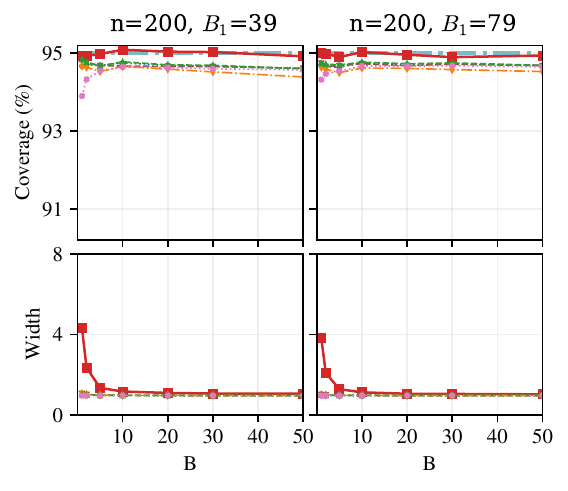}
    \end{minipage}
  \end{adjustbox}
  \caption{Coverage probabilities and interval widths for SCB and benchmark methods in Experiment~3. Benchmark methods use \(B_{\mathrm{bench}}=(B_1+1)(B+1)\), matching the same computational budget as SCB. The SCB restart threshold is \(\epsilon=0.01\). Each result is based on \(100{,}000\) independent experiments.}
  \label{figure example3}
\end{figure}

\subsection{V-statistic Experiment}

We next apply SCB to a V-statistic with target parameter
\[
\theta(P)=\mathbb E_P[h(X_1,X_2)],
\]
where $X_1$ and $X_2$ are i.i.d. samples from $P$. We take $B\in\{1,2,5,10,20\}$, $B_1\in\{39,79\}$, and $n\in\{10,20,30,40,50\}$.

\begin{example}
Let $h(x,y)=\min\left(12,(x-y)^2+x+y\right)$ and let $P=\Gamma(1,1)$. The ground truth is $\theta\approx 3.493254$, obtained through direct calculus.
\end{example}

The results are shown in Figure~\ref{figure example4}. The same coverage--width tradeoff observed in the function-of-mean experiments persists in this V-statistic example, but the finite-sample effects are more pronounced. The benchmark methods exhibit substantial undercoverage, especially for \(n=10\) and \(20\), while SCB approaches the nominal level more quickly and is already close by \(n=30\). As before, very small values of \(B\) lead to wider SCB intervals, but the widths decrease rapidly as \(B\) increases, while the coverage advantage is largely retained.

\begin{figure}[!tbp]
  \centering
  \begin{minipage}{0.90\linewidth}
    \centering
    \setlength{\tabcolsep}{0pt}
    \begin{tabular}{@{}c@{\hspace{0.015\linewidth}}c@{}}
      \begin{minipage}[t]{0.49\linewidth}\vspace{0pt}
        \centering
        \includegraphics[width=0.98\linewidth]{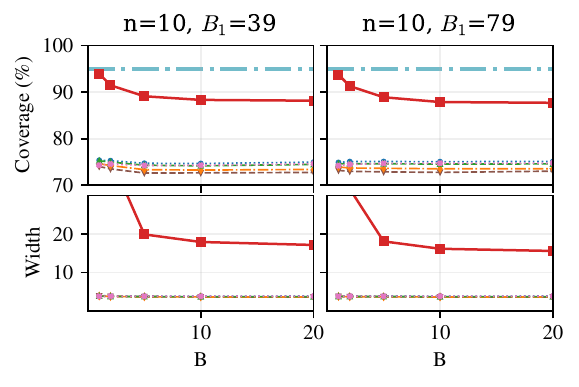}
      \end{minipage}
      &
      \begin{minipage}[t]{0.49\linewidth}\vspace{0pt}
        \centering
        \vspace*{0.9cm}
        \includegraphics[width=0.74\linewidth]{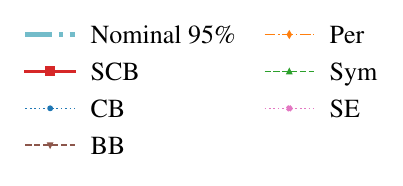}
      \end{minipage}
      \\
      \includegraphics[width=0.49\linewidth]{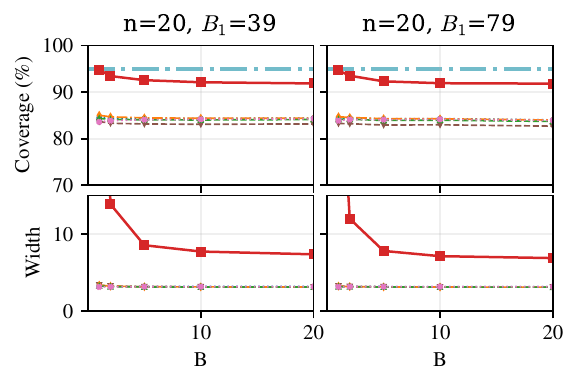}
      &
      \includegraphics[width=0.49\linewidth]{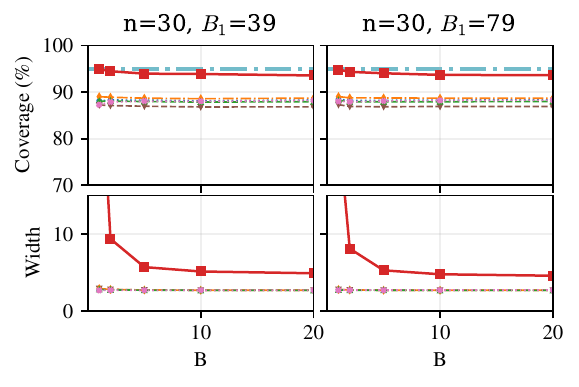}
      \\
      \includegraphics[width=0.49\linewidth]{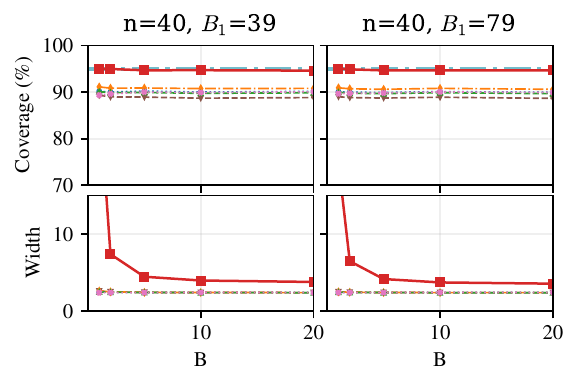}
      &
      \includegraphics[width=0.49\linewidth]{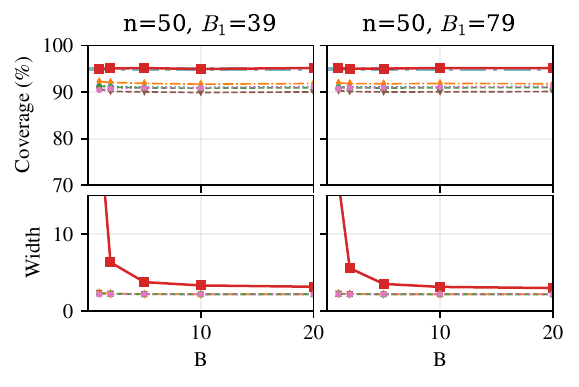}
    \end{tabular}
  \end{minipage}
  \caption{Coverage probabilities and interval widths for SCB and benchmark methods in Experiment~4. Benchmark methods use \(B_{\mathrm{bench}}=(B_1+1)(B+1)\), matching the same computational budget as SCB. The SCB restart threshold is \(\epsilon=10^{-4}\). Each result is based on \(100{,}000\) independent experiments.}
  \label{figure example4}
\end{figure}

\subsection{Input Uncertainty Quantification for M/M/k Queueing}

Finally, we apply SCB to input uncertainty quantification for queueing systems. We consider $M/M/k$ systems in which the mean inter-arrival time is denoted by $\lambda$ and the mean service time by $\mu$. In constructing the intervals, we use the empirical distributions of inter-arrival and service times, without assuming that the true system structure is known. Since the systems considered below are not stable in steady state, the target statistic is instead the expected total waiting time of the first $N=10$ customers. Each target-statistic evaluation in the bootstrap procedures is computed as the average of $1000$ Monte Carlo runs using either the empirical input distribution or the corresponding bootstrap input distribution. This number of replications is chosen so that input uncertainty dominates simulation noise.

\begin{example}
An $M/M/2$ system with inter-arrival mean $\lambda=1$ and service-time mean $\mu=8$.
\end{example}

For Experiment~5, the true value is \(\theta_5=10.66587\pm 0.00035\) with \(95\%\) confidence, obtained from \(10^9\) simulation replications. We take \(B\in\{1,2,5,10,15\}\), \(B_1\in\{39,79\}\), and \(n\in\{10,15,20,50\}\). The results are shown in Figure~\ref{figure example5}. This queueing example shows the same coverage--width tradeoff as above, but in a more pronounced form: the benchmark methods are substantially below the nominal level, whereas SCB stays close to nominal at the price of noticeably wider intervals when \(B\) is small.

\begin{figure}[!tbp]
  \centering
  \begin{minipage}{0.86\linewidth}
    \centering
    \includegraphics[width=0.88\linewidth,trim=0 0.3cm 0 0,clip]{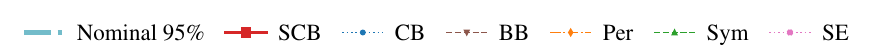}\\[0cm]
    \begin{minipage}[t]{0.485\linewidth}\vspace{0pt}
      \centering
      \includegraphics[width=0.96\linewidth]{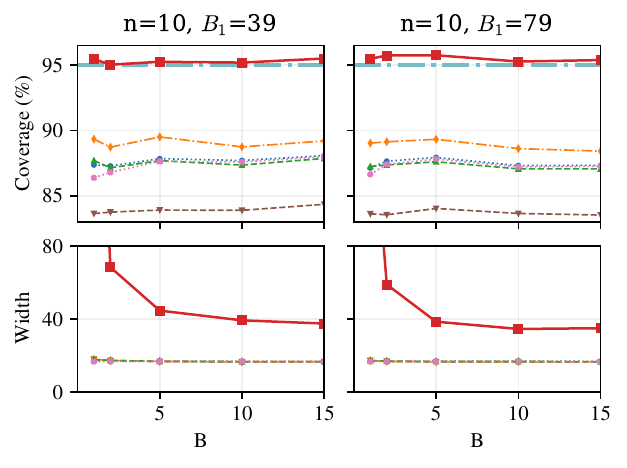}\\[0.12cm]
      \includegraphics[width=0.96\linewidth]{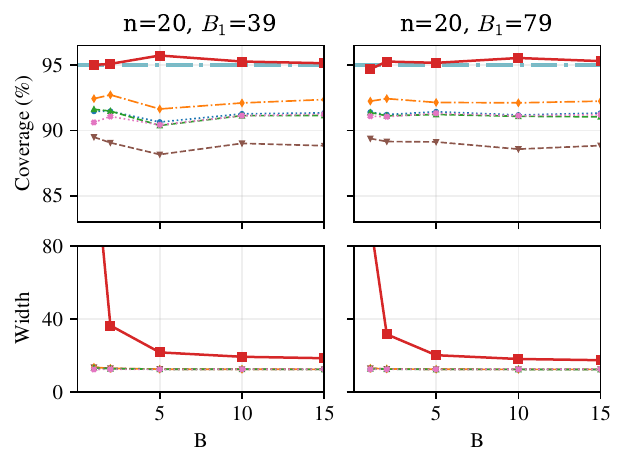}
    \end{minipage}
    \hfill
    \begin{minipage}[t]{0.485\linewidth}\vspace{0pt}
      \centering
      \includegraphics[width=0.96\linewidth]{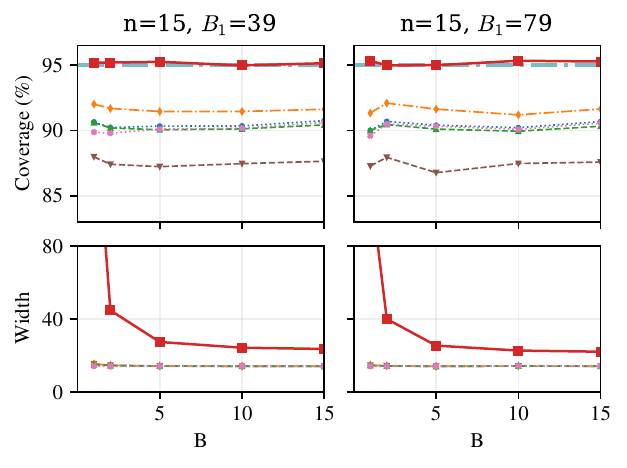}\\[0.12cm]
      \includegraphics[width=0.96\linewidth]{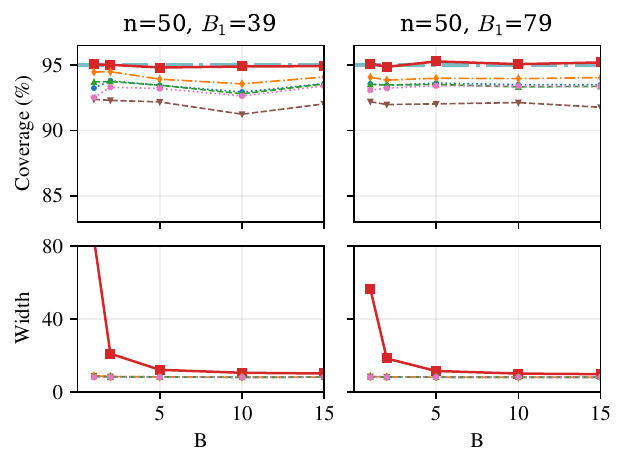}
    \end{minipage}
  \end{minipage}
  \caption{Coverage probabilities and interval widths for SCB and benchmark methods in Experiment~5, where the target statistic is the expected total waiting time of the first \(N=10\) customers in an \(M/M/2\) system. Benchmark methods use \(B_{\mathrm{bench}}=(B_1+1)(B+1)\), matching the same computational budget as SCB. The SCB restart threshold is \(\epsilon=10^{-3}\). Each result is based on \(10{,}000\) independent experiments.}
  \label{figure example5}
\end{figure}

\begin{example}
An $M/M/1$ system with inter-arrival mean $\lambda=1$ and service-time mean $\mu=1.1$.
\end{example}

For Experiment~6, the true value is \(\theta_6=1.77567\pm 0.00009\) with \(95\%\) confidence, obtained from \(10^9\) simulation replications. We take \(B\in\{1,2,5,10,20,30\}\), \(B_1\in\{39,79\}\), and \(n\in\{25,50,100,200,300,400\}\). The results are shown in Figure~\ref{figure example6}. Compared with Experiment~5, the benchmark methods are closer to the nominal level, especially as \(n\) increases, and hence the coverage advantage of SCB is more moderate. SCB nevertheless remains the most consistently calibrated method across the configurations. 

\begin{figure}[!tbp]
  \centering
  \begin{minipage}{0.84\linewidth}
    \centering
    \begin{minipage}[t]{0.485\linewidth}\vspace{0pt}
      \centering
      \includegraphics[width=0.95\linewidth]{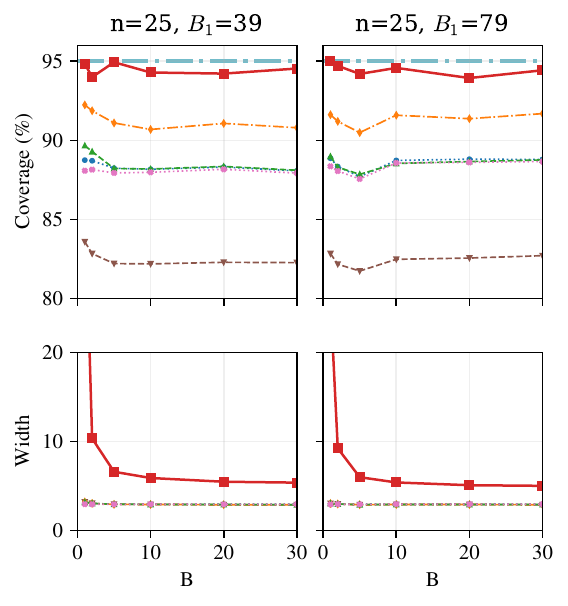}\\[0.12cm]
      \includegraphics[width=0.95\linewidth]{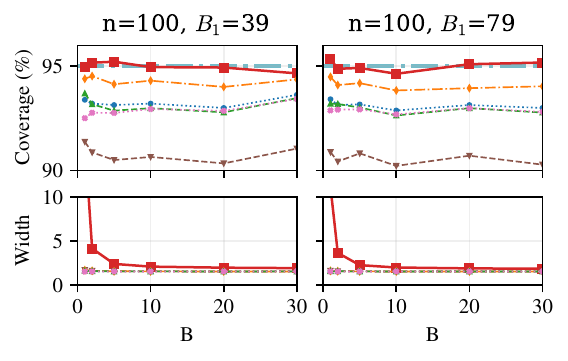}\\[0.12cm]
      \includegraphics[width=0.95\linewidth]{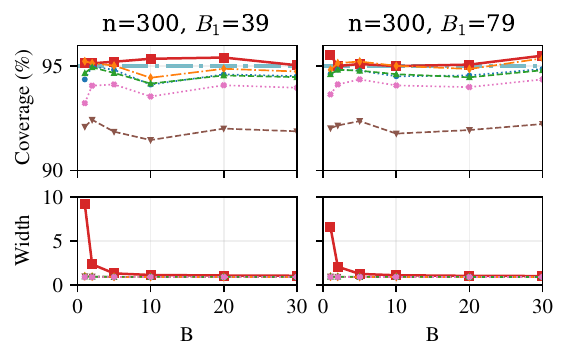}
    \end{minipage}
    \hfill
    \begin{minipage}[t]{0.485\linewidth}\vspace{0pt}
      \centering
      \includegraphics[width=0.82\linewidth]{images/example1legend.pdf}\\[0.62cm]
      \includegraphics[width=0.95\linewidth]{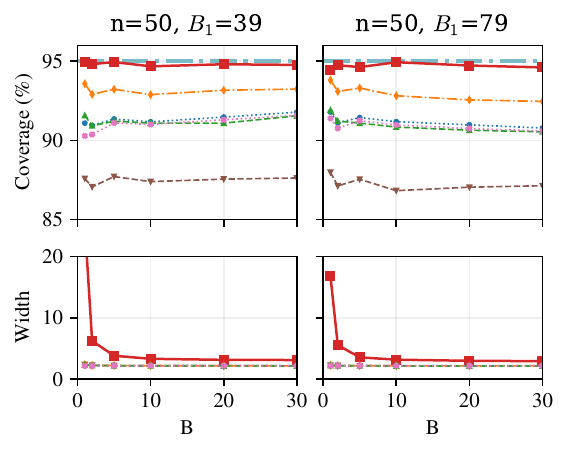}\\[0.12cm]
      \includegraphics[width=0.95\linewidth]{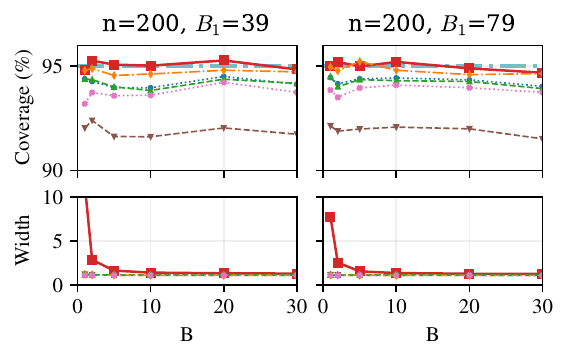}\\[0.12cm]
      \includegraphics[width=0.95\linewidth]{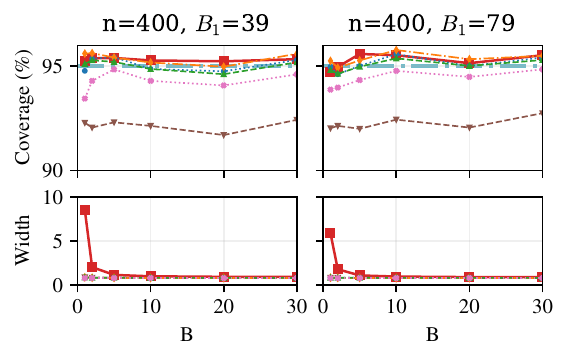}
    \end{minipage}
  \end{minipage}
  \caption{Coverage probabilities and interval widths for SCB and benchmark methods in Experiment~6, where the target statistic is the expected total waiting time of the first \(N=10\) customers in an \(M/M/1\) system. Benchmark methods use \(B_{\mathrm{bench}}=(B_1+1)(B+1)\), matching the same computational budget as SCB. The SCB restart threshold is \(\epsilon=10^{-8}\). Each result is based on \(10{,}000\) independent experiments.}
  \label{figure example6}
\end{figure}

\medskip
\noindent\emph{Practical recommendation.}
Taken together, the numerical results support the intended coverage--computation tradeoff of SCB. The coverage performance is robust across the range of inner bootstrap sizes considered; in particular, good coverage does not require taking \(B\) large. The choice of \(B\) mainly affects the conservativeness of the intervals. Very small choices such as \(B=1\) or \(2\) keep the computation minimal but can lead to wide intervals, necessitated to counteract the high variability of the cheap scale estimators \(\hat S\) and \(S_{b_1}^*\). The outer bootstrap size \(B_1\) controls the Monte Carlo stability of the empirical pivotal quantiles; increasing \(B_1\) from \(39\) to \(79\) tends to make the intervals slightly more stable and sometimes shorter, but this effect is not substantial.

We therefore recommend using a moderate outer bootstrap size, such as \(B_1=39\), together with a small but not minimal inner bootstrap size, typically \(B\in\{3,4,5\}\), with \(B=5\) as a conservative default choice. If additional computational budget is available, \(B\) can be increased modestly to further reduce interval width.

\begin{acks}[Acknowledgments]
We gratefully acknowledge support from the InnoHK initiative of the Innovation and Technology Commission of the Hong Kong Special Administrative Region Government, Laboratory for AI-Powered Financial Technologies, and Columbia Innovation Hub Award.
\end{acks}

\begin{appendix}
\section{Auxiliary Edgeworth Expansion Results}\label{appendix a}

This appendix collects the Edgeworth expansion results used as background tools in the proofs. We first state the population Edgeworth expansion for smooth functions of sample means, then its bootstrap analogue, and finally the corresponding resampled bootstrap analogue needed for the nested resampling arguments.

\subsection{Population Edgeworth Expansion}

\begin{theorem}[Adapted from Hall (2013), Theorem 2.2]\label{Theo Hall1}
Assume that, for an integer $\nu \ge 1$, the function $\widetilde A$ has $\nu+3$ continuous derivatives in a neighborhood of $\mu$, that $\widetilde A(\mu)=0$, that $\mathbb E\|\mathbf X\|^{\nu+3}<\infty$, that the characteristic function $\chi$ of $\mathbf X$ satisfies Cram\'er's condition
\[
\limsup_{\|t\|\to\infty} \abs{\chi(t)} < 1,
\]
and that the asymptotic variance of $\sqrt n\,\widetilde A(\bar{\mathbf X})$ equals $1$. Then there exists a constant $\mathsf C_\nu>0$ such that
\[
\mathbb P_n\!\left(\sqrt n\,\widetilde A(\bar{\mathbf X})\le x\right)
=
\Phi(x)
+
\sum_{j=1}^{\nu} n^{-j/2}\pi_j(x)\phi(x)
+
\varepsilon_n^\nu(x),
\]
uniformly in $x\in\mathbb R$, where
\[
\sup_{x\in\mathbb R}\abs{\varepsilon_n^\nu(x)}
\le
\mathsf C_\nu n^{-(\nu+1)/2}.
\]
Moreover, for each $j=1,\ldots,\nu$, $\pi_j$ is a polynomial of degree $3j-1$, odd when $j$ is even and even when $j$ is odd, whose coefficients depend polynomially on the moments of $\mathbf X$ up to order $j+2$ and on $\widetilde A$.
\end{theorem}
Compared with the original statement in \cite{hall2013bootstrap}, we strengthen the smoothness and moment assumptions from order $\nu+2$ to order $\nu+3$, so that the remainder satisfies the sharper bound
\[
\sup_{x\in\mathbb R}\abs{\varepsilon_n^\nu(x)}
\le
\mathsf C_\nu n^{-(\nu+1)/2}.
\]

\subsection{Bootstrap Edgeworth Expansion}

\begin{theorem}[Adapted from Hall (2013), Theorem 5.1]\label{Theo Hall2}
Let $\lambda>0$ be given, and let $l=l(\lambda)$ be a sufficiently large positive number. Assume that $g$ and $h$ each have $\nu+3$ bounded derivatives in a neighborhood of $\mu$, that $\mathbb E\|\mathbf X\|^l<\infty$, and that the characteristic function $\chi$ of $\mathbf X$ satisfies Cram\'er's condition
\[
\limsup_{\|t\|\to\infty} \abs{\chi(t)} < 1.
\]
Then there exists a constant $C>0$ such that
\begin{align*}
\mathbb P_n\!\left(
\sup_{x\in\mathbb R}
\left|
\mathbb P_n\!\left(\sqrt n\,\tilde A(\bar{\mathbf X}^*)\le x \mid \mathcal X\right)
-
\Phi(x)
-
\sum_{j=1}^{\nu} n^{-j/2}\hat\pi_j(x)\phi(x)
\right|
>
C n^{-(\nu+1)/2}
\right)
=
O(n^{-\lambda}),    
\end{align*}
\[
\mathbb P_n\!\left(
\max_{1\le j\le \nu}
\sup_{x\in\mathbb R}
(1+\abs{x})^{-(3j-1)}\abs{\hat\pi_j(x)}
>
C
\right)
=
O(n^{-\lambda}),
\]
and
\begin{align}\label{eqtemp}
\mathbb P_n\!\left(
\max_{1\le j\le \nu}
\sup_{x\in\mathbb R}
(1+\abs{x})^{-(3j-1)}\abs{\hat\pi_j'(x)}
>
C
\right)
=
O(n^{-\lambda}).
\end{align}
Here, $\tilde A$ is either $A$ or $A_s$. Correspondingly, we write $\pi_j=p_j$ in the former case and $\pi_j=q_j$ in the latter case, and $\hat\pi_j$ denotes the corresponding bootstrap version obtained by replacing the population moments in $\pi_j$ with the sample moment estimators.
\end{theorem}

Theorem \ref{Theo Hall2} is a slight strengthening of \cite{hall2013bootstrap}, in that it includes \eqref{eqtemp} as an additional conclusion. This extension is immediate from the proof in \cite{hall2013bootstrap}, since the argument applies to any polynomial $\pi_j$ of degree $3j-1$, and therefore also to its derivative $\pi_j'$.

\subsection{Resampled Bootstrap Edgeworth Expansion}

\begin{theorem}[Resampled version of Hall (2013), Theorem 5.1]\label{Theo HallHe}
Let $\lambda>0$ be given, and let $l=l(\lambda)$ be a sufficiently large positive number. Assume that $g$ and $h$ each have $\nu+3$ bounded derivatives in a neighborhood of $\mu$, that $\mathbb E\|\mathbf X\|^l<\infty$, and that the characteristic function $\chi$ of $\mathbf X$ satisfies Cram\'er's condition. Then there exists a constant $C>0$ such that
\[
\mathbb P_n\!\left(
\sup_{x\in\mathbb R}
\left|
\mathbb P_n\!\left(\sqrt n\,\widetilde A^*(\bar{\mathbf X}^{**}) \le x \mid \mathcal X^*,\mathcal X\right)
-\Phi(x)
-\sum_{j=1}^{\nu} n^{-j/2}\pi_j^*(x)\phi(x)
\right|
>
C n^{-(\nu+1)/2}
\right)
=
O(n^{-\lambda}),
\]
\[
\mathbb P_n\!\left(
\max_{1\le j\le \nu}
\sup_{x\in\mathbb R}
(1+\abs{x})^{-(3j-1)}\abs{\pi_j^*(x)}
>
C
\right)
=
O(n^{-\lambda}),
\]
and
\[
\mathbb P_n\!\left(
\max_{1\le j\le \nu}
\sup_{x\in\mathbb R}
(1+\abs{x})^{-(3j-1)}\abs{(\pi_j^*)'(x)}
>
C
\right)
=
O(n^{-\lambda}).
\]
Here, $\widetilde A^*$ is either $A^*$ or $A_s^*$. Correspondingly, we write $\pi_j^*=p_j^*$ in the former case and $\pi_j^*=q_j^*$ in the latter case. The polynomials $p_j^*$ and $q_j^*$ are obtained by replacing the population moments in $p_j$ and $q_j$ with the corresponding moment estimators computed from the resampled data.
\end{theorem}
\begin{proof}
Although this version is not explicitly stated in \cite{hall2013bootstrap}, it is implicit in the discussion in Section 3.11 therein. The result follows by a straightforward modification of the proof of Theorem 5.1 in \cite{hall2013bootstrap}, replacing $\hat P$ by $P^*$ throughout. In particular, the corresponding events $\mathcal E_1,\dots,\mathcal E_4$ can be defined analogously, and their probabilities remain $1-O(n^{-\lambda})$. Conditional on these events, the remainder of the argument proceeds identically to that of Theorem 5.1. We therefore omit the details.
\end{proof}
\section{Proof of Edgeworth Expansion for $t$-Distributions}\label{appendix B}
Throughout this section, we use the notation
\[
C_{[a,b,c,\cdots]}
\]
to denote a generic constant depending \emph{only} on the quantities $a,b,c,\cdots$. Its value may change from line to line, even when the notation appears unchanged.

\subsection{Expectation Cancellation of Monomial Moment Errors under an Even Weight}\label{secprop1}

This subsection proves Proposition~\ref{theo:expect-cancellation-good}. The proof relies on three auxiliary lemmas. The first lemma linearizes monomial functions of empirical moments and controls the expectation of the nonlinear remainder. The second lemma gives moment bounds for the remainder in the jackknife expansion of the pivotal statistic on a high-probability Lipschitz event. The third lemma converts a uniform Edgeworth remainder into an expectation bound for a restricted class of test functions. These ingredients are then combined to show that the monomial sample-moment estimation error cancels in expectation after being multiplied by an even smooth function of the pivotal statistic.

Throughout this subsection, \(\tilde A\) denotes the smooth function defined in either \eqref{A def} or \eqref{A def2}.
\subsubsection{Linearization of Monomial Sample-Moment Estimators}\label{secProp1lemma1}
Let \(\mathbf X_1,\mathbf X_2,\ldots \in \mathbb R^d\) be i.i.d. random vectors. Let \(\mathbb N_0\) be the set of nonnegative integers. For a vector \(\mathbf x=(x_1,\ldots,x_d)\in\mathbb R^d\) and a multi-index \(\alpha=(\alpha_1,\ldots,\alpha_d)\in\mathbb N_0^d\), write
\[
\abs{\alpha}\coloneqq \alpha_1+\cdots+\alpha_d,
\qquad
\mathbf x^\alpha\coloneqq x_1^{\alpha_1}\cdots x_d^{\alpha_d}.
\]
For each multi-index \(\alpha\in\mathbb N_0^d\), define the mixed moment, the corresponding sample moment, and their difference by
\[
m_\alpha\coloneqq\mathbb E[\mathbf X_1^\alpha],
\qquad
\hat m_{n,\alpha}\coloneqq\frac1n\sum_{i=1}^n \mathbf X_i^\alpha,
\qquad
\Delta_{n,\alpha}
\coloneqq
\hat m_{n,\alpha}-m_\alpha
=
\frac1n\sum_{i=1}^n \bigl(\mathbf X_i^\alpha-m_\alpha\bigr).
\]
Let \(\mathcal A\subset \mathbb N_0^d\) be a finite collection of multi-indices, and write
\[
\bm m\coloneqq\bigl(m_\alpha\bigr)_{\alpha\in\mathcal A},
\qquad
\hat{\bm m}_n\coloneqq\bigl(\hat m_{n,\alpha}\bigr)_{\alpha\in\mathcal A}.
\]
We regard a vector \(\bm z\in\mathbb R^{\abs{\mathcal A}}\) as $\bm z=(z_\alpha)_{\alpha\in\mathcal A}$. Let \(H:\mathbb R^{\abs{\mathcal A}}\to\mathbb R\) be a monomial with unit coefficient. Thus, for some integers \(r_\alpha\in\mathbb N_0\),
\[
H(\bm z)=\prod_{\alpha\in\mathcal A}z_\alpha^{r_\alpha}.
\]
Define
\begin{align}\label{def ell}
    d_H\coloneqq \sum_{\alpha\in\mathcal A}r_\alpha,
\qquad
a_{\max}\coloneqq \max_{\alpha\in\mathcal A}\abs{\alpha},
\qquad
\ell\coloneqq d_H\times a_{\max}.
\end{align}
Assume that $\mathbb E\|\mathbf X_1\|^\ell<\infty$. For notational simplicity, for \(\alpha\in\mathcal A\), write
\[
\partial_\alpha H(\bm z)
\coloneqq
\frac{\partial}{\partial z_\alpha}H(\bm z).
\]
Finally, define
\begin{align}\label{eq MMM}
    M\coloneqq H(\bm m),
\qquad
M_n\coloneqq H(\hat{\bm m}_n).
\end{align}
Throughout this subsection, \(C_M\) denotes a finite constant depending only on the algebraic form of the monomial \(H\) and on \(\mathcal A\), but not on \(n\), \(\bm m\), or the underlying distribution of the i.i.d. sequence \(\mathbf X_1,\mathbf X_2,\ldots\). Its value may change from line to line.

\begin{lemma}\label{lm linearization}
There exists a measurable function \(\varrho_M:\mathbb R^d\to\mathbb R\) such that
\[
M_n-M=\frac1n\sum_{i=1}^n \varrho_M(\mathbf X_i)+R_{[4]},
\]
where
\[
\varrho_M(\mathbf x)
=
\sum_{\alpha\in\mathcal A}
\partial_\alpha H(\bm m)\times
\bigl(\mathbf x^\alpha-m_\alpha\bigr),
\qquad
\mathbb E[\varrho_M(\mathbf X_1)]=0,
\]
and
\[
\mathbb E\abs{R_{[4]}}
\le
\frac{C_M}{n}
\left(1+\mathbb E\|\mathbf X_1\|^\ell\right)^{d_H-1}.
\]
\end{lemma}

\begin{proof}
Let $\bm{\Delta}_n\coloneqq \hat{\bm m}_n-\bm m
$. Since \(H\) is a monomial, its multivariate Taylor expansion at \(\bm m\) terminates after finitely many terms. Therefore,
\[
H(\hat{\bm m}_n)-H(\bm m)
=
\sum_{\alpha\in\mathcal A}
\partial_\alpha H(\bm m)\times 
\Delta_{n,\alpha}
+
R_{[4]},
\]
where \(R_{[4]}\) is a finite linear combination of monomials in the coordinates of \(\bm{\Delta}_n\) with total degree at least two and at most \(d_H\). Hence
\begin{align*}
H(\hat{\bm m}_n)-H(\bm m)
&=
\sum_{\alpha\in\mathcal A}
\partial_\alpha H(\bm m)
\times\frac1n\sum_{i=1}^n
\bigl(\mathbf X_i^\alpha-m_\alpha\bigr)
+
R_{[4]} \\
&=
\frac1n\sum_{i=1}^n
\sum_{\alpha\in\mathcal A}
\partial_\alpha H(\bm m)
\bigl(\mathbf X_i^\alpha-m_\alpha\bigr)
+
R_{[4]} \\
&=
\frac1n\sum_{i=1}^n \varrho_M(\mathbf X_i)+R_{[4]}.
\end{align*}
The identity \(\mathbb E[\mathbf X_1^\alpha]=m_\alpha\) implies $\mathbb E[\varrho_M(\mathbf X_1)]=0$. It remains to bound the remainder. For a vector $\bm k=(k_\alpha)_{\alpha\in\mathcal A}\in\mathbb N_0^{\abs{\mathcal A}}$, write $\abs{\bm k}\coloneqq \sum_{\alpha\in\mathcal A}k_\alpha$ and $\bm{\Delta}_n^{\bm k}\coloneqq
\prod_{\alpha\in\mathcal A}\Delta_{n,\alpha}^{k_\alpha}$. By the binomial expansion,
\[
H(\bm m+\bm{\Delta}_n)
=
\prod_{\alpha\in\mathcal A}
(m_\alpha+\Delta_{n,\alpha})^{r_\alpha}
=
\sum_{0\le k_\alpha\le r_\alpha}
b_{\bm k}(\bm m)\times\bm{\Delta}_n^{\bm k},\quad\text{with}\quad b_{\bm k}(\bm m)
\coloneqq
\prod_{\alpha\in\mathcal A}
\binom{r_\alpha}{k_\alpha}
m_\alpha^{r_\alpha-k_\alpha}.
\]
The terms with \(\abs{\bm k}=0\) give \(H(\bm m)\), and the terms with \(\abs{\bm k}=1\) give the linear part $\frac1n\sum_{i=1}^n \varrho_M(\mathbf X_i)$. Therefore,
\[
R_{[4]}
=
\sum_{\substack{0\le k_\alpha\le r_\alpha\\ \abs{\bm k}\ge2}}
b_{\bm k}(\bm m)\bm{\Delta}_n^{\bm k}.
\]
We now bound each term in this finite sum. First, for every \(\alpha\in\mathcal A\),
\[
\abs{m_\alpha}
=
\left|\mathbb E[\mathbf X_1^\alpha]\right|
\le
\mathbb E\abs{\mathbf X_1^\alpha}
\le
\mathbb E\|\mathbf X_1\|^{\abs{\alpha}}
\le
1+\mathbb E\|\mathbf X_1\|^\ell.
\]
Define $s\coloneqq\abs{\bm k}$. Then $\sum_{\alpha\in\mathcal A}(r_\alpha-k_\alpha)
=
d_H-s$. Therefore, 
\begin{align}\label{eq ag0}
    \abs{b_{\bm k}(\bm m)}
\le
C_M
\left(1+\mathbb E\|\mathbf X_1\|^\ell\right)^{d_H-s}.
\end{align}
Next, by H\"older's inequality,
\begin{align}\label{eq ag1}
    \mathbb E\abs{\bm{\Delta}_n^{\bm k}}
=
\mathbb E\left[
\prod_{\alpha\in\mathcal A}
\abs{\Delta_{n,\alpha}}^{k_\alpha}
\right]
\le
\prod_{\alpha:k_\alpha>0}
\left(
\mathbb E\abs{\Delta_{n,\alpha}}^s
\right)^{k_\alpha/s}.
\end{align}
For each \(\alpha\in\mathcal A\), \(\Delta_{n,\alpha}\) is the sample average of the centered random variables $\mathbf X_i^\alpha-m_\alpha$. By the Marcinkiewicz--Zygmund inequality for centered sums, for every integer \(s\ge2\), there exists a finite constant \(C_s\), depending only on \(s\), such that
\begin{align}\label{eq ag2}
    \mathbb E\abs{\Delta_{n,\alpha}}^s
\le
\frac{C_s}{n^{s/2}}\times 
\mathbb E\abs{\mathbf X_1^\alpha-m_\alpha}^s.
\end{align}
Since \(s\le d_H\), \(\abs{\alpha}\le a_{\max}\), and \(\ell=d_H\times a_{\max}\), we have \(s\times\abs{\alpha}\le \ell\). Hence
\begin{align}\label{eq ag3}
    \mathbb E\abs{\mathbf X_1^\alpha-m_\alpha}^s
\le
C_M\left(1+\mathbb E\abs{\mathbf X_1^\alpha}^s+\abs{m_\alpha}^s\right)
\le
C_M\left(1+\mathbb E\|\mathbf X_1\|^\ell\right).
\end{align}
Since \(2\le s\le d_H\), \(\max_{2\le j\le d_H}C_j\) depends only on \(d_H\), and hence only on the algebraic form of \(H\). Absorbing this finite maximum into \(C_M\), it follows from \eqref{eq ag1}, \eqref{eq ag2} and \eqref{eq ag3} that
\[
\mathbb E\abs{\bm{\Delta}_n^{\bm k}}
\le
C_M \times n^{-s/2}
\left(1+\mathbb E\|\mathbf X_1\|^\ell\right).
\]
Combining this with the bound on \(b_{\bm k}(\bm m)\), i.e., \eqref{eq ag0}, we obtain
\[
\mathbb E\left[
\left|b_{\bm k}(\bm m)\bm{\Delta}_n^{\bm k}\right|
\right]
\le
C_M
\left(1+\mathbb E\|\mathbf X_1\|^\ell\right)^{d_H-s}\times 
n^{-s/2}
\left(1+\mathbb E\|\mathbf X_1\|^\ell\right).
\]
Also, since \(s\ge2\), \(1+\mathbb E\|\mathbf X_1\|^\ell\ge1\) and \(d_H-s+1\le d_H-1\), we get $\mathbb E\left[
\left|b_{\bm k}(\bm m)\bm{\Delta}_n^{\bm k}\right|
\right]
\le
\frac{C_M}{n}
\left(1+\mathbb E\|\mathbf X_1\|^\ell\right)^{d_H-1}$. Finally, the sum defining \(R_{[4]}\) contains only finitely many terms, and the number of terms depends only on the monomial \(H\) and on \(\mathcal A\). Therefore,
\[
\mathbb E\abs{R_{[4]}}
\le
\frac{C_M}{n}
\left(1+\mathbb E\|\mathbf X_1\|^\ell\right)^{d_H-1}.
\]
This proves the lemma.
\end{proof}

\begin{exam}
To illustrate the notation and Lemma~\ref{lm linearization}, consider the one-dimensional case \(d=1\), so that we simply write \(X_1\in\mathbb R\). Let $M=m_1m_3^2$, $H(z_1,z_3)=z_1z_3^2$, $\mathcal A=\{1,3\}$, $d_H=3$, $a_{\max}=3$ and $\ell=d_H\times a_{\max}=9$. Moreover, $\varrho_M(x)
=
m_3^2(x-m_1)
+
2m_1m_3(x^3-m_3)$. Lemma~\ref{lm linearization} yields
\[
\left(\frac{1}{n}\sum_{i=1}^n X_i\right)
\left(\frac{1}{n}\sum_{i=1}^n X_i^3\right)^2
-
M
=
\frac1n\sum_{i=1}^n \varrho_M(X_i)+R_{[4]}.
\]
More explicitly, $R_{[4]}
=
2m_3\Delta_{n,1}\Delta_{n,3}
+
m_1\Delta_{n,3}^2
+
\Delta_{n,1}\Delta_{n,3}^2$. If \(\mathbb E\abs{X_1}^9<\infty\), $\mathbb E\abs{R_{[4]}}
\le
\frac{C_M}{n}
\left(1+\mathbb E\abs{X_1}^9\right)^2$.
\end{exam}

\subsubsection{Jackknife Expansion and Remainder Bounds for the Pivotal Statistic}

We next record a decomposition of the jackknife difference of the pivotal statistic in the smooth-function model. The purpose is to separate the leading first-order term from the higher-order remainder. Assume that Assumption~\ref{assume1} holds, and define
\[
\bm{\mu}\coloneqq \mathbb E[\mathbf X_1],\qquad
\bar{\mathbf X}_n\coloneqq \frac1n\sum_{j=1}^n \mathbf X_j,
\qquad
W_n\coloneqq \sqrt n\,\tilde A(\bar{\mathbf X}_n).
\]
For each \(i\), define the leave-one-out sample mean
\[
\bar{\mathbf X}_{-i}\coloneqq \frac1{n-1}\sum_{j\ne i}\mathbf X_j,
\qquad
W_{n,-i}\coloneqq \sqrt{\,n-1\,}\,\tilde A(\bar{\mathbf X}_{-i}),
\]
and set $\varphi(\mathbf x)\coloneqq \nabla \tilde A(\bm{\mu})^\top(\mathbf x-\bm{\mu})$. Then $\mathbb E[\varphi(\mathbf X_1)]=\mathbb E[\varphi(\bar{\mathbf X}_{-i})]=0$. The decomposition we need is
\begin{equation}\label{eq:jackknife-pivotal-expansion}
W_n-W_{n,-i}
=
\frac{1}{\sqrt{n}}\bigl(\varphi(\mathbf X_i)-\varphi(\bar{\mathbf X}_{-i})\bigr)
+R_{[5]},
\end{equation}
where \(R_{[5]}\) denotes the higher-order remainder term corresponding to the fixed index \(i\) under consideration. Since \(\tilde A\in C^2\) in a neighborhood of \(\bm{\mu}\), there exist \(\delta>0\) and \(L>0\) such that
\[
\|\nabla \tilde A(\mathbf u)-\nabla \tilde A(\mathbf v)\|
\le
L\|\mathbf u-\mathbf v\|,
\]
whenever \(\|\mathbf u-\bm{\mu}\|\le \delta\) and \(\|\mathbf v-\bm{\mu}\|\le \delta\), where \(\|\cdot\|\) denotes the Euclidean norm. We therefore introduce the event on which all leave-one-out sample means remain inside this Lipschitz neighborhood:
\begin{align}\label{Eq goodevent Lipschitz}
\mathcal E_n^{\mathrm{lip}}
\coloneqq 
\Bigl\{
\|\bar{\mathbf X}_{-i}-\bm{\mu}\|\le \delta
\text{ for all }1\le i\le n
\Bigr\}.
\end{align}
Note that \(\mathcal E_n^{\mathrm{lip}}\) implies \(\|\bar{\mathbf X}_n-\bm{\mu}\|\le \delta\), since $\bar{\mathbf X}_n=\frac1n\sum_{i=1}^n \bar{\mathbf X}_{-i}$. For any prescribed \(\lambda>0\), if \(\mathbb E\|\mathbf X_1\|^l<\infty\) for some sufficiently large \(l=l(\lambda)\), then Markov's inequality, together with a union bound, gives
\[
\mathbb P(\mathcal E_n^{\mathrm{lip}})=1-O(n^{-\lambda}).
\]

\begin{lemma}[Bounds for the jackknife expansion terms]\label{lm:gfunc}
Under Assumption~\ref{assume1}, suppose further that \(\nabla\tilde A\) is \(L\)-Lipschitz on the neighborhood $\{\mathbf x:\|\mathbf x-\bm\mu\|\le \delta\}$. Then, in the decomposition \eqref{eq:jackknife-pivotal-expansion}, for any fixed integer \(r\ge 1\), provided that the moment order \(l\) in Assumption~\ref{assume1} is chosen sufficiently large,
\[
\mathbb E\left[
\ind_{\mathcal E_n^{\mathrm{lip}}}\times\abs{R_{[5]}}^r
\right]
\le
\frac{1}{n^r}\times 
C_{[r,L,\|{\nabla\tilde A(\bm{\mu})}\|,\{\mathbb E\|\mathbf X_1\|^k\}_{k=1,\dots,l}]}.
\]
In particular, the constant depends on $r$, $L$, $\|\nabla \tilde A(\bm\mu)\|$, and finite moments of \(\|\mathbf X_1\|\).
\end{lemma}

\begin{proof}
We first prove the result for \(r=1\) in Steps 1--3, and then extend the argument to \(r\ge 2\) in Step 4. We begin with the identity
\begin{align}\label{Eq lm2start}
W_n-W_{n,-i}
=
\sqrt n\bigl(\tilde A(\bar{\mathbf X}_n)-\tilde A(\bar{\mathbf X}_{-i})\bigr)
+
(\sqrt n-\sqrt{n-1})\tilde A(\bar{\mathbf X}_{-i}).
\end{align}

\medskip
\noindent
\emph{Step 1: The first term.}
On the event \(\mathcal E_n^{\mathrm{lip}}\), we have
\begin{align*}
&\left|
\tilde A(\bar{\mathbf X}_n)-\tilde A(\bar{\mathbf X}_{-i})
-\nabla \tilde A(\bm{\mu})^\top(\bar{\mathbf X}_n-\bar{\mathbf X}_{-i})
\right|
\\
=&
\left|
\int_0^1
\Bigl[
\nabla \tilde A\bigl(\bar{\mathbf X}_{-i}+t(\bar{\mathbf X}_n-\bar{\mathbf X}_{-i})\bigr)
-\nabla \tilde A(\bm{\mu})
\Bigr]^\top
(\bar{\mathbf X}_n-\bar{\mathbf X}_{-i})\,dt
\right|
\\
\le&
\int_0^1
\left\|
\nabla \tilde A\bigl(\bar{\mathbf X}_{-i}+t(\bar{\mathbf X}_n-\bar{\mathbf X}_{-i})\bigr)
-\nabla \tilde A(\bm{\mu})
\right\|
\left\|
\bar{\mathbf X}_n-\bar{\mathbf X}_{-i}
\right\|\,dt
\\
\le&
L\int_0^1
\left\|
\bar{\mathbf X}_{-i}+t(\bar{\mathbf X}_n-\bar{\mathbf X}_{-i})-\bm{\mu}
\right\|
\left\|
\bar{\mathbf X}_n-\bar{\mathbf X}_{-i}
\right\|\,dt
\\
\le&
L\int_0^1
\left\|
(1-t)(\bar{\mathbf X}_{-i}-\bm{\mu})
+t(\bar{\mathbf X}_n-\bm{\mu})
\right\|
\left\|
\bar{\mathbf X}_n-\bar{\mathbf X}_{-i}
\right\|\,dt
\\
\le&
L
\left(
\|\bar{\mathbf X}_n-\bm{\mu}\|
+
\|\bar{\mathbf X}_{-i}-\bm{\mu}\|
\right)
\|\bar{\mathbf X}_n-\bar{\mathbf X}_{-i}\|.
\end{align*}
It follows that
\begin{equation}\label{eq lm2first}
\begin{split}
&\mathbb E\left[
\ind_{\mathcal E_n^{\mathrm{lip}}}
\sqrt n
\left|
\tilde A(\bar{\mathbf X}_n)-\tilde A(\bar{\mathbf X}_{-i})
-\nabla \tilde A(\bm{\mu})^\top(\bar{\mathbf X}_n-\bar{\mathbf X}_{-i})
\right|
\right]
\\
\le&
L\sqrt n
\mathbb E\left[
\left(
\|\bar{\mathbf X}_n-\bm{\mu}\|
+
\|\bar{\mathbf X}_{-i}-\bm{\mu}\|
\right)
\|\bar{\mathbf X}_n-\bar{\mathbf X}_{-i}\|
\right]
\\
\le&
L\sqrt n
\left(
\mathbb E
\left(
\|\bar{\mathbf X}_n-\bm{\mu}\|
+
\|\bar{\mathbf X}_{-i}-\bm{\mu}\|
\right)^2
\right)^{1/2}
\left(
\mathbb E
\|\bar{\mathbf X}_n-\bar{\mathbf X}_{-i}\|^2
\right)^{1/2}
\\
\le&
L\sqrt n
\left(
2\mathbb E\|\bar{\mathbf X}_n-\bm{\mu}\|^2
+
2\mathbb E\|\bar{\mathbf X}_{-i}-\bm{\mu}\|^2
\right)^{1/2}
\left(
\frac{\sigma_2^2}{n(n-1)}
\right)^{1/2},
\qquad
\sigma_2^2\coloneqq \mathbb E\|\mathbf X_1-\bm{\mu}\|^2
\\
\le&
L\sqrt n
\left(
2\sigma_2^2\left(\frac{2n-1}{n(n-1)}\right)
\times
\frac{\sigma_2^2}{n(n-1)}
\right)^{1/2}
\le
\frac{4L}{n}\sigma_2^2.
\end{split}
\end{equation}
Also note that
\begin{align*}
\nabla \tilde A(\bm{\mu})^\top(\bar{\mathbf X}_n-\bar{\mathbf X}_{-i})
=
\frac{1}{n}\nabla \tilde A(\bm{\mu})^\top(\mathbf X_i-\bar{\mathbf X}_{-i})
=
\frac{1}{n}
\left(
\varphi(\mathbf X_i)-\varphi(\bar{\mathbf X}_{-i})
\right).
\end{align*}

\medskip
\noindent
\emph{Step 2: The second term.}
We next consider the second term in \eqref{Eq lm2start}. By the gradient Lipschitz condition,
\begin{equation}\label{eq Liprelax}
\begin{split}
&\mathbb E
\left|
\ind_{\mathcal E_n^{\mathrm{lip}}}
\left(\sqrt n-\sqrt{n-1}\right)
\tilde A(\bar{\mathbf X}_{-i})
\right|
\\
\le&
\frac{2}{\sqrt n}
\mathbb E
\left|
\ind_{\mathcal E_n^{\mathrm{lip}}}
\tilde A(\bar{\mathbf X}_{-i})
\right|
\\
\le&
\frac{2}{\sqrt n}
\left(
\|\nabla \tilde A(\bm{\mu})\|
\mathbb E\|\bar{\mathbf X}_{-i}-\bm{\mu}\|
+
\frac{L}{2}
\mathbb E\|\bar{\mathbf X}_{-i}-\bm{\mu}\|^2
\right)
\\
\le&
\frac{2}{\sqrt n}
\left(
\|\nabla \tilde A(\bm{\mu})\|
\frac{\sigma_2}{\sqrt{n-1}}
+
\frac{L}{2(n-1)}
\sigma_2^2
\right)
\\
\le&
\frac{4}{n}
\left(
\|\nabla \tilde A(\bm{\mu})\|\sigma_2
+
\frac{L}{2\sqrt n}\sigma_2^2
\right).
\end{split}
\end{equation}

\medskip
\noindent
\emph{Step 3: Combining the two terms.}
By definition,
\begin{align*}
R_{[5]}
=
W_n-W_{n,-i}
-
\frac{1}{\sqrt n}
\bigl(
\varphi(\mathbf X_i)-\varphi(\bar{\mathbf X}_{-i})
\bigr).
\end{align*}
Therefore,
\begin{align*}
&\mathbb E
\left|
\ind_{\mathcal E_n^{\mathrm{lip}}}
R_{[5]}
\right|
\\
\le&
\mathbb E
\left|
\ind_{\mathcal E_n^{\mathrm{lip}}}
\sqrt n
\bigl(
\tilde A(\bar{\mathbf X}_n)-\tilde A(\bar{\mathbf X}_{-i})
\bigr)
-
\frac{1}{\sqrt n}
\bigl(
\varphi(\mathbf X_i)-\varphi(\bar{\mathbf X}_{-i})
\bigr)
\right|+
\mathbb E
\left|
\ind_{\mathcal E_n^{\mathrm{lip}}}
(\sqrt n-\sqrt{n-1})
\tilde A(\bar{\mathbf X}_{-i})
\right|
\\
\le&
\frac{4L}{n}\sigma_2^2
+
\frac{4}{n}
\left(
\|\nabla \tilde A(\bm{\mu})\|\sigma_2
+
\frac{L}{2\sqrt n}\sigma_2^2
\right).
\end{align*}
Thus, $\mathbb E
\left[
\ind_{\mathcal E_n^{\mathrm{lip}}}\times
\abs{R_{[5]}}
\right]
\le
\frac{1}{n}
C_{[L,\|\nabla\tilde A(\bm{\mu})\|,\{\mathbb E\|\mathbf X_1\|^k\}_{k=1,\dots,l}]}$. This proves the claim for \(r=1\).

\medskip
\noindent
\emph{Step 4: Higher moments.}
We now extend the preceding argument to \(r\ge 2\). In this case, the estimate in \eqref{eq lm2first} becomes
\begin{equation}
\begin{split}
&\mathbb E\left[
\ind_{\mathcal E_n^{\mathrm{lip}}}
n^{r/2}
\left|
\tilde A(\bar{\mathbf X}_n)-\tilde A(\bar{\mathbf X}_{-i})
-\nabla \tilde A(\bm{\mu})^\top(\bar{\mathbf X}_n-\bar{\mathbf X}_{-i})
\right|^r
\right]
\\
\le&
L^r n^{r/2}
\left(
\mathbb E
\left(
\|\bar{\mathbf X}_n-\bm{\mu}\|
+
\|\bar{\mathbf X}_{-i}-\bm{\mu}\|
\right)^{2r}
\right)^{1/2}
\left(
\mathbb E
\|\bar{\mathbf X}_n-\bar{\mathbf X}_{-i}\|^{2r}
\right)^{1/2}
\\
=&
L^r n^{r/2}
\left(
\mathbb E
\left(
\|\bar{\mathbf X}_n-\bm{\mu}\|
+
\|\bar{\mathbf X}_{-i}-\bm{\mu}\|
\right)^{2r}
\right)^{1/2}
n^{-r}
\left(
\mathbb E
\|\mathbf X_i-\bar{\mathbf X}_{-i}\|^{2r}
\right)^{1/2}
\\
\le&
L^r n^{-r/2}
\left(
\mathbb E
\left(
\|\bar{\mathbf X}_n-\bm{\mu}\|
+
\|\bar{\mathbf X}_{-i}-\bm{\mu}\|
\right)^{2r}
\right)^{1/2}
2^{(2r-1)/2}
\left(
\mathbb E\|\mathbf X_i-\bm{\mu}\|^{2r}
+
\mathbb E\|\bar{\mathbf X}_{-i}-\bm{\mu}\|^{2r}
\right)^{1/2}
\\
\le&
2^{(2r-1)/2}
L^r n^{-r/2}
\left(
\mathbb E
\left(
\|\bar{\mathbf X}_n-\bm{\mu}\|
+
\|\bar{\mathbf X}_{-i}-\bm{\mu}\|
\right)^{2r}
\right)^{1/2}\times
\left(
\mathbb E\|\mathbf X_i-\bm{\mu}\|^{2r}
+
\mathbb E\|\bar{\mathbf X}_{-i}-\bm{\mu}\|^{2r}
\right)^{1/2}
\\
\le&
2^{2r-1}
L^r n^{-r/2}
\left(
\mathbb E\|\bar{\mathbf X}_n-\bm{\mu}\|^{2r}
+
\mathbb E\|\bar{\mathbf X}_{-i}-\bm{\mu}\|^{2r}
\right)^{1/2}\times
\left(
\mathbb E\|\mathbf X_i-\bm{\mu}\|^{2r}
+
\mathbb E\|\bar{\mathbf X}_{-i}-\bm{\mu}\|^{2r}
\right)^{1/2}
\\
\le&
2^{2r-1}
L^r n^{-r/2}
\left(
2C_{[r]}n^{-r}
\left(
\mathbb E\|\mathbf X_1-\bm{\mu}\|^{2r}
+
\mathbb E\|\mathbf X_1-\bm{\mu}\|^{2r}
\right)
\right)^{1/2}
\left(
2\mathbb E\|\mathbf X_1-\bm{\mu}\|^{2r}
\right)^{1/2}
\\
\le&
n^{-r}
C_{[r,L,\|\nabla\tilde A(\bm{\mu})\|,\{\mathbb E\|\mathbf X_1\|^k\}_{k=1,\dots,l}]}.
\end{split}
\end{equation}
Here we have used a standard Rosenthal--Marcinkiewicz--Zygmund type inequality,
\begin{align*}
\mathbb E\|\bar{\mathbf X}_n-\bm{\mu}\|^r
\le
C_{[r]}n^{-r/2}
\mathbb E\|\mathbf X_1-\bm{\mu}\|^r,
\end{align*}
together with $\abs{x+y}^r\le 2^{r-1}\bigl(\abs{x}^r+\abs{y}^r\bigr)$. Similar to \eqref{eq Liprelax}, we have
\begin{equation}\label{eq eq20}
\begin{split}
&\mathbb E
\left|
\ind_{\mathcal E_n^{\mathrm{lip}}}
\left(\sqrt n-\sqrt{n-1}\right)
\tilde A(\bar{\mathbf X}_{-i})
\right|^r
\\
\le&
2^r n^{-r/2}
\mathbb E
\left(
\|\nabla \tilde A(\bm{\mu})\|
\|\bar{\mathbf X}_{-i}-\bm{\mu}\|
+
\frac{L}{2}
\|\bar{\mathbf X}_{-i}-\bm{\mu}\|^2
\right)^r
\\
\le&
2^{2r-1}n^{-r/2}
\left(
\|\nabla \tilde A(\bm{\mu})\|^r
\mathbb E\|\bar{\mathbf X}_{-i}-\bm{\mu}\|^r
+
\frac{L^r}{2^r}
\mathbb E\|\bar{\mathbf X}_{-i}-\bm{\mu}\|^{2r}
\right)
\\
\le&
2^{2r-1}n^{-r/2}
\left(
C_{[r]}(n-1)^{-r/2}
\|\nabla \tilde A(\bm{\mu})\|^r
\mathbb E\|\mathbf X_1-\bm{\mu}\|^r
+
\frac{L^r}{2^r}
C_{[r]}(n-1)^{-r}
\mathbb E\|\mathbf X_1-\bm{\mu}\|^{2r}
\right)
\\
\le&
n^{-r}
C_{[r,L,\|\nabla\tilde A(\bm{\mu})\|,\{\mathbb E\|\mathbf X_1\|^k\}_{k=1,\dots,l}]}.
\end{split}
\end{equation}
Combining the two estimates above gives
\[
\mathbb E\left[
\ind_{\mathcal E_n^{\mathrm{lip}}}\times
\abs{R_{[5]}}^r
\right]
\le
\frac{1}{n^r}
C_{[r,L,\|\nabla\tilde A(\bm{\mu})\|,\{\mathbb E\|\mathbf X_1\|^k\}_{k=1,\dots,l}]}.
\]
This completes the proof.
\end{proof}

\begin{remark}
The indicator can be removed if an additional moment condition is imposed on the pivotal statistic. For instance, in the case $r=1$, assume that there exists $\epsilon>0$ such that
\[
\sup_n
\mathbb E
\left|
\sqrt n\,\tilde A(\bar{\mathbf X}_n)
\right|^{1+\epsilon}
<\infty.
\]
Then H\"older's inequality can be applied directly to the pivotal statistic, without restricting to the event $\mathcal E_n^{\mathrm{lip}}$. However, this condition need not hold in general, as the following example shows.
\end{remark}

\begin{exam}
Consider the one-dimensional case \(d=1\), where \(\mathbf X_i=X_i\sim N(0,1)\). Let
\[
\tilde A(x)\coloneqq x+x^2+x^4\exp(x^4),
\]
which grows explosively for large \(\abs{x}\). The smooth-function model still applies, and the Edgeworth expansion remains valid, but \(R_{[5]}\) is not integrable. Indeed, $(\sqrt n-\sqrt{n-1})\times
\mathbb E
\abs{\tilde A(\bar{\mathbf X}_{-i})}
=
\infty$.
\end{exam}

\subsubsection{Expectation Bounds from Uniform Edgeworth Remainders}

The final auxiliary lemma turns the uniform Edgeworth remainder into a bound for expectations. The argument uses only the uniform remainder bound and integration by parts. In particular, we restrict attention to absolutely continuous functions satisfying $f'\in L^1(\mathbb R)$, and we do not attempt to handle functions with polynomial growth. Thus, the result is considerably weaker than the classical Edgeworth expansion results for expectations, such as Theorem~20.1 in \cite{bhattacharya2010normal} and the subsequent corollaries. Nevertheless, this restricted form is sufficient for our purpose.

\begin{lemma}\label{lm lemma3 expectation}
Let $W_n \coloneqq \sqrt{n}\widetilde A(\bar{\mathbf X})$ be the pivotal statistic under Assumption \ref{assume1}. Define
\[
F_n(x)\coloneqq \mathbb P_n(W_n\le x),
\qquad
G_n^\nu(x)\coloneqq \Phi(x)+\sum_{j=1}^{\nu} n^{-j/2}\pi_j(x)\phi(x),
\]
and let $\varepsilon_n^\nu(x)\coloneqq F_n(x)-G_n^\nu(x)$. Let $\lambda$ denote the Lebesgue measure. Suppose that $f$ is absolutely continuous with respect to $\lambda$ and satisfies $f'\in L^1(\lambda)$. Let $Z\sim N(0,1)$. Then
\begin{align*}
   \left|
   \EE[f(W_n)]- \EE[f(Z)]- \sum_{j=1}^{\nu} n^{-j/2}\int f(x) d[\pi_j(x)\phi(x)]
   \right|
   \le
   \mathsf{C}_\nu \|f'\|_{L^1(\mathbb R)} n^{-(\nu+1)/2},
\end{align*}
where $\mathsf{C}_\nu$ is defined in \eqref{eq theorem1 rate}.
\end{lemma}

\begin{proof} Since \(F_n\), \(\Phi\), and each \(\pi_j(x)\phi(x)\) correspond to probability measures or signed measures of finite total variation, the remainder \(\varepsilon_n\), as their finite linear combination, is itself of bounded variation. Therefore, the following Lebesgue--Stieltjes integral is well-defined:
\begin{align*}
    \int_{\mathbb R} f(x)\,d\varepsilon_n^\nu(x)
    =
    -
    \int_{\mathbb R} f'(x)\varepsilon_n^\nu(x)\,dx,
\end{align*}
where we used the integration-by-parts, together with the facts that \(f\) is absolutely continuous, \(f'\in L^1(\mathbb R)\), and \(\varepsilon_n^\nu(\pm\infty)=0\). Hence
\begin{align*}
    \left|
\int_{\mathbb R} f(x)\,d\varepsilon_n^\nu(x)
\right|
\le
\int_{\mathbb R}\abs{f'(x)}\times\abs{\varepsilon_n^\nu(x)}\,dx
\le
\mathsf{C}_\nu\times\|f'\|_{L^1(\mathbb R)}\mathsf \times n^{-(\nu+1)/2}.
\end{align*}
\end{proof}

\subsubsection{Proof of Proposition~\ref{theo:expect-cancellation-good}: Expectation Cancellation of Monomial Moment Errors under an Even Weight}

We prove Proposition~\ref{theo:expect-cancellation-good} by establishing a slightly more quantitative bound, in which the dependence of the implicit constant is made explicit. Since \(g\), \(g'\), and \(g''\) are uniformly bounded and \(g''\in L^1(\mathbb R)\), there exists a finite constant \(C_{[g]}\) such that
\begin{align}\label{eq cgbound}
    \sup_{x\in\mathbb R}\left\{\abs{g(x)}+\abs{g'(x)}+\abs{g''(x)}\right\}
    +
    \|g''\|_{L^1(\mathbb R)}
    \le C_{[g]}.
\end{align}
We also fix the local constants used in the jackknife expansion of the pivotal statistic. Under Assumption~\ref{assume1}, \(\tilde A\) is \(C^2\) in a neighborhood of \(\bm\mu\). Hence there exist constants \(\delta>0\) and \(L<\infty\) such that
\[
\|\nabla \tilde A(\mathbf u)-\nabla \tilde A(\mathbf v)\|
\le
L\|\mathbf u-\mathbf v\|
\]
whenever $\|\mathbf u-\bm\mu\|\le \delta$ and $\|\mathbf v-\bm\mu\|\le \delta$. Let \(C_M\) denote the monomial constant from Lemma~\ref{lm linearization}, associated with the monomial \(M\) and its plug-in estimator \(M_n\). We use \(C_M\) exactly as provided by Lemma~\ref{lm linearization}; no further expansion of this constant is needed here. It suffices to prove the following explicit bound:
\begin{align}\label{eq cancel}
\left|
\mathbb E\left[g(W_n)\times(M_n-M)\right]
\right|
\le
\frac{1}{n}\times 
C_{[\mathsf C_1,\delta,L,\|\nabla\tilde A(\bm{\mu})\|,C_M,C_{[g]},\{\mathbb E\|\mathbf X_1\|^k\}_{k=1,\ldots,l}]},
\end{align}
where the constant on the right-hand side depends on the first-order Edgeworth remainder constant \(\mathsf C_1\) from Theorem~\ref{Theo Hall1}, the local neighborhood radius \(\delta\) and the local Lipschitz constant \(L\) from Lemma~\ref{lm:gfunc}, the gradient norm \(\|\nabla\tilde A(\bm{\mu})\|\), the monomial constant \(C_M\) from Lemma~\ref{lm linearization}, the regularity constant \(C_{[g]}\) associated with \(g\), and finitely many moments \(\{\mathbb E\|\mathbf X_1\|^k\}_{k=1,\ldots,l}\).
\begin{proof}
By Lemma~\ref{lm linearization},
\[
M_n-M
=
\frac1n\sum_{i=1}^n \varrho_M(\mathbf X_i)+R_{[4]},
\qquad
\mathbb E[\varrho_M(\mathbf X_1)]=0,
\qquad
\mathbb E\abs{R_{[4]}}
\le
\frac{C_M}{n}
\left(1+\mathbb E\|\mathbf X_1\|^\ell\right)^{d_H-1}.
\]
Therefore,
\begin{align*}
\mathbb E\left[g(W_n)(M_n-M)\right]
&=
\mathbb E\left[g(W_n)\varrho_M(\mathbf X_1)\right]
+
\mathbb E\left[g(W_n)R_{[4]}\right]  \\
&=
\mathbb E\left[
\bigl(g(W_n)-g(W_{n,-1})\bigr)\varrho_M(\mathbf X_1)
\right]
+
\mathbb E\left[
g(W_{n,-1})\varrho_M(\mathbf X_1)
\right]
+
\mathbb E\left[
g(W_n)R_{[4]}
\right].
\end{align*}
The middle term vanishes because \(W_{n,-1}\) is independent of \(\mathbf X_1\) and
\(\mathbb E[\varrho_M(\mathbf X_1)]=0\). Moreover, using \(\abs{g}\le C_{[g]}\) and Lemma~\ref{lm linearization},
\[
\left|
\mathbb E\left[g(W_n)R_{[4]}\right]
\right|
\le
C_{[g]}\mathbb E\abs{R_{[4]}}
\le
\frac{1}{n}
C_{[C_M,C_{[g]},\{\mathbb E\|\mathbf X_1\|^k\}_{k=1,\ldots,l}]}.
\]
It remains to bound
\[
J\coloneqq
\mathbb E\left[
\bigl(g(W_n)-g(W_{n,-1})\bigr)\varrho_M(\mathbf X_1)
\right].
\]
Let
\[
D_{n,1}\coloneqq W_n-W_{n,-1}.
\]
By the jackknife expansion \eqref{eq:jackknife-pivotal-expansion},
\[
D_{n,1}
=
\frac{1}{\sqrt n}
\bigl(\varphi(\mathbf X_1)-\varphi(\bar{\mathbf X}_{-1})\bigr)
+
R_{[5]}.
\]
Taylor's theorem gives
\[
g(W_n)-g(W_{n,-1})
=
g'(W_{n,-1})D_{n,1}
+
\frac12 g''(\xi_{n,1})D_{n,1}^2,
\]
where \(\xi_{n,1}\) is an intermediate point between \(W_n\) and \(W_{n,-1}\). Hence
\begin{align*}
g(W_n)-g(W_{n,-1})
=&
\frac{1}{\sqrt n}
g'(W_{n,-1})\varphi(\mathbf X_1)
-
\frac{1}{\sqrt n}
g'(W_{n,-1})\varphi(\bar{\mathbf X}_{-1})
+
g'(W_{n,-1})R_{[5]}
+
\frac12 g''(\xi_{n,1})D_{n,1}^2.
\end{align*}
Let \(\mathcal E_n^{\mathrm{lip}}\) be the event in \eqref{Eq goodevent Lipschitz}. Choose \(\lambda\ge2\). By choosing the moment order \(l\) sufficiently large,
\[
\mathbb P(\mathcal E_n^{\mathrm{lip}})=1-O(n^{-\lambda}).
\]
We decompose \(J\) according to this event and the four terms in the Taylor expansion:
\[
J=\mathcal T_1+\mathcal T_2+\mathcal T_3+\mathcal T_4+\mathcal T_5,
\]
where
\begin{align*}
\mathcal T_1
&\coloneqq
\frac{1}{\sqrt n}
\mathbb E\left[
\ind_{\mathcal E_n^{\mathrm{lip}}}
\varrho_M(\mathbf X_1)
g'(W_{n,-1})
\varphi(\mathbf X_1)
\right], \\
\mathcal T_2
&\coloneqq
-\frac{1}{\sqrt n}
\mathbb E\left[
\ind_{\mathcal E_n^{\mathrm{lip}}}
\varrho_M(\mathbf X_1)
g'(W_{n,-1})
\varphi(\bar{\mathbf X}_{-1})
\right], \\
\mathcal T_3
&\coloneqq
\mathbb E\left[
\ind_{\mathcal E_n^{\mathrm{lip}}}
\varrho_M(\mathbf X_1)
g'(W_{n,-1})
R_{[5]}
\right], \\
\mathcal T_4
&\coloneqq
\frac12
\mathbb E\left[
\ind_{\mathcal E_n^{\mathrm{lip}}}
\varrho_M(\mathbf X_1)
g''(\xi_{n,1})
D_{n,1}^2
\right], \\
\mathcal T_5
&\coloneqq
\mathbb E\left[
\left(1-\ind_{\mathcal E_n^{\mathrm{lip}}}\right)
\bigl(g(W_n)-g(W_{n,-1})\bigr)
\varrho_M(\mathbf X_1)
\right].
\end{align*}
We first control the contribution from the complement of the good event. Since \(\abs{g}\le C_{[g]}\),
\begin{align*}
\abs{\mathcal T_5}
&\le
2C_{[g]}
\mathbb E\left[
\left(1-\ind_{\mathcal E_n^{\mathrm{lip}}}\right)
\abs{\varrho_M(\mathbf X_1)}
\right] \\
&\le
2C_{[g]}
\sqrt{1-\mathbb P(\mathcal E_n^{\mathrm{lip}})}
\left(\mathbb E\abs{\varrho_M(\mathbf X_1)}^2\right)^{1/2} \\
&\le
\frac{1}{n}
C_{[\delta,C_M,C_{[g]},\{\mathbb E\|\mathbf X_1\|^k\}_{k=1,\ldots,l}]}.
\end{align*}
We next handle the term involving \(\varphi(\bar{\mathbf X}_{-1})\). Since \(\mathbf X_1\) is independent of \(W_{n,-1}\) and \(\bar{\mathbf X}_{-1}\), and since \(\mathbb E[\varrho_M(\mathbf X_1)]=0\),
\[
\mathbb E\left[
\varrho_M(\mathbf X_1)
g'(W_{n,-1})
\varphi(\bar{\mathbf X}_{-1})
\right]
=0.
\]
Therefore,
\begin{align*}
\abs{\mathcal T_2}
&\le
\frac{1}{\sqrt n}
\mathbb E\left[
\left(1-\ind_{\mathcal E_n^{\mathrm{lip}}}\right)
\left|
\varrho_M(\mathbf X_1)
g'(W_{n,-1})
\varphi(\bar{\mathbf X}_{-1})
\right|
\right] \\
&\le
\frac{C_{[g]}}{\sqrt n}
\sqrt{1-\mathbb P(\mathcal E_n^{\mathrm{lip}})}
\left(\mathbb E\abs{\varrho_M(\mathbf X_1)}^4\right)^{1/4}
\left(\mathbb E\abs{\varphi(\bar{\mathbf X}_{-1})}^4\right)^{1/4} \\
&\le
\frac{C_{[g]}}{\sqrt n}
\sqrt{1-\mathbb P(\mathcal E_n^{\mathrm{lip}})}
\left(\mathbb E\abs{\varrho_M(\mathbf X_1)}^4\right)^{1/4}
\|\nabla\tilde A(\bm\mu)\|
\left(\mathbb E\|\bar{\mathbf X}_{-1}-\bm\mu\|^4\right)^{1/4} \\
&\le
\frac{1}{n}
C_{[\delta,\|\nabla\tilde A(\bm{\mu})\|,C_M,C_{[g]},\{\mathbb E\|\mathbf X_1\|^k\}_{k=1,\ldots,l}]}.
\end{align*}
For \(\mathcal T_3\), H\"older's inequality gives
\begin{align*}
\abs{\mathcal T_3}
&\le
C_{[g]}
\left(\mathbb E\abs{\varrho_M(\mathbf X_1)}^2\right)^{1/2}
\left(
\mathbb E\left[
\ind_{\mathcal E_n^{\mathrm{lip}}}\abs{R_{[5]}}^2
\right]
\right)^{1/2} \le
\frac{1}{n}
C_{[L,\|\nabla\tilde A(\bm{\mu})\|,C_M,C_{[g]},\{\mathbb E\|\mathbf X_1\|^k\}_{k=1,\ldots,l}]}.
\end{align*}
For \(\mathcal T_4\), using \(\abs{g''}\le C_{[g]}\),
\begin{align*}
\abs{\mathcal T_4}
&\le
C_{[g]}
\left(\mathbb E\abs{\varrho_M(\mathbf X_1)}^2\right)^{1/2}
\left(
\mathbb E\left[
\ind_{\mathcal E_n^{\mathrm{lip}}}D_{n,1}^4
\right]
\right)^{1/2}.
\end{align*}
From the decomposition of \(D_{n,1}\),
\begin{align*}
\mathbb E\left[
\ind_{\mathcal E_n^{\mathrm{lip}}}D_{n,1}^4
\right]
&\le
\frac{8}{n^2}
\mathbb E\left[
\left|\varphi(\mathbf X_1)-\varphi(\bar{\mathbf X}_{-1})\right|^4
\right]
+
8
\mathbb E\left[
\ind_{\mathcal E_n^{\mathrm{lip}}}\abs{R_{[5]}}^4
\right] \le
\frac{1}{n^2}
C_{[L,\|\nabla\tilde A(\bm{\mu})\|,\{\mathbb E\|\mathbf X_1\|^k\}_{k=1,\ldots,l}]}.
\end{align*}
Consequently,
\[
\abs{\mathcal T_4}
\le
\frac{1}{n}
C_{[L,\|\nabla\tilde A(\bm{\mu})\|,C_M,C_{[g]},\{\mathbb E\|\mathbf X_1\|^k\}_{k=1,\ldots,l}]}.
\]
It remains to bound the leading term \(\mathcal T_1\). We write it as the full independent product minus its bad-event correction:
\begin{align*}
\mathcal T_1
=&
\frac{1}{\sqrt n}
\mathbb E\left[
\varrho_M(\mathbf X_1)
g'(W_{n,-1})
\varphi(\mathbf X_1)
\right] -
\frac{1}{\sqrt n}
\mathbb E\left[
\left(1-\ind_{\mathcal E_n^{\mathrm{lip}}}\right)
\varrho_M(\mathbf X_1)
g'(W_{n,-1})
\varphi(\mathbf X_1)
\right].
\end{align*}
The bad-event correction is bounded by
\begin{align*}
&\frac{1}{\sqrt n}
\mathbb E\left[
\left(1-\ind_{\mathcal E_n^{\mathrm{lip}}}\right)
\left|
\varrho_M(\mathbf X_1)
g'(W_{n,-1})
\varphi(\mathbf X_1)
\right|
\right] \\
&\qquad\le
\frac{C_{[g]}}{\sqrt n}
\sqrt{1-\mathbb P(\mathcal E_n^{\mathrm{lip}})}
\left(\mathbb E\abs{\varrho_M(\mathbf X_1)}^4\right)^{1/4}
\left(\mathbb E\abs{\varphi(\mathbf X_1)}^4\right)^{1/4} \\
&\qquad\le
\frac{1}{n}\times 
C_{[\delta,\|\nabla\tilde A(\bm{\mu})\|,C_M,C_{[g]},\{\mathbb E\|\mathbf X_1\|^k\}_{k=1,\ldots,l}]}.
\end{align*}
For the full expectation, independence between \(W_{n,-1}\) and \(\mathbf X_1\) gives
\[
\mathbb E\left[
\varrho_M(\mathbf X_1)
g'(W_{n,-1})
\varphi(\mathbf X_1)
\right]
=
\mathbb E[g'(W_{n,-1})]
\mathbb E[\varrho_M(\mathbf X_1)\varphi(\mathbf X_1)].
\]
Since \(\varrho_M\) is a polynomial and \(\varphi\) is linear,
\[
\left|
\mathbb E[\varrho_M(\mathbf X_1)\varphi(\mathbf X_1)]
\right|
\le
C_{[\|\nabla\tilde A(\bm{\mu})\|,C_M,\{\mathbb E\|\mathbf X_1\|^k\}_{k=1,\ldots,l}]}.
\]
It remains only to bound \(\mathbb E[g'(W_{n,-1})]\). The random variable \(W_{n,-1}\) has the same law as the pivotal statistic constructed from a sample of size \(n-1\). Since \(g\) is even, \(g'\) is odd. Moreover, \(g'\) is absolutely continuous and \((g')'=g''\in L^1(\mathbb R)\). Applying Lemma~\ref{lm lemma3 expectation} to \(f=g'\), with sample size \(n-1\) and \(\nu=1\), yields
\begin{align*}
\mathbb E[g'(W_{n,-1})]
=&
\int_{\mathbb R} g'(x)\,d\Phi(x)
+
\frac{1}{\sqrt{n-1}}
\int_{\mathbb R} g'(x)\,d[\pi_1(x)\phi(x)] +
\int_{\mathbb R} g'(x)\,d\varepsilon_{n-1}^1(x).
\end{align*}
The first term is zero because \(g'\) is odd. The second term is of order \((n-1)^{-1/2}\), with its constant absorbed into
\[
C_{[\mathsf C_1,C_{[g]},\{\mathbb E\|\mathbf X_1\|^k\}_{k=1,\ldots,l}]}.
\]
Here \(\mathsf C_1\) denotes the first-order Edgeworth remainder constant from Theorem~\ref{Theo Hall1}. For the last term, Lemma~\ref{lm lemma3 expectation} and the definition of \(C_{[g]}\) give
\[
\left|
\int_{\mathbb R} g'(x)\,d\varepsilon_{n-1}^1(x)
\right|
\le
\mathsf C_1\|g''\|_{L^1(\mathbb R)}(n-1)^{-1}
\le
\mathsf C_1 C_{[g]}(n-1)^{-1}.
\]
Thus
\[
\left|
\mathbb E[g'(W_{n,-1})]
\right|
\le
\frac{1}{\sqrt n}
C_{[\mathsf C_1,C_{[g]},\{\mathbb E\|\mathbf X_1\|^k\}_{k=1,\ldots,l}]},
\]
and therefore
\[
\abs{\mathcal T_1}
\le
\frac{1}{n}
C_{[\mathsf C_1,\delta,\|\nabla\tilde A(\bm{\mu})\|,C_M,C_{[g]},\{\mathbb E\|\mathbf X_1\|^k\}_{k=1,\ldots,l}]}.
\]
Combining the bounds for \(\mathcal T_1,...,\mathcal T_5\), we obtain
\[
\abs{J}
\le
\frac{1}{n}
C_{[\mathsf C_1,\delta,L,\|\nabla\tilde A(\bm{\mu})\|,C_M,C_{[g]},\{\mathbb E\|\mathbf X_1\|^k\}_{k=1,\ldots,l}]}.
\]
Together with the bounds for the two preliminary terms, this gives
\begin{align*}
\left|
\mathbb E\left[g(W_n)(M_n-M)\right]
\right|
&\le
\left|
\mathbb E\left[
\bigl(g(W_n)-g(W_{n,-1})\bigr)\varrho_M(\mathbf X_1)
\right]
\right|
+
\left|
\mathbb E\left[
g(W_{n,-1})\varrho_M(\mathbf X_1)
\right]
\right| +
\left|
\mathbb E\left[
g(W_n)R_{[4]}
\right]
\right| \\
&\le
\frac{1}{n}
C_{[\mathsf C_1,\delta,L,\|\nabla\tilde A(\bm{\mu})\|,C_M,C_{[g]},\{\mathbb E\|\mathbf X_1\|^k\}_{k=1,\ldots,l}]}.
\end{align*}
This proves \eqref{eq cancel}, and hence Proposition~\ref{theo:expect-cancellation-good}.
\end{proof}

\subsection{Resampled Analogues of Expectation Cancellation Proposition and its Three Auxiliary Lemmas}\label{sec:resampled-three-lemmas}
We now formulate the conditional versions of the three auxiliary lemmas used in Section \ref{secprop1}. The purpose is to make explicit that, on suitable empirical good events, all constants appearing in the resampled arguments can be bounded by deterministic constants depending only on finitely many \emph{population} norm moments rather than the empirical norm moments, up to arbitrary fixed tolerances.

We now formulate two good events. First, fix a sufficiently large integer $l\ge 1$ and a tolerance $\varepsilon_{\mathrm m}>0$. Define the empirical norm-moment event
\begin{align}\label{good event emp moment}
\mathcal E_{\mathrm m}
=
\left\{
\frac1n\sum_{i=1}^n\|\mathbf X_i\|^k
\le
\mathbb E\|\mathbf X_1\|^k+\varepsilon_{\mathrm m},
\quad k=1,\ldots,l
\right\}.
\end{align}
This event requires the empirical norm moments up to order $l$ to be bounded above by their population counterparts, up to the fixed tolerance $\varepsilon_{\mathrm m}$. Second, recall the notations from Lemma~\ref{lm:gfunc}. For a fixed \(\varepsilon_\mu>0\), define
\begin{align}\label{good event mean}
\mathcal E_{\bm\mu}
\coloneqq
\left\{
\|\hat{\bm\mu}_n-\bm\mu\|\le \varepsilon_\mu
\right\},\quad \text{where}\quad \hat{\bm\mu}_n
\coloneqq
\bar{\mathbf X}_n
=
\frac1n\sum_{i=1}^n\mathbf X_i.
\end{align}
Both $\mathcal E_{\mathrm m}$ and $\mathcal E_{\bm\mu}$ are measurable with respect to $\sigma(\mathcal X)$. Moreover, their complements have arbitrarily high polynomial decay, provided sufficiently high moments of $\|\mathbf X_1\|$ exist, as stated in the next two lemmas.
\begin{lemma}\label{emevent1}
Fix \(l\ge1\) and \(\varepsilon_{\mathrm m}>0\). For any prescribed \(\lambda>0\), if the moment order in Assumption~\ref{assume1} is chosen sufficiently large, then $\mathbb P\bigl(\mathcal E_{\mathrm m}\bigr)
=
1-O(n^{-\lambda})$.
\end{lemma}

\begin{proof}
Fix \(k\le l\) and set $Y_i^{(k)}
=
\|\mathbf X_i\|^k-\mathbb E\|\mathbf X_1\|^k$. For every integer \(p\ge2\), Markov's inequality and the Marcinkiewicz--Zygmund inequality imply
\[
\begin{aligned}
\mathbb P\left(
\frac1n\sum_{i=1}^n\|\mathbf X_i\|^k
>
\mathbb E\|\mathbf X_1\|^k+\varepsilon_{\mathrm m}
\right)
\le
\varepsilon_{\mathrm m}^{-p}
\mathbb E\left|
\frac1n\sum_{i=1}^nY_i^{(k)}
\right|^p  \le
\frac{C_p}{\varepsilon_{\mathrm m}^{p}n^{p/2}}
\mathbb E\abs{Y_1^{(k)}}^p \le
\frac{C_{[p,k,\varepsilon_{\mathrm m}]}}{n^{p/2}}
\left(1+\mathbb E\|\mathbf X_1\|^{kp}\right).
\end{aligned}
\]
Taking a union bound over \(k=1,\ldots,l\), we obtain
\[
\mathbb P\bigl(\mathcal E_{\mathrm m}^c\bigr)
\le
\frac{1
}{n^{p/2}}\times C_{[p,l,\varepsilon_{\mathrm m},\{\mathbb E\|\mathbf X_1\|^j\}_{j=1,\ldots,lp}]}.
\]
Choosing \(p\) such that \(p/2\ge\lambda\), and assuming the moment order in Assumption~\ref{assume1} is large enough to ensure \(\mathbb E\|\mathbf X_1\|^{lp}<\infty\), gives $\mathbb P\bigl(\mathcal E_{\mathrm m}^c\bigr)
=
O(n^{-\lambda})$ as claimed.
\end{proof}
\begin{lemma}
Fix \(\varepsilon_\mu>0\). For any prescribed \(\lambda>0\), if the moment order in Assumption~\ref{assume1} is chosen sufficiently large, then $\mathbb P\bigl(\mathcal E_{\bm\mu}\bigr)
=
1-O(n^{-\lambda})$.
\end{lemma}
\begin{proof}
This is a simpler case of Lemma~\ref{emevent1}; hence the proof is omitted.
\end{proof}
Recall that the sample is denoted by
$\mathcal X=\{\mathbf X_1,\ldots,\mathbf X_n\}$, and let
\(\mathcal X^*=\{\mathbf X_1^*,\ldots,\mathbf X_n^*\}\) be i.i.d.\ draws from the empirical distribution of \(\mathcal X\). In the remaining part of this section, we write
\[
\hat{\mathbb E}[\cdot]\coloneqq \mathbb E[\cdot\mid \mathcal X],
\qquad
\hat{\mathbb P}(\cdot)\coloneqq \mathbb P(\cdot\mid \mathcal X).
\]

\subsubsection{Lemma~\ref{lm linearization} revisited}
Recall the notation from the subsection of Lemma~\ref{lm linearization}. Analogously, define
\[
\hat m_{n,\alpha}^*\coloneqq \frac1n\sum_{i=1}^n(\mathbf X_i^*)^\alpha,
\qquad
\hat{\bm m}_n^*\coloneqq(\hat m_{n,\alpha}^*)_{\alpha\in\mathcal A},
\qquad
M_n^*\coloneqq H(\hat{\bm m}_n^*).
\]
Also define
\[
\varrho_{M_n}(\mathbf x)
\coloneqq
\sum_{\alpha\in\mathcal A}
\partial_\alpha H(\hat{\bm m}_n)\times
\left(\mathbf x^\alpha-\hat m_{n,\alpha}\right).
\]
Then, analogously to the population version, we have $\hat{\mathbb E}[\varrho_{M_n}(\mathbf X_1^*)]=0$.

\begin{lemma}[Resampled version of Lemma~\ref{lm linearization}]\label{resample lemma1}
For a given monomial \(M\) as defined in \eqref{eq MMM}, define \(\hat R_{[4]}\) by
\[
M_n^*-M_n
=
\frac1n\sum_{i=1}^n
\varrho_{M_n}(\mathbf X_i^*)
+
\hat R_{[4]}.
\]
For every fixed integer \(p\ge1\), if \(l\) is chosen sufficiently large depending on \(M\) and \(p\), then on the event \(\mathcal E_{\mathrm m}\),
\begin{align}\label{lm6 eqa}
    \ind_{\mathcal E_{\mathrm m}}\times
\hat{\mathbb E}[\abs{\hat R_{[4]}}]
\le
\frac{1}{n}
C_{[M,\{\mathbb E\|\mathbf X_1\|^k\}_{k=1,\ldots,l}]}.
\end{align}
Moreover,
\begin{align}\label{lm6 eqb}
    \ind_{\mathcal E_{\mathrm m}}\times
\hat{\mathbb E}
\left[
\abs{\varrho_{M_n}(\mathbf X_1^*)}^p
\right]
\le
C_{[p,M,\{\mathbb E\|\mathbf X_1\|^k\}_{k=1,\ldots,l}]}.
\end{align}
\end{lemma}
\begin{proof}
Conditional on \(\mathcal X\), Lemma~\ref{lm linearization} applied to the empirical distribution \(\hat P_n\) gives the same remainder bound with constant
\(C_{[M,\{\hat{\mathbb E}\|\mathbf X_1^*\|^k\}_{k=1,\ldots,l}]}\).
On \(\mathcal E_{\mathrm m}\), this constant is bounded by
\(C_{[M,\{\mathbb E\|\mathbf X_1\|^k\}_{k=1,\ldots,l}]}\), up to absorbing the fixed tolerance \(\varepsilon_{\mathrm m}\), and hence \eqref{lm6 eqa} follows. It remains to prove \eqref{lm6 eqb}. By definition,
\[
\varrho_{M_n}(\mathbf x)
=
\sum_{\alpha\in\mathcal A}
\partial_\alpha H(\hat{\bm m}_n)
\left(\mathbf x^\alpha-\hat m_{n,\alpha}\right).
\]
Since \(H\) is a fixed monomial, its derivatives are polynomial functions of their arguments. Therefore, on \(\mathcal E_{\mathrm m}\), the coefficients
\(\partial_\alpha H(\hat{\bm m}_n)\) are bounded by deterministic constants depending only on \(M\) and the population norm moments up to order \(l\). Moreover, for each \(\alpha\in\mathcal A\), if \(l\) is chosen sufficiently large depending on \(M\) and \(p\), then the conditional \(p\)-moment of
\((\mathbf X_1^*)^\alpha-\hat m_{n,\alpha}\) is bounded by the empirical norm moments controlled on \(\mathcal E_{\mathrm m}\). Since \(\mathcal A\) is finite, combining these coefficient bounds and moment bounds proves \eqref{lm6 eqb}.
\end{proof}

\subsubsection{Lemma~\ref{lm:gfunc} revisited}

We next formulate the resampled version of Lemma~\ref{lm:gfunc}. Recall from \eqref{good event mean} that \(\mathcal E_{\bm\mu}=\{\|\hat{\bm\mu}_n-\bm\mu\|\le \varepsilon_\mu\}\), where \(\hat{\bm\mu}_n=\bar{\mathbf X}_n\). In this subsection, we set \(\varepsilon_\mu=\delta/2\), where \(\delta>0\) is chosen so that \(\nabla\tilde A\) is \(L\)-Lipschitz on \(\{\mathbf x:\|\mathbf x-\bm\mu\|\le\delta\}\). Conditional on \(\mathcal X\), let \(\bar{\mathbf X}_n^*=n^{-1}\sum_{j=1}^n\mathbf X_j^*\) and \(\bar{\mathbf X}_{-i}^*=(n-1)^{-1}\sum_{j\ne i}\mathbf X_j^*\). Define
\[
W_n^*
=
\sqrt n\{\tilde A(\bar{\mathbf X}_n^*)-\tilde A(\hat{\bm\mu}_n)\},
\qquad
W_{n,-i}^*
=
\sqrt{n-1}\{\tilde A(\bar{\mathbf X}_{-i}^*)-\tilde A(\hat{\bm\mu}_n)\},
\]
and set \(\hat\varphi(\mathbf x)=\nabla\tilde A(\hat{\bm\mu}_n)^\top(\mathbf x-\hat{\bm\mu}_n)\). We also define the resampled Lipschitz event
\begin{align}\label{Eq newLipschitz}
\mathcal E_n^{*,\mathrm{lip}}
\coloneqq
\left\{
\|\bar{\mathbf X}_{-i}^*-\hat{\bm\mu}_n\|
\le
\delta/2
\text{ for all }1\le i\le n
\right\}.
\end{align}
On \(\mathcal E_{\bm\mu}\cap\mathcal E_n^{*,\mathrm{lip}}\), all \(\bar{\mathbf X}_{-i}^*\) lie in the \(\delta\)-neighborhood of \(\bm\mu\). Since \(\bar{\mathbf X}_n^*=n^{-1}\sum_{i=1}^n\bar{\mathbf X}_{-i}^*\), the same is true for \(\bar{\mathbf X}_n^*\). We emphasize that, unlike \(\mathcal E_{\bm\mu}\) and \(\mathcal E_{\mathrm m}\), the event \(\mathcal E_n^{*,\mathrm{lip}}\) is not measurable with respect to the original sample \(\mathcal X\) alone; it is measurable with respect to the joint randomness of \((\mathcal X,\mathcal X^*)\).

\begin{lemma}[Resampled version of Lemma~\ref{lm:gfunc}]\label{lm:gfunc resampled}
Define \(\hat R_{[5]}\) through the decomposition
\begin{align}\label{eq:resampled-jackknife-pivotal-expansion}
W_n^*-W_{n,-1}^*
=
\frac{1}{\sqrt n}
\left(
\hat\varphi(\mathbf X_1^*)-\hat\varphi(\bar{\mathbf X}_{-1}^*)
\right)
+
\hat R_{[5]}.
\end{align}
For every fixed integer \(r\ge1\), if \(l\) in Assumption~\ref{assume1} is chosen sufficiently large depending on \(r\), then
\begin{align}\label{eq resampled lm gfunc}
\ind_{\mathcal E_{\bm\mu}\cap\mathcal E_{\mathrm m}}\times 
\hat{\mathbb E}
\left[
\ind_{\mathcal E_n^{*,\mathrm{lip}}}\times 
\abs{\hat R_{[5]}}^r
\right]
\le
\frac{1}{n^r}
C_{[r,L,\|\nabla\tilde A(\bm\mu)\|,\{\mathbb E\|\mathbf X_1\|^k\}_{k=1,\ldots,l}]}.
\end{align}
\end{lemma}

\begin{proof}
Conditional on \(\mathcal X\), Lemma~\ref{lm:gfunc} applied to the empirical distribution \(\hat P_n\) gives the same bound with constant 
\[
C_{[r,L,\|\nabla\tilde A(\hat{\bm\mu}_n)\|,
\{\hat{\mathbb E}\|\mathbf X_1^*\|^k\}_{k=1,\ldots,l}]}.
\]
On \(\mathcal E_{\bm\mu}\cap\mathcal E_n^{*,\mathrm{lip}}\), the points appearing in the decomposition \eqref{eq:resampled-jackknife-pivotal-expansion} lie in the \(\delta\)-neighborhood of \(\bm\mu\), so the same Lipschitz constant \(L\) applies. On \(\mathcal E_{\bm\mu}\), since \(\varepsilon_\mu=\delta/2\),
\[
\|\nabla\tilde A(\hat{\bm\mu}_n)\|
\le
\|\nabla\tilde A(\bm\mu)\|+L\delta/2.
\]
On \(\mathcal E_{\mathrm m}\), the empirical moments
\(\hat{\mathbb E}\|\mathbf X_1^*\|^k=n^{-1}\sum_{i=1}^n\|\mathbf X_i\|^k\)
are bounded by the corresponding population moments up to the fixed tolerance. Absorbing fixed tolerances into the constant gives \eqref{eq resampled lm gfunc}.
\end{proof}

\subsubsection{Lemma~\ref{lm lemma3 expectation} revisited}
Now we discuss the conditional analogue of Lemma~\ref{lm lemma3 expectation}. Let \(F_n^*(x)\coloneqq \hat{\mathbb P}(W_n^*\le x)\), let
\[
G_n^{*,\nu}(x)\coloneqq \Phi(x)+\sum_{j=1}^{\nu}n^{-j/2}\hat\pi_j(x)\phi(x),
\qquad
\hat\varepsilon_n^\nu(x)\coloneqq F_n^*(x)-G_n^{*,\nu}(x),
\]
and recall the event \(\hat{\mathcal E}_{\nu}\) defined via \eqref{eq goodevent0}, \eqref{eq goodevent1} and \eqref{eq goodevent2}, on which
\[
\sup_{x\in\mathbb R}\abs{\hat\varepsilon_n^\nu(x)}
\le
\widehat{\mathsf C}_{\nu}n^{-(\nu+1)/2}.
\]
Here \(\widehat{\mathsf C}_{\nu}\) denotes the deterministic envelope constant appearing in the definition of \(\hat{\mathcal E}_{\nu}\).

\begin{lemma}[Resampled expectation bound from the conditional Edgeworth remainder]\label{lm:resampled-lemma3-expectation}
Suppose the conditional Edgeworth expansion in Theorem~\ref{Theo Hall2} holds on \(\hat{\mathcal E}_{\nu}\). If \(f\) is absolutely continuous with respect to the Lebesgue measure and \(f'\in L^1(\mathbb R)\), then
\begin{align}\label{eq resampled expectation remainder}
\ind_{\hat{\mathcal E}_{\nu}}
\left|
\hat{\mathbb E}[f(W_n^*)]
-
\mathbb E[f(Z)]
-
\sum_{j=1}^{\nu}n^{-j/2}\int_{\mathbb R} f(x)d[\hat\pi_j(x)\phi(x)]
\right|
\le
\widehat{\mathsf C}_{\nu}\|f'\|_{L^1(\mathbb R)}n^{-(\nu+1)/2},
\end{align}
where \(Z\sim N(0,1)\), and \(\widehat{\mathsf C}_{\nu}\) is the same deterministic envelope constant as above.
\end{lemma}
\begin{proof}
The proof is almost identical to that of Lemma~\ref{lm lemma3 expectation}, conditional on \(\mathcal X\). The event \(\hat{\mathcal E}_{\nu}\) not only gives the uniform bound \(\sup_{x\in\mathbb R}\abs{\hat\varepsilon_n^\nu(x)}\le \widehat{\mathsf C}_{\nu}n^{-(\nu+1)/2}\), but also ensures that \(\hat\varepsilon_n^\nu\) is of finite variation, so that the integration-by-parts argument is well-defined. The desired bound then follows from the same integration-by-parts argument, the uniform bound on \(\hat\varepsilon_n^\nu\), and \(f'\in L^1(\mathbb R)\).
\end{proof}

\subsubsection{Proposition~\ref{theo:expect-cancellation-good} revisited}

We now formulate the resampled version of Proposition~\ref{theo:expect-cancellation-good}. The conclusion is stated on the empirical good event \({\hat{\mathcal E}_{\nu=1}\cap\mathcal E_{\mathrm m}\cap\mathcal E_{\bm\mu}}\), where \(\hat{\mathcal E}_{\nu=1}\) is defined via \eqref{eq goodevent0}, \eqref{eq goodevent1} and \eqref{eq goodevent2} with \(\nu=1\). We emphasize that in \eqref{eq resampled proposition cancellation}, the LHS is a random variable measurable in \(\sigma(\mathcal{X})\) while the RHS is anchored in constants.
\begin{proposition}[Resampled version of Proposition~\ref{theo:expect-cancellation-good}]\label{theo:resample-expect-cancellation-good}
Under Assumption~\ref{assume1}, let \(W_n^*\) be the resampled pivotal quantity defined above. Let \(M\) be a monomial in the marginal and mixed moments of \(\mathbf X_1\), and let \(M_n\) and \(M_n^*\) be its plug-in estimators based on the sample and resample moments, respectively. Let \(g\) be as in Proposition~\ref{theo:expect-cancellation-good}. Then, for \(l\) sufficiently large,
\begin{align}\label{eq resampled proposition cancellation}
\ind_{\hat{\mathcal E}_{\nu=1}\cap\mathcal E_{\mathrm m}\cap\mathcal E_{\bm\mu}}\times 
\left|
\hat{\mathbb E}
\left[
g(W_n^*)(M_n^*-M_n)
\right]
\right|
\le
\frac1n
C_{[\widehat{\mathsf C}_1,\delta,L,\|\nabla\tilde A(\bm\mu)\|,C_M,C_{[g]},\{\mathbb E\|\mathbf X_1\|^k\}_{k=1,\ldots,l}]}.
\end{align}
Here \(\widehat{\mathsf C}_1\) is the deterministic envelope constant appearing in the definition of \(\hat{\mathcal E}_{\nu=1}\).
\end{proposition}
\begin{proof}
Conditional on \(\mathcal X\), the proof of Proposition~\ref{theo:expect-cancellation-good} applies with \(W_n\), \(M\), and \(\varrho_M\) replaced by \(W_n^*\), \(M_n\), and \(\varrho_{M_n}\), using Lemma~\ref{resample lemma1}, Lemma~\ref{lm:gfunc resampled}, and Lemma~\ref{lm:resampled-lemma3-expectation}. These lemmas give the same cancellation bound, with constants initially depending on empirical moments and on \(\|\nabla\tilde A(\hat{\bm\mu}_n)\|\). It remains only to make these constants deterministic. On \(\mathcal E_{\mathrm m}\), the empirical moments \(\hat{\mathbb E}\|\mathbf X_1^*\|^k\) are bounded by the corresponding population moments up to the fixed tolerance, and on \(\mathcal E_{\bm\mu}\), the local Lipschitz property gives \(\|\nabla\tilde A(\hat{\bm\mu}_n)\|\le \|\nabla\tilde A(\bm\mu)\|+L\delta/2\). The only additional change relative to the non-resampled proof is the replacement of \(\varrho_M\) by \(\varrho_{M_n}\) in the resampled analogues of the \(\mathcal T_j\)-terms; the required conditional moment bounds for this replacement are exactly \eqref{lm6 eqb}. The constant \(\widehat{\mathsf C}_1\) enters only through the conditional Edgeworth remainder bound on \(\hat{\mathcal E}_{\nu=1}\). Absorbing the fixed tolerances into the constant proves \eqref{eq resampled proposition cancellation}.
\end{proof}

\subsection{Proof of the Two-Sided Case in Theorem \ref{theorem uniform version}}\label{seca1}

Steps 1--4 largely follow \cite{lam2022}, with additional work to make the remainder bounds uniform and to track their dependence on the relevant constants. Steps 5 and 6 contain the main new argument, which improves the rate from \(O(n^{-3/2})\) to \(O(n^{-2})\) by using the cancellation of monomial moment errors established in Proposition~\ref{theo:expect-cancellation-good}. This explicit tracking of constants will also be used later for the resampled version of the theorem.

\medskip
\noindent
\emph{Step 1: Good event and uniform convergence constants.}
By Theorem~\ref{Theo Hall1}, for any $\nu\geq 1$ there exists a constant $\mathsf C_\nu$ such that
\begin{align}\label{eq theorem1 rate}
\sup_{x\in\mathbb R}
\left|
\PP\left(\sqrt n A_s(\bar{\mathbf X})\le x\right)
-\Phi(x)
-\sum_{j=1}^{\nu}n^{-j/2}\pi_j(x)\phi(x)
\right|
=
\mathsf C_\nu n^{-(\nu+1)/2}.
\end{align}
By Theorem~\ref{Theo Hall2}, we define the good event $\hat{\mathcal E}_\nu(\mathcal X)$, measurable with respect to the $\sigma$-algebra generated by $\mathcal X$, by
\begin{align}
&\sup_x
\left|
\PP\left(\sqrt n \tilde A(\bar{\mathbf X}_b^*)\le x\mid \mathcal X\right)
-\Phi(x)
-\sum_{k=1}^{\nu}n^{-k/2}\hat\pi_k(x)\phi(x)
\right|
\le
\hat{\mathsf C}_\nu n^{-(\nu+1)/2},
\label{eq goodevent0}\\
&\max_{1\le k\le \nu}
\sup_x
\abs{\hat\pi_k(x)}(1+\abs{x})^{-(3k-1)}
\le
\hat{\mathsf C}_\nu,
\label{eq goodevent1}\\
&\max_{1\le k\le \nu}
\sup_x
\abs{\hat\pi_k'(x)}(1+\abs{x})^{-(3k-1)}
\le
\hat{\mathsf C}_\nu.
\label{eq goodevent2}
\end{align}
This event occurs with probability $1-O(n^{-\lambda})$. Throughout this proof, $\mathsf C_\nu$ and $\hat{\mathsf C}_\nu$ denote the constants controlling the convergence rates of these Edgeworth expansions.

\medskip
\noindent
\emph{Step 2: Total variation bounds for signed measures induced by empirical Edgeworth expansions.}
Let $Q^*$ be the resample distribution of $\sqrt n A_s(\bar{\mathbf X}^*)$, which depends on the realization of $\mathcal X$. Let $Q_\nu^*$ be its approximation by the finite Edgeworth expansion
\begin{align}\label{eq hatq}
Q_\nu^*(x)
\coloneqq
\Phi(x)
+
\sum_{k=1}^{\nu}n^{-k/2}\hat p_k(x)\phi(x),
\end{align}
where $\hat p_k(\cdot)$ are the Edgeworth polynomials with coefficients estimated from the sample $\mathcal X$. The conditional two-sided coverage probability of cheap bootstrap is
\begin{equation*}
U(\mathcal X)
\coloneqq
\mathbb E\left[
\iiint
\ind
\left\{
\left|
\frac{\sqrt n A_s(\bar{\mathbf X})}
{\sqrt{\frac1B\sum_{b=1}^B z_b^2}}
\right|
\le \mathsf q
\right\}
dQ^*(z_B)\cdots dQ^*(z_1)
\mid
\mathcal X
\right].
\end{equation*}
On $\hat{\mathcal E}_\nu(\mathcal X)$, the error incurred by approximating the resample measure $Q^*$ by its finite Edgeworth expansion $Q_\nu^*$ is uniformly bounded by $2\hat{\mathsf C}_\nu n^{-(\nu+1)/2}$ over all intervals in $\mathbb R$. By symmetry of the integration region, it suffices to replace $Q_\nu^*(x)$ by its even part,
\begin{equation}
Q_{\mathrm{ts},\nu}^*(x)
=
\Phi(x)
+
\sum_{\substack{1\le k\le \nu\\ k\text{ even}}}
n^{-k/2}\hat p_k(x)\phi(x),
\end{equation}
where the subscript $\mathrm{ts}$ stands for ``two-sided''. Note that \eqref{eq goodevent1} and \eqref{eq goodevent2} imply that the total variation of the signed measures $\hat F_k\coloneqq \hat p_k(x)\phi(x)$ are uniformly bounded over all $\mathcal X\in \hat{\mathcal E}_\nu$:
\begin{align*}
\|\hat F_k\|_{\mathrm{TV}}
=
\int \abs{\hat F_k'(x)}dx \le
\int
\left(
\abs{\hat p_k'(x)}
+
\abs{x}\abs{\hat p_k(x)}
\right)\phi(x)dx 
\le
\hat{\mathsf C}_\nu
\int
(1+\abs{x})^{3k}\phi(x)dx.
\end{align*}
For the same reason,
\[
\left\|Q_{\mathrm{ts},\nu}^*\right\|_{\mathrm{TV}}
\le
\hat{\mathsf C}_\nu(1+\nu)
\int(1+\abs{x})^{3\nu}\phi(x)dx
=
C_{[\hat{\mathsf C}_\nu]}.
\]

\medskip
\noindent
\emph{Step 3: Replacement of $Q^*$ by $Q_{\mathrm{ts},\nu}^*$.}
Consider $U(\mathcal X)$ on $\hat{\mathcal E}_\nu$. We have
\begin{align*}
U(\mathcal X)\ind_{\hat{\mathcal E}_\nu}
=&
\mathbb E\left[
\iiint
\ind
\left\{
\left|
\frac{\sqrt n A_s(\bar{\mathbf X})}
{\sqrt{\frac1B\sum_{b=1}^B z_b^2}}
\right|
\le \mathsf q
\right\}
dQ_{\mathrm{ts},\nu}^*(z_B)\cdots dQ_{\mathrm{ts},\nu}^*(z_1)
\mid
\mathcal X
\right]\ind_{\hat{\mathcal E}_\nu}
+
R_{[1]}\ind_{\hat{\mathcal E}_\nu},
\end{align*}
where $R_{[1]}=\sum_{b=1}^B R_b$ and
\begin{align*}
R_b
\coloneqq
\mathbb E\left[
\iiint
\ind\{\ldots\}
\prod_{k=b+1}^B dQ_{\mathrm{ts},\nu}^*(z_k)
d\left(Q^*(z_b)-Q_{\mathrm{ts},\nu}^*(z_b)\right)
\prod_{k=1}^{b-1}dQ^*(z_k)
\mid
\mathcal X
\right]
\end{align*}
is the same as in the proof of Theorem 2 in \cite{lam2022}. Therefore, uniformly over all $\mathsf q$ and all $\mathcal X\in \hat{\mathcal E}_\nu$,
\begin{align*}
\abs{R_b}
&\le
\hat{\mathsf C}_\nu n^{-(\nu+1)/2}
\mathbb E\left[
\iiint
\ind\{\ldots\}
\prod_{k=b+1}^B dQ_{\mathrm{ts},\nu}^*(z_k)
\prod_{k=1}^{b-1}dQ^*(z_k)
\mid
\mathcal X
\right] \\
&\le
\hat{\mathsf C}_\nu n^{-(\nu+1)/2}
\left(\|Q_{\mathrm{ts},\nu}^*\|_{\mathrm{TV}}\right)^{B-b}
\le
C_{[\hat{\mathsf C}_\nu,B]}n^{-(\nu+1)/2}.
\end{align*}
Consequently,
\begin{align}\label{eq r1bound}
\left|R_{[1]}\times\ind_{\hat{\mathcal E}_\nu}\right|
\le
C_{[\hat{\mathsf C}_\nu,B]}n^{-(\nu+1)/2}.
\end{align}

\medskip
\noindent
\emph{Step 4: Integration decomposition.}
By the same arguments as in \cite{lam2022}, we have
\begin{align*}
U(\mathcal X)\ind_{\hat{\mathcal E}_\nu}
=
\left(
U_0(\mathcal X)
+
\frac1n U_1(\mathcal X)
+
R_{[1]}
+
R_{[2]}
\right)
\ind_{\hat{\mathcal E}_\nu},
\end{align*}
where
\begin{align}
U_0(\mathcal X)
&\coloneqq
\mathbb E\left[
\iiint
\ind
\left\{
\left|
\frac{\sqrt n A_s(\bar{\mathbf X})}
{\sqrt{\frac1B\sum_{b=1}^B z_b^2}}
\right|
\le \mathsf q
\right\}
d\Phi(z_B)\cdots d\Phi(z_1)
\mid
\mathcal X
\right],
\label{eq easy blue}\\
U_1(\mathcal X)
&\coloneqq
B\mathbb E\left[
\iiint
\ind
\left\{
\left|
\frac{\sqrt n A_s(\bar{\mathbf X})}
{\sqrt{\frac1B\sum_{b=1}^B z_b^2}}
\right|
\le \mathsf q
\right\}
d\left(\hat p_2(z_B)\phi(z_B)\right)
d\Phi(z_{B-1})\cdots d\Phi(z_1)
\mid
\mathcal X
\right],
\label{eq hard red}\\
R_{[2]}
&\coloneqq
\frac{B}{n^2}
\mathbb E\left[
\iiint
\ind\{\ldots\}
d\left(\hat p_4(z_B)\phi(z_B)\right)
d\Phi(z_{B-1})\cdots d\Phi(z_1)
\mid
\mathcal X
\right]
\notag\\
\quad
+
\frac{B(B-1)}{2n^2} &
\mathbb E\left[
\iiint
\ind\{\ldots\}
d\left(\hat p_2(z_B)\phi(z_B)\right)
d\left(\hat p_2(z_{B-1})\phi(z_{B-1})\right)
d\Phi(z_{B-2})\cdots d\Phi(z_1)
\mid
\mathcal X
\right]
+
O(n^{-3}).
\notag
\end{align}
The remainder term $R_{[2]}$ consists of a finite sum of integrals. Again, since $\|\hat F_k\|_{\mathrm{TV}}$ is bounded, we have
\begin{align}\label{Eq remainder2}
\left|R_{[2]}\times\ind_{\hat{\mathcal E}_\nu}\right|
\le
C_{[\hat{\mathsf C}_\nu,B]}n^{-2}.
\end{align}
Before proceeding to the next step, note that, for the final $O(n^{-2})$ bound, it suffices to take $\nu=3$. In that case, the precise definition of $R_{[2]}$ should be adjusted accordingly. Summarizing Steps 1--4, we obtain
\begin{align}\label{Eq r1r2error}
\left|(R_{[1]}+R_{[2]})\times\ind_{\hat{\mathcal E}_3}\right|
\le
C_{[\hat{\mathsf C}_3,B]}n^{-2}.
\end{align}

\medskip
\noindent
\emph{Step 5: Analysis of $U_0(\mathcal X)$.}
Let $U_0$ denote the unconditional expectation of $U_0(\mathcal X)$, i.e.,
$U_0\coloneqq \EE[U_0(\mathcal X)]$. We have
\begin{equation}\label{eq r3}
\begin{split}
U_0
=&
\iiint
\ind
\left\{
\left|
\frac{z_0}
{\sqrt{\frac1B\sum_{b=1}^B z_b^2}}
\right|
\le \mathsf q
\right\}
\prod_{b=0}^B d\Phi(z_b)
\\
&+
\frac1n
\iiint
\ind
\left\{
\left|
\frac{z_0}
{\sqrt{\frac1B\sum_{b=1}^B z_b^2}}
\right|
\le \mathsf q
\right\}
d\left(q_2(z_0)\phi(z_0)\right)
\prod_{b=1}^B d\Phi(z_b)
+
R_{[3]}.
\end{split}
\end{equation}
By applying the two-sided version of \eqref{eq theorem1 rate} with
$x=\pm\mathsf q\sqrt{\frac1B\sum_{b=1}^B z_b^2}$ and $\nu=3$, we obtain
\[
\abs{R_{[3]}}
\le
n^{-2}C_{[\mathsf C_3]}.
\]
This bound holds uniformly over all quantiles $\mathsf q$, where the constant $\mathsf C_3$ is the one appearing in \eqref{eq theorem1 rate}. We now evaluate the integrals in $U_0$. The first integral in $U_0$ is
\begin{align*}
\iiint
\ind
\left\{
\left|
\frac{z_0}
{\sqrt{\frac1B\sum_{b=1}^B z_b^2}}
\right|
\le \mathsf q
\right\}
\prod_{b=0}^B d\Phi(z_b)
=
\PP\left(\abs{t_B}\le \mathsf q\right)
=
2\Psi_B(\mathsf q)-1,
\end{align*}
where $t_B$ has the Student's $t$ distribution with $B$ degrees of freedom, and $\Psi_B$ denotes its cumulative distribution function. By substituting the density function of the $\chi_B^2$ distribution, the second integral is
\begin{align*}
&\iiint
\ind
\left\{
\left|
\frac{z_0}
{\sqrt{\frac1B\sum_{b=1}^B z_b^2}}
\right|
\le \mathsf q
\right\}
d\left(q_2(z_0)\phi(z_0)\right)
\prod_{b=1}^B d\Phi(z_b)
\\
&\qquad
=
\sum_{k=1,3,5}
b_{[2,k]}
\frac{
2^{(k+1)/2}\Gamma\left(\frac{B+k}{2}\right)
}
{\sqrt{\pi}\Gamma(B/2)}
\cdot
\frac{
B^{B/2}\mathsf q^k
}
{
(B+\mathsf q^2)^{(B+k)/2}
}.
\end{align*}
This is part of the expression for $\zeta_2(\mathsf q)$, and we denote it by $\zeta_2^q(\mathsf q)$. Consequently,
\[
U_0
=
2\Psi_B(\mathsf q)-1
+
\frac1n\zeta_2^q(\mathsf q)
+
R_{[3]}.
\]

\noindent
\emph{Step 6: Analysis of $U_1(\mathcal X)$.}
Denote the pivotal quantity by $W_n\coloneqq \sqrt n A_s(\bar{\mathbf X})$ and define the random variable $r\in\sigma(\mathcal X)$ by
\[
r
\coloneqq
\frac{\sqrt n A_s(\bar{\mathbf X})}{\mathsf q}
=
\frac{W_n}{\mathsf q}.
\]
Intuitively, $\sqrt B r$ is the radius of the random integration boundary. Recall the definition of the Edgeworth polynomials in \eqref{Eq edgepoly}, and define
\[
\hat p_2(x)
=
\sum_{k=1,3,5}\hat a_{[2,k]}x^k
\]
as the counterpart of $p_2(x)$ with coefficients estimated from the sample. By using the density of the $\chi^2_{B-1}$ distribution and substituting this expression for $\hat p_2(x)$, we obtain
\begin{equation}\label{eq32}
\begin{split}
U_1(\mathcal X)
=&
-\frac{4B}
{2^{(B-1)/2}\Gamma\left(\frac{B-1}{2}\right)\sqrt{2\pi}}
\times
e^{-Br^2/2}
\times
\int_0^{\abs{r}\sqrt B}
x\hat p_2(x)(Br^2-x^2)^{(B-3)/2}dx
\\
=&
-\sum_{k=1,3,5}
\hat a_{[2,k]}
\times
\frac{
B^{(B+k+1)/2}\Gamma\left(1+\frac{k}{2}\right)
}
{
2^{B/2-1}\sqrt{\pi}\Gamma\left(\frac{1+B+k}{2}\right)
}
\times
\abs{r}^{B+k-1}
\times
\exp\left(-\frac{Br^2}{2}\right).
\end{split}
\end{equation}
Equivalently, each \(\hat a_{[2,k]}\) is a polynomial in the sample moment estimators. Hence, $U_1(\mathcal X)$ can be written as a finite sum of terms of the form
\[
M_n\times g_{[B,k]}(r),
\]
where $M_n$ is a monomial in the sample moment estimators and
\begin{align}\label{eq gbkx}
g_{[B,k]}(x)
\coloneqq
\abs{x}^{B-1+k}\times\exp\left(-\frac{Bx^2}{2}\right).
\end{align}
Note that, for each \(k=1,3,5\), \(g_{[B,k]}\in C^2(\mathbb R)\) whenever \(B\geq 2\). Moreover, \(g_{[B,k]}\), \(g_{[B,k]}'\), and \(g_{[B,k]}''\) are uniformly bounded on \(\mathbb R\). Let $M$ denote the corresponding monomial evaluated at the true moments. Then each such term can be decomposed as
\[
M_n\times g_{[B,k]}(r)
=
M\times g_{[B,k]}(r)
+
(M_n-M)\times g_{[B,k]}(r).
\]
We restrict to $\mathsf q\ge \varepsilon_0>0$ in order to obtain a uniform bound for the second term. Therefore, for $B\geq 2$ and any $\varepsilon_0>0$, Proposition~\ref{theo:expect-cancellation-good}, applied with \(\tilde A=A_s\) and the even weight \(x\mapsto g_{[B,k]}(x/\mathsf q)\), yields
\begin{align}\label{eq 16}
&\sup_{\mathsf q\ge \varepsilon_0}
\left|
\EE\left[
(M_n-M)\times g_{[B,k]}\left(\frac{W_n}{\mathsf q}\right)
\right]
\right|
\le 
\frac1n
\times
\left(1+\frac{B^2}{\varepsilon_0^4}\right)
\times
C_{[\mathsf C_1,\delta,L,\|\nabla A_s(\bm\mu)\|,C_M,C_{[g]},\{\mathbb E\|\mathbf X_1\|^k\}_{k=1,\ldots,l}]}.
\end{align}
To treat all three values \(k=1,3,5\), we may take
\begin{align}\label{eq gbkxs}
C_{[g]}
=
\sum_{k=1,3,5}
\left(
\sup_{x\in\mathbb R}\abs{g_{[B,k]}(x)}
+
\sup_{x\in\mathbb R}\abs{g_{[B,k]}'(x)}
+
\sup_{x\in\mathbb R}\abs{g_{[B,k]}''(x)}
\right),
\end{align}
which depends only on \(B\).

Now recall the original definition of \(U_1(\mathcal X)\) in \eqref{eq hard red}. Define \(\tilde U_1(\mathcal X)\) by replacing \(\hat p_2\) with \(p_2\), namely,
\begin{align*}
\tilde U_1(\mathcal X)
\coloneqq
B\times\mathbb E\left[
\iiint
\ind
\left\{
\left|
\frac{\sqrt n A_s(\bar{\mathbf X})}
{\sqrt{\frac1B\sum_{b=1}^B z_b^2}}
\right|
\le \mathsf q
\right\}
d\left(p_2(z_B)\phi(z_B)\right)
d\Phi(z_{B-1})\cdots d\Phi(z_1)
\mid
\mathcal X
\right].
\end{align*}
Equation~\eqref{eq 16} implies that
\begin{align*}
\sup_{\mathsf q\ge \varepsilon_0}
\left|
\EE\left[
U_1(\mathcal X)-\tilde U_1(\mathcal X)
\right]
\right|
\le
\frac1n
\times
C_{[\varepsilon_0,\mathsf C_1,\delta,L,\|\nabla A_s(\bm\mu)\|,C_M,B,\{\mathbb E\|\mathbf X_1\|^k\}_{k=1,\ldots,l}]}.
\end{align*}
Moreover,
\begin{align*}
\EE[\tilde U_1(\mathcal X)]
=&
B
\times
\iiint
\ind
\left\{
\left|
\frac{z_0}
{\sqrt{\frac1B\sum_{b=1}^B z_b^2}}
\right|
\le \mathsf q
\right\}
d\left(p_2(z_B)\phi(z_B)\right)
d\Phi(z_{B-1})\cdots d\Phi(z_1)d\Phi(z_0)
+
R_{[7]}
\\
=&
-\sum_{k=1,3,5}
a_{[2,k]}
\times
\frac{
\Gamma\left(\frac{k}{2}+1\right)
\Gamma\left(\frac{B+k}{2}\right)
}
{
\pi\Gamma\left(\frac{B+k+1}{2}\right)
}
\times
\frac{
2^{(k+1)/2}B^{(B+1+k)/2}\mathsf q
}
{
(B+\mathsf q^2)^{(B+k)/2}
}
+
R_{[7]}
\eqqcolon
\zeta_2^p(\mathsf q)+R_{[7]}.
\end{align*}
As in the argument for $R_{[3]}$ in \eqref{eq r3}, we have
\[
\abs{R_{[7]}}
\le
n^{-1}\times C_{[\mathsf C_1]}.
\]
Note also that $\zeta_2^p(\mathsf q)$ and $\zeta_2^q(\mathsf q)$, defined in Step 5, together constitute the quantity $\zeta_2(\mathsf q)$ in Theorem~\ref{theorem uniform version}.

\medskip
\noindent
\emph{Step 7: Proof summary.}
Since $\PP(\abs{T}\le \mathsf q)=\EE[U(\mathcal X)]$, it remains to combine the preceding bounds. Choose $\nu=3$ and choose $\lambda$ sufficiently large. From Step 4,
\[
U(\mathcal X)\ind_{\hat{\mathcal E}_3}
=
\left(
U_0(\mathcal X)
+
\frac1n U_1(\mathcal X)
+
R_{[1]}
+
R_{[2]}
\right)
\ind_{\hat{\mathcal E}_3},
\]
with
\[
\left|(R_{[1]}+R_{[2]})\ind_{\hat{\mathcal E}_3}\right|
\le
C_{[\hat{\mathsf C}_3,B]}n^{-2}.
\]
The complement of $\hat{\mathcal E}_3$ contributes only $O(n^{-\lambda})$, because $0\le U(\mathcal X)\le 1$. From Step 5,
\[
\EE[U_0(\mathcal X)]
=
2\Psi_B(\mathsf q)-1
+
\frac1n\zeta_2^q(\mathsf q)
+
R_{[3]},
\qquad
\abs{R_{[3]}}
\le
n^{-2}C_{[\mathsf C_3]}.
\]
From Step 6,
\[
\sup_{\mathsf q\ge \varepsilon_0}
\left|
\EE\left[
U_1(\mathcal X)-\tilde U_1(\mathcal X)
\right]
\right|
\le
\frac1n
\times
C_{[\varepsilon_0,\mathsf C_1,\delta,L,\|\nabla A_s(\bm\mu)\|,C_M,B,\{\mathbb E\|\mathbf X_1\|^k\}_{k=1,\ldots,l}]}.
\]
Moreover,
\[
\EE[\tilde U_1(\mathcal X)]
=
\zeta_2^p(\mathsf q)+R_{[7]},
\qquad
\abs{R_{[7]}}
\le
n^{-1}C_{[\mathsf C_1]}.
\]
Therefore, with
\[
\zeta_2(\mathsf q)\coloneqq \zeta_2^q(\mathsf q)+\zeta_2^p(\mathsf q),
\]
we obtain
\begin{align*}
&\sup_{\mathsf q\ge \varepsilon_0}
\left|
\PP\left(\right|T\left|\le \mathsf q\right)
-
\left(2\Psi_B(\mathsf q)-1\right)
-
\frac1n\zeta_2(\mathsf q)
\right|\\
\le &
\frac1{n^2}
C_{[\varepsilon_0,\mathsf C_1,\hat{\mathsf C}_3,\mathsf C_3,\delta,L,\|\nabla A_s(\bm\mu)\|,C_M,B,\{\mathbb E\|\mathbf X_1\|^k\}_{k=1,\ldots,l}]}.
\end{align*}
This proves the desired two-sided uniform expansion.

\subsection{Proof of the One-Sided Case in Theorem \ref{theorem uniform version}}\label{sec theorem1B2}

As in the two-sided case, let \(Q^*\) be the resample distribution of \(\sqrt n A_s(\bar{\mathbf X}^*)\), and let \(Q^*_\nu\) be its finite Edgeworth approximation defined in \eqref{eq hatq}. For a given quantile \(\mathsf q\in\mathbb R\), the conditional one-sided coverage probability of cheap bootstrap is
\begin{equation*}
V(\mathcal X)
\coloneqq
\mathbb E\left[
\iiint
\ind\left\{
\frac{\sqrt n A_s(\bar{\mathbf X})}{\sqrt{\frac1B\sum_{b=1}^B z_b^2}}
\le \mathsf q
\right\}
dQ^*(z_B)\cdots dQ^*(z_1)
\mid
\mathcal X
\right].
\end{equation*}
Recall the good event \(\hat{\mathcal E}_\nu(\mathcal X)\) from Step 1 of the two-sided proof. By the same arguments as in Steps 1--4 there, we have
\begin{align}\label{eq34}
V(\mathcal X)\times\ind_{\hat{\mathcal E}_\nu}
=
\left(
V_0(\mathcal X)
+
\frac{1}{\sqrt n}\times V_{1/2}(\mathcal X)
+
\frac{1}{n}\times V_1(\mathcal X)
+
R_{[8]}
\right)\times\ind_{\hat{\mathcal E}_\nu},
\end{align}
where
\begin{align}
V_0(\mathcal X)
\coloneqq&
\mathbb E\left[
\iiint
\ind\left\{
\frac{\sqrt n A_s(\bar{\mathbf X})}{\sqrt{\frac1B\sum_{b=1}^B z_b^2}}
\le \mathsf q
\right\}
d\Phi(z_B)\cdots d\Phi(z_1)
\mid
\mathcal X
\right],
\label{eq defv0}\\
V_{1/2}(\mathcal X)
\coloneqq&
B\times
\mathbb E\left[
\iiint
\ind\left\{
\frac{\sqrt n A_s(\bar{\mathbf X})}{\sqrt{\frac1B\sum_{b=1}^B z_b^2}}
\le \mathsf q
\right\}
d\bigl(\hat p_1(z_B)\phi(z_B)\bigr)
d\Phi(z_{B-1})\cdots d\Phi(z_1)
\mid
\mathcal X
\right],
\label{eq defvhalf}\\
V_1(\mathcal X)
\coloneqq&
B\times
\mathbb E\left[
\iiint
\ind\left\{
\frac{\sqrt n A_s(\bar{\mathbf X})}{\sqrt{\frac1B\sum_{b=1}^B z_b^2}}
\le \mathsf q
\right\}
d\bigl(\hat p_2(z_B)\phi(z_B)\bigr)
d\Phi(z_{B-1})\cdots d\Phi(z_1)
\mid
\mathcal X
\right]
\notag\\
+
\frac{B(B-1)}{2}\times &
\mathbb E\left[
\iiint
\ind\left\{
\frac{\sqrt n A_s(\bar{\mathbf X})}{\sqrt{\frac1B\sum_{b=1}^B z_b^2}}
\le \mathsf q
\right\}
d\bigl(\hat p_1(z_B)\phi(z_B)\bigr)
d\bigl(\hat p_1(z_{B-1})\phi(z_{B-1})\bigr)
d\Phi(z_{B-2})\cdots d\Phi(z_1)
\mid
\mathcal X
\right].
\notag
\end{align}
As in the two-sided case, it suffices to take \(\nu=3\), so that
\begin{align}\label{eq up-r8}
\left|R_{[8]}\times\ind_{\hat{\mathcal E}_3}\right|
\le
C_{[\hat{\mathsf C}_3,B]}\times n^{-2}.
\end{align}

\medskip
\noindent
\emph{Step 1: Integral simplification.}
In view of the evenness and exponential tail decay of \(\hat p_1(\cdot)\phi(\cdot)\), together with the symmetry of the integration region, we have
\begin{align*}
V_{1/2}(\mathcal X)\equiv 0
\end{align*}
for every realization of the sample \(\mathcal X\). The same argument also applies to the second expectation in \(V_1(\mathcal X)\), hence
\begin{align}\label{eq defv1}
V_1(\mathcal X)
\equiv&
B\times
\mathbb E\left[
\iiint
\ind\left\{
\frac{\sqrt n A_s(\bar{\mathbf X})}{\sqrt{\frac1B\sum_{b=1}^B z_b^2}}
\le \mathsf q
\right\}
d\bigl(\hat p_2(z_B)\phi(z_B)\bigr)
d\Phi(z_{B-1})\cdots d\Phi(z_1)
\mid
\mathcal X
\right].
\end{align}

\medskip
\noindent
\emph{Step 2: Analysis of \(V_0(\mathcal X)\).}
Let \(V_0\coloneqq \E[V_0(\mathcal X)]\). By applying the Edgeworth expansion of \(W_n=\sqrt n A_s(\bar{\mathbf X})\) in the original representation \eqref{eq defv0}, we obtain
\begin{equation}\label{eq up-v0-expansion}
\begin{split}
V_0
=&
\iiint
\ind\left\{
\frac{z_0}{\sqrt{\frac1B\sum_{b=1}^B z_b^2}}
\le \mathsf q
\right\}
\prod_{b=0}^B d\Phi(z_b)
\\
&+
\frac{1}{\sqrt n}\times
\iiint
\ind\left\{
\frac{z_0}{\sqrt{\frac1B\sum_{b=1}^B z_b^2}}
\le \mathsf q
\right\}
d\bigl(q_1(z_0)\phi(z_0)\bigr)
\prod_{b=1}^B d\Phi(z_b)
\\
&+
\frac{1}{n}\times
\iiint
\ind\left\{
\frac{z_0}{\sqrt{\frac1B\sum_{b=1}^B z_b^2}}
\le \mathsf q
\right\}
d\bigl(q_2(z_0)\phi(z_0)\bigr)
\prod_{b=1}^B d\Phi(z_b)
+
R_{[9]}.
\end{split}
\end{equation}
As in Step 5 of the two-sided proof, the Edgeworth remainder satisfies $\abs{R_{[9]}}
\le
n^{-3/2}\times C_{[\mathsf C_3]}$. Now evaluate the three integrals in \eqref{eq up-v0-expansion}. The first one equals
\begin{align*}
\iiint
\ind\left\{
\frac{z_0}{\sqrt{\frac1B\sum_{b=1}^B z_b^2}}
\le \mathsf q
\right\}
\prod_{b=0}^B d\Phi(z_b)
=
\PP(t_B\le \mathsf q)
=
\Psi_B(\mathsf q).
\end{align*}
The second integral equals
\begin{align*}
\sum_{k=0,2}
b_{[1,k]}
\times
\frac{B^{B/2+k}\mathsf q^k}{\sqrt{2\pi}(B+\mathsf q^2)^{B/2+k}}
\eqqcolon
\zeta_{1}^{\mathrm{up}}(\mathsf q),
\end{align*}
and the third integral equals
\begin{align*}
\sum_{k=1,3,5}
b_{[2,k]}
\times
\frac{
2^{k/2}\Gamma\left(\frac{B+k}{2}\right)\mathsf q^k B^{B/2}
}{
\sqrt{2\pi}\Gamma\left(\frac{B}{2}\right)(B+\mathsf q^2)^{(B+k)/2}
}
\eqqcolon
{\zeta_{2,q}^\mathrm{up}}(\mathsf q).
\end{align*}
Therefore,
\begin{align}\label{eq up-v0-final}
V_0
=
\Psi_B(\mathsf q)
+
\frac{1}{\sqrt n}\times\zeta_{1}^{\mathrm{up}}(\mathsf q)
+
\frac{1}{n}\times{\zeta_{2,q}^\mathrm{up}}(\mathsf q)
+
R_{[9]}.
\end{align}

\medskip
\noindent
\emph{Step 3: Analysis of \(V_1(\mathcal X)\).}
After Step 1, only the first expectation in \eqref{eq defv1} remains. Namely,
\begin{align*}
V_1(\mathcal X)
=
B\times
\mathbb E\left[
\iiint
\ind\left\{
\frac{\sqrt n A_s(\bar{\mathbf X})}{\sqrt{\frac1B\sum_{b=1}^B z_b^2}}
\le \mathsf q
\right\}
d\bigl(\hat p_2(z_B)\phi(z_B)\bigr)
d\Phi(z_{B-1})\cdots d\Phi(z_1)
\mid
\mathcal X
\right].
\end{align*}
Since our final goal is only a uniform error of order \(O(n^{-3/2})\), we do not analyze \(V_1(\mathcal X)\) via the cancellation argument used in the two-sided case. Instead, we follow the simpler approach in \cite{lam2022}, Theorem 2. Using the explicit form of \(\hat p_2\), we may write
\begin{align*}
V_1(\mathcal X)
=
B\times\sum_{k=1,3,5}\hat a_{[2,k]}\times
\mathbb E\left[
\iiint
\ind\{\ldots\}
d\bigl(z_B^k\phi(z_B)\bigr)
d\Phi(z_{B-1})\cdots d\Phi(z_1)
\mid
\mathcal X
\right].
\end{align*}
The conditional expectation above is bounded by a universal constant:
\begin{align*}
&\left|
\mathbb E\left[
\iiint
\ind\{\ldots\}
d\bigl(z_B^k\phi(z_B)\bigr)
d\Phi(z_{B-1})\cdots d\Phi(z_1)
\mid
\mathcal X
\right]
\right|
\\
&\qquad\le
\iiint
\abs{z_B^k\phi(z_B)}dz_B d\Phi(z_{B-1})\cdots d\Phi(z_1).
\end{align*}
Define \(\tilde V_1(\mathcal X)\) by replacing \(\hat a_{[2,k]}\) with \(a_{[2,k]}\):
\begin{align*}
\tilde V_1(\mathcal X)
\coloneqq
B\times
\sum_{k=1,3,5}a_{[2,k]}\times
\mathbb E\left[
\iiint
\ind\{\ldots\}
d\bigl(z_B^k\phi(z_B)\bigr)
d\Phi(z_{B-1})\cdots d\Phi(z_1)
\mid
\mathcal X
\right].
\end{align*}
Then
\begin{align*}
\E\left|V_1(\mathcal X)-\tilde V_1(\mathcal X)\right|
\le
\frac{1}{\sqrt n}\times
C_{[B,\{\E\|\mathbf X_1\|^k\}_{k=1,\ldots,l}]},
\end{align*}
because each \(\hat a_{[2,k]}\) is a polynomial in sample moment estimators, and sufficient moment assumptions imply uniform integrability. Next, the analysis of \(\tilde V_1(\mathcal X)\) is parallel to that of \(V_0(\mathcal X)\):
\begin{align*}
 &\E[\tilde V_1(\mathcal X)]\\
=&
B\times
\mathbb E\left[
\iiint
\ind\left\{
\frac{z_0}{\sqrt{\frac1B\sum_{b=1}^B z_b^2}}
\le \mathsf q
\right\}
d\bigl(p_2(z_B)\phi(z_B)\bigr)
d\Phi(z_{B-1})\cdots d\Phi(z_1)d\Phi(z_0)
\right]\\
+
R_{[10]}
\\
=&
-\sum_{k=1,3,5}
a_{[2,k]}
\times
\frac{
2^{(k-1)/2}
\Gamma\left(\frac{B+k}{2}\right)
\Gamma\left(1+\frac{k}{2}\right)
\mathsf q B^{(B+k+1)/2}
}{
\pi
\Gamma\left(\frac{B+k+1}{2}\right)
(B+\mathsf q^2)^{(B+k)/2}
}
+
R_{[10]}
\\
\eqqcolon&
{\zeta_{2,p}^\mathrm{up}}(\mathsf q)+R_{[10]},
\end{align*}
with
\[
\abs{R_{[10]}}
\le
\frac{1}{\sqrt n}\times C_{[\mathsf C_1,B]}.
\]
Combining the last two displays, we obtain
\begin{align}\label{eq up-v1-final}
\E[V_1(\mathcal X)]
=
{\zeta_{2,p}^\mathrm{up}}(\mathsf q)
+
R_{[11]},
\qquad
\abs{R_{[11]}}
\le
\frac{1}{\sqrt n}\times
C_{[\mathsf C_1,B,\{\E\|\mathbf X_1\|^k\}_{k=1,\ldots,l}]}.
\end{align}

\medskip
\noindent
\emph{Step 4: Summary.}
Define
\[
{\zeta_{2}^\mathrm{up}}(\mathsf q)
\coloneqq
{\zeta_{2,q}^\mathrm{up}}(\mathsf q)+{\zeta_{2,p}^\mathrm{up}}(\mathsf q),
\]
i.e.
\begin{align*}
{\zeta_{2}^\mathrm{up}}(\mathsf q)
=&
\sum_{k=1,3,5}
b_{[2,k]}
\times
\frac{
2^{k/2}\Gamma\left(\frac{B+k}{2}\right)\mathsf q^k B^{B/2}
}{
\sqrt{2\pi}\Gamma\left(\frac{B}{2}\right)(B+\mathsf q^2)^{(B+k)/2}
}\\
- &\sum_{k=1,3,5}
a_{[2,k]}
\times
\frac{
2^{(k-1)/2}
\Gamma\left(\frac{B+k}{2}\right)
\Gamma\left(\frac{k+2}{2}\right)
\mathsf q B^{(B+k+1)/2}
}{
\pi
\Gamma\left(\frac{B+k+1}{2}\right)
(B+\mathsf q^2)^{(B+k)/2}
}
\end{align*}
Combining \eqref{eq up-r8}, \eqref{eq up-v0-final}, and \eqref{eq up-v1-final}, we obtain
\begin{align*}
\sup_{\mathsf q\in\mathbb R}
\left|
\PP(T\le \mathsf q)
-
\Psi_B(\mathsf q)
-
\frac{\zeta_{1}^{\mathrm{up}}(\mathsf q)}{\sqrt n}
-
\frac{{\zeta_{2}^\mathrm{up}}(\mathsf q)}{n}
\right|
\le
n^{-3/2}\times
C_{[\mathsf C_1,\mathsf C_3,\hat{\mathsf C}_3,B,\{\E\|\mathbf X_1\|^k\}_{k=1,\ldots,l}]}.
\end{align*}
This is exactly the desired upper-tail expansion.

\subsection{Proof of the Two-Sided Case in Theorem \ref{theorem resample uniform version}}\label{sec-proof-resample-ts}

We prove the symmetric two-sided expansion for the resampled statistic. Fix \(\varepsilon_0>0\) and \(B\geq 2\). Throughout this proof, all suprema over \(\mathsf q\) are taken over \(\mathsf q\ge \varepsilon_0\). Recall the empirical Edgeworth good event \(\hat{\mathcal E}_{\nu}\) defined in \eqref{eq goodevent0}, \eqref{eq goodevent1}, and \eqref{eq goodevent2}, the empirical norm-moment event \(\mathcal E_{\mathrm m}\) defined in \eqref{good event emp moment}, and the empirical mean event \(\mathcal E_{\bm\mu}\) defined in \eqref{good event mean}. Define
\begin{align}\label{eq good event resample G}
\mathcal G
\coloneqq
\hat{\mathcal E}_{\nu=3}
\cap
\hat{\mathcal E}_{\nu=1}
\cap
\mathcal E_{\mathrm m}
\cap
\mathcal E_{\bm\mu}.
\end{align}
By the high-probability bounds for these events, for any prescribed \(\lambda>0\), the moment order in Assumption~\ref{assume1} can be chosen sufficiently large so that
\[
\PP(\mathcal G)=1-O(n^{-\lambda}).
\]
It therefore suffices to prove the desired \(n^{-2}\) bound on \(\mathcal G\), with a deterministic constant. Let \(W_n^*\) denote the resampled pivotal numerator,
\[
W_n^*
\coloneqq
\sqrt n\,A_s^*(\bar{\mathbf X}^*),
\]
where \(A_s^*\) is the empirical analogue of \(A_s\) based on \(\mathcal X\). Conditional on \((\mathcal X,\mathcal X^*)\), let \(Q^{**}\) denote the distribution of one inner resampling coordinate. For a fixed \(\mathsf q\), define
\[
U(\mathcal X,\mathcal X^*)
\coloneqq
\mathbb E\left[
\iiint
\ind
\left\{
\left|
\frac{W_n^*}{\sqrt{\frac1B\sum_{b=1}^B z_b^2}}
\right|
\le \mathsf q
\right\}
dQ^{**}(z_B)\cdots dQ^{**}(z_1)
\mid
\mathcal X,\mathcal X^*
\right].
\]
Then \(\PP\left(\abs{T^*}\le \mathsf q\mid \mathcal X\right)
=
\mathbb E\left[
U(\mathcal X,\mathcal X^*)
\mid
\mathcal X
\right]\).

\medskip
\noindent
\emph{Step 1: Reduction to \(U_0\) and \(U_1\).}
Conditional on \((\mathcal X,\mathcal X^*)\), the same replacement argument used in Steps 1--4 of the proof of Theorem~\ref{theorem uniform version} applies to the inner resampling distribution \(Q^{**}\). Thus,
\[
U(\mathcal X,\mathcal X^*)
=
U_0(\mathcal X,\mathcal X^*)
+
\frac1n U_1(\mathcal X,\mathcal X^*)
+
R_{[1]}^*
+
R_{[2]}^*,
\]
up to a bad-event contribution whose conditional expectation is \(O(n^{-2})\) after choosing \(\lambda\) sufficiently large. Here
\[
U_0(\mathcal X,\mathcal X^*)
\coloneqq
\mathbb E\left[
\iiint
\ind
\left\{
\left|
\frac{W_n^*}{\sqrt{\frac1B\sum_{b=1}^B z_b^2}}
\right|
\le \mathsf q
\right\}
d\Phi(z_B)\cdots d\Phi(z_1)
\mid
\mathcal X,\mathcal X^*
\right],
\]
and
\begin{align*}
&U_1(\mathcal X,\mathcal X^*)
\coloneqq\\
B\times &
\mathbb E\left[
\iiint
\ind
\left\{
\left|
\frac{W_n^*}{\sqrt{\frac1B\sum_{b=1}^B z_b^2}}
\right|
\le \mathsf q
\right\}
d\left(p_2^*(z_B)\phi(z_B)\right)
d\Phi(z_{B-1})\cdots d\Phi(z_1)
\mid
\mathcal X,\mathcal X^*
\right],
\end{align*}
where \(p_2^*\) is the second-order Edgeworth polynomial associated with the empirical distribution of \(\mathcal X^*\). On \(\mathcal G\), the total variation bounds for the signed Edgeworth measures are controlled by deterministic constants depending only on the displayed quantities below. Hence
\[
\ind_{\mathcal G}\times
\left|
\mathbb E\left[
R_{[1]}^*+R_{[2]}^*
\mid
\mathcal X
\right]
\right|
\le
\frac1{n^2}
C_{[\hat{\mathsf C}_3,B,\{\mathbb E\|\mathbf X_1\|^k\}_{k=1,\ldots,l}]}.
\]

\medskip
\noindent
\emph{Step 2: Analysis of \(U_0\).}
For fixed \((\mathcal X,\mathcal X^*)\), the Gaussian integration in \(U_0\) gives
\[
U_0(\mathcal X,\mathcal X^*)
=
2\Psi_B\left(\frac{W_n^*}{\mathsf q}\right)-1.
\]
Taking conditional expectation given \(\mathcal X\) and applying the conditional Edgeworth expansion on \(\hat{\mathcal E}_{\nu=3}\), exactly as in Step 5 of the proof of Theorem~\ref{theorem uniform version}, yields
\[
\ind_{\mathcal G}\times
\left|
\mathbb E\left[
U_0(\mathcal X,\mathcal X^*)
\mid
\mathcal X
\right]
-
\left(2\Psi_B(\mathsf q)-1\right)
-
\frac1n\hat\zeta_2^q(\mathsf q)
\right|
\le
\frac1{n^2}
C_{[\hat{\mathsf C}_3]},
\]
uniformly over \(\mathsf q\ge\varepsilon_0\). Here \(\hat\zeta_2^q\) denotes the sample-based analogue of the first summation term in \(\zeta_2^{\mathrm{ts}}\) in \eqref{eq zeta2}, obtained by replacing the population moment coefficients there with their sample plug-in estimators.

\medskip
\noindent
\emph{Step 3: Analysis of \(U_1\).}
Write
\[
p_2^*(x)=\sum_{k=1,3,5}a_{[2,k]}^*x^k,
\qquad
\hat p_2(x)=\sum_{k=1,3,5}\hat a_{[2,k]}x^k,
\]
where \(a_{[2,k]}^*\) is computed from the outer resample \(\mathcal X^*\), while \(\hat a_{[2,k]}\) is computed from the original sample \(\mathcal X\). The same chi-square calculation as in Step 6 of the proof of Theorem~\ref{theorem uniform version} gives
\[
U_1(\mathcal X,\mathcal X^*)
=
-\sum_{k=1,3,5}
a_{[2,k]}^*
\times
\frac{
B^{(B+k+1)/2}\Gamma\left(1+\frac{k}{2}\right)
}
{
2^{B/2-1}\sqrt{\pi}\Gamma\left(\frac{1+B+k}{2}\right)
}
\times
g_{[B,k]}\left(\frac{W_n^*}{\mathsf q}\right),
\]
where \(g_{[B,k]}\) is defined in \eqref{eq gbkx}. Since \(B\geq 2\), each \(g_{[B,k]}\in C^2(\mathbb R)\), and \(g_{[B,k]}\), \(g_{[B,k]}'\), and \(g_{[B,k]}''\) are uniformly bounded on \(\mathbb R\). Define \(\tilde U_1(\mathcal X,\mathcal X^*)\) by replacing \(p_2^*\) with \(\hat p_2\) in the definition of \(U_1(\mathcal X,\mathcal X^*)\) above. Each difference \(a_{[2,k]}^*-\hat a_{[2,k]}\) is a finite linear combination of differences of the form \(M_n^*-M_n\), where \(M_n^*\) and \(M_n\) are the resample and sample plug-in estimators of a monomial moment. Hence Proposition~\ref{theo:resample-expect-cancellation-good}, applied with \(\tilde A=A_s\) and with the even weight \(x\mapsto g_{[B,k]}(x/\mathsf q)\), gives
\[
\ind_{\mathcal G}\times
\sup_{\mathsf q\ge \varepsilon_0}
\left|
\mathbb E\left[
U_1(\mathcal X,\mathcal X^*)
-
\tilde U_1(\mathcal X,\mathcal X^*)
\mid
\mathcal X
\right]
\right|
\le
\frac1n
C_{[\varepsilon_0,\widehat{\mathsf C}_1,\delta,L,\|\nabla A_s(\bm\mu)\|,C_M,B,\{\mathbb E\|\mathbf X_1\|^k\}_{k=1,\ldots,l}]}.
\]
It remains to evaluate \(\tilde U_1\). Since \(\tilde U_1\) uses the sample-based coefficients \(\hat a_{[2,k]}\), its coefficients are fixed conditional on \(\mathcal X\). Applying the conditional Edgeworth expansion to \(W_n^*\) on \(\hat{\mathcal E}_{\nu=1}\), in the same way as in the population proof, gives
\[
\ind_{\mathcal G}\times
\sup_{\mathsf q\ge \varepsilon_0}
\left|
\mathbb E\left[
\tilde U_1(\mathcal X,\mathcal X^*)
\mid
\mathcal X
\right]
-
\hat\zeta_2^p(\mathsf q)
\right|
\le
\frac1n
C_{[\widehat{\mathsf C}_1,B]}.
\]
Here \(\hat\zeta_2^p\) denotes the sample-based analogue of the second summation term in \(\zeta_2^{\mathrm{ts}}\) in \eqref{eq zeta2}, obtained by replacing the population moment coefficients there with their sample plug-in estimators. Combining the last two displays yields
\[
\ind_{\mathcal G}\times
\sup_{\mathsf q\ge \varepsilon_0}
\left|
\mathbb E\left[
U_1(\mathcal X,\mathcal X^*)
\mid
\mathcal X
\right]
-
\hat\zeta_2^p(\mathsf q)
\right|
\le
\frac1n
C_{[\varepsilon_0,\widehat{\mathsf C}_1,\delta,L,\|\nabla A_s(\bm\mu)\|,C_M,B,\{\mathbb E\|\mathbf X_1\|^k\}_{k=1,\ldots,l}]}.
\]

\medskip
\noindent
\emph{Step 4: Conclusion.}
Combining the decomposition of \(U\), the bound for \(U_0\), and the bound for \(U_1\), we obtain on \(\mathcal G\)
\begin{align*}
&\sup_{\mathsf q\ge \varepsilon_0}
\left|
\PP\left(\right|T^*\left|\le \mathsf q\mid \mathcal X\right)
-
\left(2\Psi_B(\mathsf q)-1\right)
-
\frac1n\left(\hat\zeta_2^q(\mathsf q)+\hat\zeta_2^p(\mathsf q)\right)
\right|
\\
&\qquad\le
\frac1{n^2}
C_{[\varepsilon_0,\widehat{\mathsf C}_1,\hat{\mathsf C}_3,\delta,L,\|\nabla A_s(\bm\mu)\|,C_M,B,\{\mathbb E\|\mathbf X_1\|^k\}_{k=1,\ldots,l}]}.
\end{align*}
By definition, \(\hat\zeta_2^{\mathrm{ts}}(\mathsf q)
\coloneqq
\hat\zeta_2^q(\mathsf q)+\hat\zeta_2^p(\mathsf q)\). Therefore,
\begin{align*}
&\ind_{\mathcal G}\times
\sup_{\mathsf q\ge \varepsilon_0}
\left|
\PP\left(\right|T^*\left|\le \mathsf q\mid \mathcal X\right)
-
\left(2\Psi_B(\mathsf q)-1\right)
-
\frac1n\hat\zeta_2^{\mathrm{ts}}(\mathsf q)
\right|\\
\le &
\frac1{n^2}
C_{[\varepsilon_0,\widehat{\mathsf C}_1,\hat{\mathsf C}_3,\delta,L,\|\nabla A_s(\bm\mu)\|,C_M,B,\{\mathbb E\|\mathbf X_1\|^k\}_{k=1,\ldots,l}]}.
\end{align*}
Since \(\PP(\mathcal G^c)=O(n^{-\lambda})\), the claimed high-probability bound in \eqref{Eq resample ts uniform convergence} follows.

\subsection{Proof of the One-Sided Case in Theorem \ref{theorem resample uniform version}}\label{resample-one-side}
Recall the event \(\mathcal G\) defined in \eqref{eq good event resample G}. For a given quantile \(\mathsf q\in\mathbb R\), let \(Q^{**}\) denote the inner resample distribution and define
\begin{equation*}
V(\mathcal X,\mathcal X^*)
\coloneqq
\mathbb E\left[
\iiint
\ind\left\{
\frac{\sqrt n\,A_s^{*}(\bar{\mathbf X}_b^{*})}{\sqrt{\frac1B\sum_{b=1}^B z_b^2}}
\le \mathsf q
\right\}
dQ^{**}(z_B)\cdots dQ^{**}(z_1)
\mid
\mathcal X,\mathcal X^*
\right].
\end{equation*}
Then
\[
\PP(T^*\le \mathsf q\mid \mathcal X)
=
\mathbb E\left[V(\mathcal X,\mathcal X^*)\mid \mathcal X\right].
\]

\medskip
\noindent
\emph{Step 1: Reduction to \(V_0^*\) and \(V_1^*\).}
Analogously to the one-sided proof of Theorem~\ref{theorem uniform version}, on the corresponding inner resampling Edgeworth good event \(\mathcal E_3^*\) from Theorem~\ref{Theo HallHe}, we have
\begin{align*}
V(\mathcal X,\mathcal X^*)\times\ind_{\mathcal E_3^*}
=
\left(
V_0^*(\mathcal X,\mathcal X^*)
+
\frac{1}{n}\times V_1^*(\mathcal X,\mathcal X^*)
+
\hat R_{[8]}
\right)\times\ind_{\mathcal E_3^*},
\end{align*}
where
\[
\left|\hat R_{[8]}\times\ind_{\mathcal E_3^*}\right|
\le
C_{[\mathsf C_3^*,B]}\times n^{-2}.
\]

\medskip
\noindent
\emph{Step 2: Analysis of \(V_0^*\).}
Analogously to Step 2 of the one-sided proof, on \(\hat{\mathcal E}_{\nu=3}\),
\[
\mathbb E\left[V_0^*(\mathcal X,\mathcal X^*)\mid \mathcal X\right]
=
\Psi_B(\mathsf q)
+
\frac{\hat\zeta_{1}^{\mathrm{up}}(\mathsf q)}{\sqrt n}
+
\frac{\hat\zeta_{2,q}^{\mathrm{up}}(\mathsf q)}{n}
+
\hat R_{[9]},
\qquad
\abs{\hat R_{[9]}}
\le
C_{[\hat{\mathsf C}_3]}\times n^{-3/2}.
\]
Here \(\hat\zeta_{1}^{\mathrm{up}}\) and \(\hat\zeta_{2,q}^{\mathrm{up}}\) are obtained from the corresponding population terms by replacing the population moment coefficients with their sample plug-in estimators.

\medskip
\noindent
\emph{Step 3: Analysis of \(V_1^*\).}
For the \(V_1^*\) term, let \(\tilde V_1^*(\mathcal X,\mathcal X^*)\) be obtained by replacing \(p_2^*\) with \(\hat p_2\). Then, analogously to Step 3 of the one-sided proof and to the coefficient-control argument in the resampled two-sided proof, on \(\mathcal E_{\mathrm m}\),
\begin{align*}
    &\mathbb E\left[
\left|V_1^*(\mathcal X,\mathcal X^*)-\tilde V_1^*(\mathcal X,\mathcal X^*)\right|
\mid
\mathcal X
\right]\\
\le &
C_{[B,\{\mathbb E\|\mathbf X_1^*\|^k\}_{k=1,\ldots,l}]}\times n^{-1/2}
\le
C_{[B,\{\mathbb E\|\mathbf X_1\|^k\}_{k=1,\ldots,l}]}\times n^{-1/2}.
\end{align*}
Also, on \(\hat{\mathcal E}_{\nu=1}\),
\[
\mathbb E\left[\tilde V_1^*(\mathcal X,\mathcal X^*)\mid \mathcal X\right]
=
\hat\zeta_{2,p}^{\mathrm{up}}(\mathsf q)+\hat R_{[10]},
\qquad
\abs{\hat R_{[10]}}
\le
C_{[\hat{\mathsf C}_1,B]}\times n^{-1/2}.
\]
Here \(\hat\zeta_{2,p}^{\mathrm{up}}\) is obtained from the corresponding population term by replacing the population moment coefficients with their sample plug-in estimators. Therefore,
\[
\mathbb E\left[V_1^*(\mathcal X,\mathcal X^*)\mid \mathcal X\right]
=
\hat\zeta_{2,p}^{\mathrm{up}}(\mathsf q)+\hat R_{[11]},
\qquad
\abs{\hat R_{[11]}}
\le
n^{-1/2}\times
C_{[\hat{\mathsf C}_1,B,\{\mathbb E\|\mathbf X_1\|^k\}_{k=1,\ldots,l}]}.
\]

\medskip
\noindent
\emph{Step 4: Conclusion.}
Define $\hat\zeta_{2}^{\mathrm{up}}(\mathsf q)
\coloneqq
\hat\zeta_{2,q}^{\mathrm{up}}(\mathsf q)
+
\hat\zeta_{2,p}^{\mathrm{up}}(\mathsf q)$. Combining Steps 1--3, on \(\mathcal G\) we obtain
\begin{align*}
&\sup_{\mathsf q\in\mathbb R}
\left|
\PP(T^*\le \mathsf q\mid \mathcal X)
-
\Psi_B(\mathsf q)
-
\frac{\hat\zeta_{1}^{\mathrm{up}}(\mathsf q)}{\sqrt n}
-
\frac{\hat\zeta_{2}^{\mathrm{up}}(\mathsf q)}{n}
\right|
\le
n^{-3/2}\times
C_{[\hat{\mathsf C}_1,\hat{\mathsf C}_3,\mathsf C_3^*,B,\{\mathbb E\|\mathbf X_1\|^k\}_{k=1,\ldots,l}]}.
\end{align*}
This is the desired one-sided resampled expansion.

\section{Proof of the Cornish--Fisher Expansion for $t$-distributions}
In Appendix \ref{conjectural framework}, we discuss a conjectural general framework for higher-order Edgeworth expansions of Student's $t$-type pivots, in which the usual Edgeworth polynomials are replaced by quasi-polynomials with analogous structural properties. Based on this framework, we develop the corresponding general Cornish--Fisher expansions and establish their uniform approximation over $\alpha$ in compact subsets of $(0,1)$. The remaining appendices provide the detailed proofs supporting this framework.

Although the general framework is conjectural, the special cases needed in this paper, namely those with $\nu\leq 3$, are already validated by Theorem \ref{theorem uniform version}. Consequently, all theorems and propositions stated in the main text are established rigorously.

\subsection{Conjectural Framework for Edgeworth Expansions on $t$ Statistics}\label{conjectural framework}

\begin{definition}[Student's $t$ Edgeworth quasi-polynomials]
Fix $B\ge 1$ and an integer $k\ge 1$. Let $\mathcal Q_{B,k}$ denote the class of all functions $\eta_k(\mathsf q)$ that can be written as a \emph{finite} linear combination of the form
\[
\eta_k(\mathsf q)
=
\sum
c_{[k,m,\ell]}
\frac{\mathsf q^m}{(B+\mathsf q^2)^{\ell/2}},
\]
where $c_{[k,m,\ell]}\in\mathbb R$, $m\in \mathbb Z_{\ge 0}$, and $\ell\in \mathbb Z$, subject to the constraint $m\le 3k-1$. In addition, $m$ is odd when $k$ is even, and even when $k$ is odd. Elements of $\mathcal Q_{B,k}$ are called \emph{Student's $t$ Edgeworth quasi-polynomials of order $k$}. We also define
\[
\mathcal Q_B \coloneqq \bigcup_{k\ge 1}\mathcal Q_{B,k}.
\]
\end{definition}
These quasi-polynomials resemble ordinary polynomials and enjoy several useful structural properties.
\begin{lemma}\label{lem:tqp-closure}
For any fixed integer $B\geq1$ and each integer $k\ge 1$, $\mathcal Q_{B,k}\subset C^\infty(\mathbb R)$. Moreover, the class $\mathcal Q_B$ is closed under multiplication and differentiation. In particular:
\begin{itemize}
    \item if $f,g\in \mathcal Q_B$, then $fg\in \mathcal Q_B$, with parity determined by the usual product rule;
    \item if $f\in \mathcal Q_B$, then $f'\in \mathcal Q_B$, and differentiation reverses parity;
    \item the density function $\psi_B(\cdot)$ belongs to $\mathcal Q_B$.
\end{itemize}
\end{lemma}

We \emph{conjecture} that the general Student's $t$ Edgeworth expansion admits the following quasi-polynomial form:
\begin{align}
G_{n,\nu}^{\mathrm{up}}(\mathsf q)
&\coloneqq 
\PsiB(\mathsf q)+\sum_{k=1}^{\nu} n^{-k/2}\eta_k^{\mathrm{up}}(\mathsf q)\psiB(\mathsf q),
\label{eq:Gup}
\\
G_{n,\nu}^{\mathrm{ts}}(\mathsf q)
&\coloneqq
2\PsiB(\mathsf q)-1+\sum_{\substack{1\le k\le \nu\\ k\text{ even}}} n^{-k/2}\eta_k^{\mathrm{ts}}(\mathsf q)\psiB(\mathsf q),
\label{eq:Gts}
\end{align}
where $\eta_k^{\mathrm{up}},\eta_k^{\mathrm{ts}}\in\mathcal Q_B$. Moreover, each coefficient $c_{[k,m,\ell]}$ in $\eta_k(\mathsf q)=\sum c_{[k,m,\ell]}\times \mathsf q^m/(B+\mathsf q^2)^{\ell/2}$ is required to be a polynomial in finitely many moments of the underlying distribution, up to order $k+2$.

Although this is stated as a conjecture, the expansion has already been verified in Theorem \ref{theorem uniform version} for the two-sided case when $\nu\le 3$ and for the upper one-sided case when $\nu\le 2$. We develop the general Cornish--Fisher expansion for Student's $t$ quasi-polynomials based on this conjectural form. Nevertheless, the special cases established in Theorem \ref{theorem uniform version} already suffice for the coverage error analysis of our studentized cheap bootstrap.

\begin{assumption}\label{assume 100}
Assume Assumption \ref{assume1}. In addition, assume that the following $t$-Edgeworth expansions hold for any fixed $B\geq1$ in the one-sided expansion and any fixed $B\geq2$ in the two-sided expansion:
\begin{align}
\sup_{\mathsf q\in\mathbb R}
\left|
\PP(T\le \mathsf q)-G_{n,\nu}^{\mathrm{up}}(\mathsf q)
\right|
&\le
\mathsf C_{\mathrm{up}}\times n^{-(\nu+1)/2},
\\
\sup_{\mathsf q\ge \varepsilon_0}
\left|
\PP(\abs{T}\le \mathsf q)-G_{n,\nu}^{\mathrm{ts}}(\mathsf q)
\right|
&\le
\mathsf C_{\mathrm{ts}}^{\varepsilon_0}\times n^{-(\nu+1)/2}.
\label{eq:up-edge}
\end{align}
\end{assumption}

Let $\hat G_{n,\nu}^{\mathrm{up}}$ and $\hat G_{n,\nu}^{\mathrm{ts}}$ denote the formal expansions obtained from $G_{n,\nu}^{\mathrm{up}}$ and $G_{n,\nu}^{\mathrm{ts}}$ by replacing, in every coefficient, the population moments by the corresponding sample moments.

\begin{assumption}\label{assume 200}
Assume Assumption \ref{assume1}. Assume further that, for any fixed $B\geq1$ in the one-sided expansion and any fixed $B\geq2$ in the two-sided expansion,
\begin{align}
\PP\left(
\sup_{\mathsf q\in\mathbb R}
\left|
\PP(T^*\le \mathsf q\mid \mathcal X)-\hat G_{n,\nu}^{\mathrm{up}}(\mathsf q)
\right|
\ge
{\hat{\mathsf C}_{\mathrm{up}}}\times{n^{-(\nu+1)/2}}
\right)
&=
O(n^{-\lambda}),
\\
\PP\left(
\sup_{\mathsf q\ge \varepsilon_0}
\left|
\PP(\abs{T^*}\le \mathsf q\mid \mathcal X)-\hat G_{n,\nu}^{\mathrm{ts}}(\mathsf q)
\right|
\ge
{\hat{\mathsf C}_{\mathrm{ts}}^{\varepsilon_0}}\times{n^{-(\nu+1)/2}}
\right)
&=
O(n^{-\lambda}).
\end{align}
Moreover, the coefficients $\hat c_{[k,m,\ell]}$ of each $\hat\eta_k$ are polynomials in sample moments up to order $k+2$.
\end{assumption}
Similar to the usual Cornish--Fisher expansions, we expect these quantiles to admit formal approximations in terms of quasi-polynomials:
\begin{align}
u_{\nu,\alpha}^{\mathrm{up}}
&\coloneqq
\mathsf q+\sum_{k=1}^{\nu} n^{-k/2}\xi_k^{\mathrm{up}}(\mathsf q),
\qquad \text{where }\mathsf q\coloneqq\PsiB^{-1}(\alpha),
\label{def cfup}
\\
u_{\nu,\alpha}^{\mathrm{ts}}
&\coloneqq
\mathsf q+\sum_{\substack{2\le k\le \nu\\ k\text{ even}}} n^{-k/2}\xi_k^{\mathrm{ts}}(\mathsf q),
\qquad \text{where }\mathsf q\coloneqq\PsiB^{-1}\!\left(\frac{1+\alpha}{2}\right).
\label{def cfts}
\end{align}
Here $\xi_k^\mathrm{ts},\xi_k^\mathrm{up}\in \QB$, with $\xi_k$ even when $k$ is odd and odd when $k$ is even. Moreover, the coefficients of $\xi_k$ are polynomials in moments up to order $k+2$. The next proposition confirms this formal construction.

\begin{proposition}[Existence and uniqueness of the Cornish--Fisher quasi-polynomials]\label{prop:CF-existence}
Under Assumption~\ref{assume 100}, the following statements hold.

\begin{itemize}
    \item For any fixed $B\geq1$, there exist unique functions $\xi_1^{\mathrm{up}},\dots,\xi_\nu^{\mathrm{up}}\in \QB$ such that, for every compact interval $K\subset \R$,
\begin{align}
\sup_{\mathsf q\in K}
\left|
G_{n,\nu}^{\mathrm{up}}\!\left(\mathsf q+\sum_{k=1}^{\nu} n^{-k/2}\xi_k^{\mathrm{up}}(\mathsf q)\right)-\PsiB(\mathsf q)
\right|
\le
\mathsf A_K^{\mathrm{up}}\times n^{-(\nu+1)/2}
\label{eq:formal-up}
\end{align}
for some finite constant $\mathsf A_K^{\mathrm{up}}$. Moreover, $\xi_k^{\mathrm{up}}$ is odd when $k$ is even, and even when $k$ is odd.

    \item For any fixed $B\geq2$, there exist unique functions $\xi_k^{\mathrm{ts}}\in \QB$ for all even integers $k$ with $2\le k\le \nu$ ($\nu$ is odd) such that, for every compact interval $K\subset (0,\infty)$,
\begin{align}
\sup_{\mathsf q\in K}
\left|
G_{n,\nu}^{\mathrm{ts}}\!\left(\mathsf q+\sum_{\substack{2\le k\le \nu\\ k\text{ even}}} n^{-k/2}\xi_k^{\mathrm{ts}}(\mathsf q)\right)-\bigl(2\PsiB(\mathsf q)-1\bigr)
\right|
\le
\mathsf A_K^{\mathrm{ts}}\times n^{-(\nu+1)/2}
\label{eq:formal-ts}
\end{align}
for some finite constant $\mathsf A_K^{\mathrm{ts}}$. Moreover, each $\xi_k^{\mathrm{ts}}$ is odd.
\end{itemize}
\end{proposition}
We next introduce the resampled analogues. For $0<\alpha<1$, define the resampled quantiles by
\begin{align*}
\hat u_\alpha^{\mathrm{up}}
\coloneqq
\inf\{\mathsf q\in \R:\PP(T^*\le \mathsf q \mid \mathcal X)\ge \alpha\}
\qquad\text{and}\qquad
\hat u_\alpha^{\mathrm{ts}}
\coloneqq
\inf\{\mathsf q\ge 0:\PP(\abs{T^*}\le \mathsf q\mid \mathcal X)\ge \alpha\}.
\end{align*}
Let
\[
\hat\xi_1^{\mathrm{up}},\dots,\hat\xi_\nu^{\mathrm{up}}\in\QB
\qquad\text{and}\qquad
\hat\xi_k^{\mathrm{ts}}\in\QB,\quad 2\le k\le \nu,\ k\text{ even},
\]
denote the formal Cornish--Fisher quasi-polynomials obtained from
\[
\xi_1^{\mathrm{up}},\dots,\xi_\nu^{\mathrm{up}}
\qquad\text{and}\qquad
\xi_k^{\mathrm{ts}},\quad 2\le k\le \nu,\ k\text{ even},
\]
by replacing, in every coefficient, the population moments by the corresponding sample moments. Equivalently, the coefficients of $\hat\xi_k$ are polynomials in sample moments up to order $k+2$, with exactly the same functional form as in the population version. The corresponding formal resampled approximants are denoted by $\hat u_{\nu,\alpha}^{\mathrm{up}}$ and $\hat u_{\nu,\alpha}^{\mathrm{ts}}$, defined in the same way as in \eqref{def cfup} and \eqref{def cfts}, with $\xi_k$ replaced by $\hat\xi_k$.

The uniform approximation error of the Cornish--Fisher expansions with quasi-polynomials is given by the following two theorems.

\begin{theorem}[Uniform Cornish--Fisher expansion]\label{thm:cf uniform}
Let $\varepsilon\in(0,1/2)$. Under Assumption \ref{assume 100}, for any fixed $B\geq1$, there exists a finite constant $\mathsf C_{\mathrm{c.f.}}^{\mathrm{up}}$ such that
\begin{align*}
    \sup_{\varepsilon\leq\alpha\leq 1-\varepsilon}
    \left|u_{\nu,\alpha}^{\mathrm{up}}-u_\alpha^{\mathrm{up}}\right|
    \le
    \mathsf C_{\mathrm{c.f.}}^{\mathrm{up}} n^{-(\nu+1)/2}.
\end{align*}
When $\nu$ is odd and $B\geq2$ is fixed, there exists a finite constant $\mathsf C_{\mathrm{c.f.}}^{\mathrm{ts}}$ such that
\begin{align*}
    \sup_{\varepsilon\leq\alpha\leq 1-\varepsilon}
    \left|u_{\nu,\alpha}^{\mathrm{ts}}-u_\alpha^{\mathrm{ts}}\right|
    \le
    \mathsf C_{\mathrm{c.f.}}^{\mathrm{ts}} n^{-(\nu+1)/2}.
\end{align*}
The specific forms of the constants are given in \eqref{cf constants1} and \eqref{cf constants2}.
\end{theorem}

\begin{theorem}[Resampled uniform Cornish--Fisher expansion]\label{thm:resample cf uniform}
Let $\varepsilon\in(0,1/2)$. Under Assumption \ref{assume 200}, for any fixed $B\geq1$, there exists a finite constant $\hat{\mathsf C}_{\mathrm{c.f.}}^{\mathrm{up}}$ such that
\begin{align*}
    \PP\left(
    \sup_{\varepsilon\leq\alpha\leq 1-\varepsilon}
    \left|\hat u_{\nu,\alpha}^{\mathrm{up}}-\hat u_\alpha^{\mathrm{up}}\right|
    \ge
    \hat{\mathsf C}_{\mathrm{c.f.}}^{\mathrm{up}} n^{-(\nu+1)/2}
    \right)
    =
    O(n^{-\lambda}).
\end{align*}
When $\nu$ is odd and $B\geq2$ is fixed, there exists a finite constant $\hat{\mathsf C}_{\mathrm{c.f.}}^{\mathrm{ts}}$ such that
\begin{align*}
    \PP\left(
    \sup_{\varepsilon\leq\alpha\leq 1-\varepsilon}
    \left|\hat u_{\nu,\alpha}^{\mathrm{ts}}-\hat u_\alpha^{\mathrm{ts}}\right|
    \ge
    \hat{\mathsf C}_{\mathrm{c.f.}}^{\mathrm{ts}} n^{-(\nu+1)/2}
    \right)
    =
    O(n^{-\lambda}).
\end{align*}
\end{theorem}

\subsection{Proof of the One-Sided Case in Proposition \ref{prop:CF-existence}}\label{sec51}
~

\medskip
\noindent
\emph{Notations.}
Define
\[
\zeta_k^{\mathrm{up}}(\mathsf q)\coloneqq \eta_k^{\mathrm{up}}(\mathsf q)\psiB(\mathsf q).
\]
By Lemma \ref{lem:tqp-closure}, each $\zeta_k^{\mathrm{up}}\in\QB$, and so do its derivatives:
\[
\zeta_{k,[r]}^{\mathrm{up}}(\mathsf q)\coloneqq\left(\frac{d}{d \mathsf q}\right)^r \zeta_k^{\mathrm{up}}(\mathsf q) \in \QB
\]
Define
\[
\Delta_n^{\mathrm{up}}(\mathsf q)
\coloneqq 
\sum_{k=1}^{\nu} n^{-k/2}\xi_k^{\mathrm{up}}(\mathsf q).
\]
Then $\Delta_n^{\mathrm{up}}(\mathsf q)=O(n^{-1/2})$ uniformly on $K$, since each $\xi_k^{\mathrm{up}}\in\QB$ is smooth.

\medskip
\noindent
\emph{Step 1: Taylor expansion.}
Recall that
\begin{align*}
    G_{n,\nu}^{\mathrm{up}}(\mathsf q)
\coloneqq
\PsiB(\mathsf q)+\sum_{k=1}^{\nu} n^{-k/2}\eta_k^{\mathrm{up}}(\mathsf q)\psiB(\mathsf q)
\end{align*}
The Taylor expansion of
\begin{align*}
& G_{n,\nu}^{\mathrm{up}}(\mathsf q+\Delta_n^{\mathrm{up}}(\mathsf q))-\PsiB(\mathsf q)\\
=~&\PsiB(\mathsf q+\Delta_n^{\mathrm{up}}(\mathsf q))-\PsiB(\mathsf q)
+\sum_{k=1}^{\nu}n^{-k/2}\zeta_k^{\mathrm{up}}(\mathsf q+\Delta_n^{\mathrm{up}}(\mathsf q))\\
=~&
\sum_{r=1}^{\nu}\frac{\PsiB^{(r)}(\mathsf q)}{r!}\bigl(\Delta_n^{\mathrm{up}}(\mathsf q)\bigr)^r
+\sum_{k=1}^{\nu}n^{-k/2}\sum_{r=0}^{\nu-k}\frac{\zeta_{k,[r]}^{\mathrm{up}}(\mathsf q)}{r!}\bigl(\Delta_n^{\mathrm{up}}(\mathsf q)\bigr)^r + R_{[12]}.    
\end{align*}
Since all functions involved are infinitely smooth on the compact set $K$, we have
\begin{align*}
    \sup_{\mathsf q\in K}\bigl|\Delta_n^{\mathrm{up}}(\mathsf q)\bigr|=O(n^{-1/2}).
\end{align*}

\medskip
\noindent
\emph{Step 2: Bounding the remainder.}
For later use in proving the resampled version, we record the more explicit bound
\begin{equation}\label{eq18}
\begin{split}
      \sup_{\mathsf q\in K}\bigl|\Delta_n^{\mathrm{up}}(\mathsf q)\bigr|\leq & \sum_{k=1}^{\nu} n^{-k/2}\sup_{\mathsf q\in K} \abs{\xi_k^{\mathrm{up}}(\mathsf q)}
    \leq  \sum_{k=1}^{\nu}\sum_{[m,\ell]} n^{-k/2}c_{[k,m,\ell]}\times \sup_{\mathsf q\in K}\left|{{\mathsf q^m}}{(B+\mathsf q^2)^{-\ell/2}}\right|\\
    = & n^{-1/2}\times C_{[B,K,\nu,\{\E\|\mathbf X_1\|^k\}_{k=1,\ldots,l}]}.
\end{split}
\end{equation}
The remainder term $R_{[12]}$ consists of the remainder $O(\bigl|\Delta_n^{\mathrm{up}}\bigr|^{\nu+1})$ from $\PsiB(\mathsf q+\Delta_n^{\mathrm{up}}(\mathsf q))$ as well as $n^{-k/2}\times O(\bigl|\Delta_n^{\mathrm{up}}\bigr|^{\nu+1-k})$ from $\zeta_{k}$. By the Lagrange remainder formula,
\begin{align*}
    \abs{R_{[12]}}\leq \sum_{k=1}^\nu n^{-k/2}\times O(\bigl|\Delta_n^{\mathrm{up}}\bigr|^{\nu+1-k})\times \sup_{\mathsf q\in K}\left|\zeta_{k,[\nu-k+1]}^{\mathrm{up}}(x)\right|
\end{align*}
where $\mathsf q<x<\mathsf q+\Delta_n^{\mathrm{up}}(\mathsf q)$. Because $\Delta_n^{\mathrm{up}}(\mathsf q)=O(n^{-1/2})$ uniformly on $K$, we have $x\in K$ except for finitely many terms. Therefore,
\[
\sup_{\mathsf q\in K}\left|\zeta_{k,[\nu-k+1]}^{\mathrm{up}}(x)\right|\leq \sup_{x\in K}\left|\zeta_{k,[\nu-k+1]}^{\mathrm{up}}(x)\right|,
\]
which is another constant of the form $C_{[B,K,\nu,\{\E\|\mathbf X_1\|^k\}_{k=1,\ldots,l}]}$. Hence,
\begin{align*}
    \abs{R_{[12]}}\leq& \sum_{k=1}^\nu n^{-k/2}\times O(n^{-{(\nu+1-k})/2})\times C_{[B,K,\nu,\{\E\|\mathbf X_1\|^k\}_{k=1,\ldots,l}]}\\
    = & n^{-{(\nu+1})/2}\times C_{[B,K,\nu,\{\E\|\mathbf X_1\|^k\}_{k=1,\ldots,l}]}
\end{align*}

\medskip
\noindent
\emph{Step 3: Comparison of each order.}
In order for \eqref{eq:formal-up} to hold, we must have
\begin{align*}
    &\sum_{r=1}^{\nu}\frac{\PsiB^{(r)}(\mathsf q)}{r!}\bigl(\Delta_n^{\mathrm{up}}(\mathsf q)\bigr)^r
 +\sum_{k=1}^{\nu}n^{-k/2}\sum_{r=0}^{\nu-k}\frac{\zeta_{k,[r]}^{\mathrm{up}}(\mathsf q)}{r!}\bigl(\Delta_n^{\mathrm{up}}(\mathsf q)\bigr)^r=O(n^{-{(\nu+1})/2})\\
 \iff\quad & \psiB(\mathsf q)\bigl(\Delta_n^{\mathrm{up}}(\mathsf q)\bigr)=-\sum_{r\geq 2}^{\nu}\frac{\PsiB^{(r)}(\mathsf q)}{r!}\bigl(\Delta_n^{\mathrm{up}}(\mathsf q)\bigr)^r
 -\sum_{k=1}^{\nu}n^{-k/2}\sum_{r=0}^{\nu-k}\frac{\zeta_{k,[r]}^{\mathrm{up}}(\mathsf q)}{r!}\bigl(\Delta_n^{\mathrm{up}}(\mathsf q)\bigr)^r\\
 &+O(n^{-{(\nu+1})/2})
\end{align*}
after substituting $\Delta_n^{\mathrm{up}}(\mathsf q)
\coloneqq
\sum_{k=1}^{\nu} n^{-k/2}\xi_k^{\mathrm{up}}(\mathsf q)$. Now fix $1\le s\le \nu$, and collect the coefficient of $n^{-s/2}$. The only linear term involving $n^{-s/2}\xi_s^{\mathrm{up}}(\cdot)\psi_B(\cdot)$ comes from LHS when $r=1$. All other terms of order $n^{-s/2}$ involve $\eta_k^{\mathrm{up}}$, $\zeta_{k,[r]}^{\mathrm{up}}$, $\PsiB^{(r)}$, and
\begin{align*}
\xi_1^{\mathrm{up}},\dots,\xi_{s-1}^{\mathrm{up}}.
\end{align*}
It is impossible for $\xi_{s'}^{\mathrm{up}}$ with $s'\geq s$ to appear in the coefficient of order $n^{-s/2}$, unless $s'=s$ through the linear term. Therefore, each $\xi_k^{\mathrm{up}}$ can be solved inductively. By Lemma \ref{lem:tqp-closure}, $\QB$ is closed under differentiation and multiplication, though not automatically under addition. Thus, to conclude that $\xi_s^{\mathrm{up}}\in \QB$, it remains to verify the required parity property.

\medskip
\noindent
\emph{Step 4: Preservation of parity.}
We prove the parity statement by induction. The case $s=1$ is immediate from $\xi_1^{\mathrm{up}}=-\eta_1^{\mathrm{up}}$. Assume that, for all $k<s$, $\xi_k^{\mathrm{up}}$ is even when $k$ is odd and odd when $k$ is even. The relevant terms in the coefficient of order $O(n^{-s/2})$ are
\begin{align*}
    f_1(\mathsf q)\coloneqq&\PsiB^{(r)}(\mathsf q)\prod_{i=1}^r \xi_{j_i}^{\mathrm{up}}(\mathsf q),
\qquad j_1+\cdots+j_r=s,\\
f_2(\mathsf q)\coloneqq&\zeta_{k,[r]}^{\mathrm{up}}(\mathsf q)\prod_{i=1}^r \xi_{j_i}^{\mathrm{up}}(\mathsf q),
\qquad
k+j_1+\cdots+j_r=s,
\end{align*}
We note that $\PsiB^{(r)}$ has parity $r+1$ modulo $2$, while $\zeta_{k,[r]}^{\mathrm{up}}$ has parity $k+1+r$ modulo $2$. Therefore, under the induction hypothesis, the parity of $f_1$ is
\begin{align*}
    r+1+\sum_{i=1}^r(j_i+1)= 2r+s+1= s+1 \pmod 2
\end{align*}
and the parity of $f_2$ is
\begin{align*}
    k+1+r+\sum_{i=1}^r(j_i+1) = s+1+2r = s+1 \pmod 2
\end{align*}
This completes the proof of the one-sided case.

\subsection{Proof of the Two-Sided Case in Proposition \ref{prop:CF-existence}}\label{sec52}
The argument is largely analogous to that of the one-sided case. Hence we only provide a concise explanation for each step. Throughout this proof, $\nu$ is an odd integer.

\medskip
\noindent
\emph{Notations.}
Define, in general, that
\begin{align*}
    \zeta_k^{\mathrm{ts}}(\mathsf q)\coloneqq&\eta_k^{\mathrm{ts}}(\mathsf q)\psiB(\mathsf q)\\
    \zeta_{k,[r]}^{\mathrm{ts}}(\mathsf q)\coloneqq&\left(\frac{d}{d \mathsf q}\right)^r \zeta_k^{\mathrm{ts}}(\mathsf q) \in \QB\\
    \Delta_n^{\mathrm{ts}}(\mathsf q)
\coloneqq &
\sum_{\substack{2\le k\le \nu\\ k\text{ even}}} n^{-k/2}\xi_k^{\mathrm{ts}}(\mathsf q).
\end{align*}
Then $\Delta_n^{\mathrm{ts}}(\mathsf q)=O(n^{-1})$ uniformly on $K$, by the same smooth-on-compact argument.

\medskip
\noindent
\emph{Step 1: Taylor expansion.}
The corresponding Taylor expansion is
\begin{align*}
& G_{n,\nu}^{\mathrm{ts}}(\mathsf q+\Delta_n^{\mathrm{ts}}(\mathsf q))-(2\PsiB(\mathsf q)-1)\\
=~&
2\sum_{r=1}^{(\nu-1)/2}\frac{\PsiB^{(r)}(\mathsf q)}{r!}\bigl(\Delta_n^{\mathrm{ts}}(\mathsf q)\bigr)^r
+\sum_{\substack{2\le k\le \nu\\ k\text{ even}}}n^{-k/2}\sum_{r=0}^{(\nu-k-1)/2}\frac{\zeta_{k,[r]}^{\mathrm{ts}}(\mathsf q)}{r!}\bigl(\Delta_n^{\mathrm{ts}}(\mathsf q)\bigr)^r + R_{[12]}.    
\end{align*}
where we have $\sup_{\mathsf q\in K}\bigl|\Delta_n^{\mathrm{ts}}(\mathsf q)\bigr|=O(n^{-1})$.

\medskip
\noindent
\emph{Step 2: Bounding the remainder.}
By an almost identical argument,
\begin{align*}
    \abs{R_{[12]}}\leq & 
  n^{-(\nu+1)/2}\times C_{[B,K,\nu,\{\E\|\mathbf X_1\|^k\}_{k=1,\ldots,l}]}
\end{align*}

\medskip
\noindent
\emph{Step 3: Comparison of each order.}
In order for \eqref{eq:formal-ts} to hold, we must have
\begin{align*}
    &2\sum_{r=1}^{(\nu-1)/2}\frac{\PsiB^{(r)}(\mathsf q)}{r!}\bigl(\Delta_n^{\mathrm{ts}}(\mathsf q)\bigr)^r
 +\sum_{\substack{2\le k\le \nu\\ k\text{ even}}}n^{-k/2}\sum_{r=0}^{(\nu-k-1)/2}\frac{\zeta_{k,[r]}^{\mathrm{ts}}(\mathsf q)}{r!}\bigl(\Delta_n^{\mathrm{ts}}(\mathsf q)\bigr)^r\\
 &=O(n^{-(\nu+1)/2})\\
 \iff\quad & 2\psiB(\mathsf q)\bigl(\Delta_n^{\mathrm{ts}}(\mathsf q)\bigr)=-2\sum_{r=2}^{(\nu-1)/2}\frac{\PsiB^{(r)}(\mathsf q)}{r!}\bigl(\Delta_n^{\mathrm{ts}}(\mathsf q)\bigr)^r\\
 &-\sum_{\substack{2\le k\le \nu\\ k\text{ even}}}n^{-k/2}\sum_{r=0}^{(\nu-k-1)/2}\frac{\zeta_{k,[r]}^{\mathrm{ts}}(\mathsf q)}{r!}\bigl(\Delta_n^{\mathrm{ts}}(\mathsf q)\bigr)^r+O(n^{-(\nu+1)/2})
\end{align*}
after substituting 
\begin{align*}
\Delta_n^{\mathrm{ts}}(\mathsf q)
\coloneqq
\sum_{\substack{2\le k\le \nu\\ k\text{ even}}} n^{-k/2}\xi_k^{\mathrm{ts}}(\mathsf q).   
\end{align*}
Again, fix an even integer $2\le s< \nu$ and collect the coefficient of $n^{-s/2}$. The only linear term involving $n^{-s/2}\xi_s^{\mathrm{ts}}(\cdot)\psi_B(\cdot)$ comes from the LHS when $r=1$. The other terms of order $n^{-s/2}$ contain $\eta_k^{\mathrm{ts}}$, $\zeta_{k,[r]}^{\mathrm{ts}}$, $\PsiB^{(r)}$, and
\begin{align*}
\xi_2^{\mathrm{ts}},\xi_4^{\mathrm{ts}},\dots,\xi_{s-2}^{\mathrm{ts}}.
\end{align*}
Therefore, each $\xi_k^{\mathrm{ts}}$ can be solved inductively.

\medskip
\noindent
\emph{Step 4: Preservation of parity.}
We use a similar induction to prove the parity statement. The case $s=2$ is immediate from
\[
\xi_2^{\mathrm{ts}}=-\frac{1}{2}\eta_2^{\mathrm{ts}}.
\]
Assume the claim that $\xi_k^{\mathrm{ts}}$ is odd holds for all even integers $2\le k<s$. The relevant terms used to calculate $\xi_s^{\mathrm{ts}}$ with order $O(n^{-s/2})$ are
\begin{align*}
    f_1(\mathsf q)\coloneqq&\PsiB^{(r)}(\mathsf q)\prod_{i=1}^r \xi_{j_i}^{\mathrm{ts}}(\mathsf q),
\qquad j_1+\cdots+j_r=s,\\
f_2(\mathsf q)\coloneqq&\zeta_{k,[r]}^{\mathrm{ts}}(\mathsf q)\prod_{i=1}^r \xi_{j_i}^{\mathrm{ts}}(\mathsf q),
\qquad
k+j_1+\cdots+j_r=s,
\end{align*}
where all indices $j_i$ and $s$ are even. Again, since both $\PsiB^{(r)}$ and $\zeta_{k,[r]}^{\mathrm{ts}}$ have parity $r+1$ modulo 2, both $f_1$ and $f_2$ have parity
\begin{align*}
    r+1+\sum_{i=1}^r1= 2r+1\equiv 1 \pmod 2.
\end{align*}
Hence $\xi_s^\mathrm{ts}$ is always odd. This completes the two-sided case.
\subsection{Auxiliary Lemmas for Theorem \ref{thm:cf uniform}}\label{sec5.3}
Define
\begin{align*}
F_n^{\mathrm{up}}(\mathsf q)\coloneqq \PP(T\le \mathsf q)
\qquad\text{and}\qquad
F_n^{\mathrm{ts}}(\mathsf q)\coloneqq \PP(\abs{T}\le \mathsf q).
\end{align*}
Recall that the true quantiles are defined as
\begin{align*}
     u_\alpha^{\mathrm{up}}
\coloneqq&
\inf\{\mathsf q\in \R:F_n^{\mathrm{up}}(\mathsf q)\ge \alpha\}\qquad\text{and}\qquad
u_\alpha^{\mathrm{ts}}
\coloneqq 
\inf\{\mathsf q\ge 0:F_n^{\mathrm{ts}}(\mathsf q)\ge \alpha\}.
\end{align*}
for $0<\alpha<1$, where $F_n^{\mathrm{up}}(\mathsf q)\coloneqq \PP(T\le \mathsf q)$ and $F_n^{\mathrm{ts}}(\mathsf q)\coloneqq \PP(\abs{T}\le \mathsf q)$. Recall that $I_\varepsilon\coloneqq [\varepsilon,1-\varepsilon]$ for $\varepsilon\in(0,1/2)$ and define the corresponding compact interval $K^\mathrm{up}_\varepsilon\subset\mathbb{R}$ by
\begin{align*}
    K^\mathrm{up}_\varepsilon\coloneqq [\PsiB^{-1}(\varepsilon),\PsiB^{-1}(1-\varepsilon)].
\end{align*}
For $\varepsilon\in(0,1/4)$, define the corresponding compact interval $K_{\varepsilon}^{\mathrm{ts}}\subset\mathbb{R}_+$ by
\begin{align*}
    K_{\varepsilon}^{\mathrm{ts}}
\coloneqq
\left[
\PsiB^{-1}\!\left(\frac12+\frac{\varepsilon}{2}\right),
\PsiB^{-1}\!\left(1-\frac{\varepsilon}{2}\right)
\right].
\end{align*}
In order to show that
\begin{align*}
    \sup_{\alpha\in I_\varepsilon}\left|u_\alpha^{\cdots}-u_{\nu,\alpha}^{\cdots}\right|
    \le O(n^{-(\nu+1)/2}),
\end{align*}
we shall employ intermediate quantities $\widetilde u_{\nu,\alpha}^{\mathrm{up}}$ and $\widetilde u_{\nu,\alpha}^{\mathrm{ts}}$, which are defined as the quantiles of the finite expansion approximants defined in \eqref{eq:Gts}, i.e.
\begin{align*}
G_{n,\nu}^{\mathrm{up}}(\mathsf q)
&\coloneqq
\PsiB(\mathsf q)+\sum_{k=1}^{\nu} n^{-k/2}\eta_k^{\mathrm{up}}(\mathsf q)\psiB(\mathsf q),
\\
G_{n,\nu}^{\mathrm{ts}}(\mathsf q)
&\coloneqq
2\PsiB(\mathsf q)-1+\sum_{\substack{1\le k\le \nu\\ k\text{ even}}} n^{-k/2}\eta_k^{\mathrm{ts}}(\mathsf q)\psiB(\mathsf q).
\end{align*}
The idea is to show both
\begin{align*}
    \sup_{\alpha\in I_\varepsilon}\left|u_\alpha^{\cdots}-\widetilde u_{\nu,\alpha}^{\cdots}\right|
    \le O(n^{-(\nu+1)/2})
    \qquad\text{and}\qquad
    \sup_{\alpha\in I_\varepsilon}\left|\widetilde u_{\nu,\alpha}^{\cdots}-u_{\nu,\alpha}^{\cdots}\right|
    \le O(n^{-(\nu+1)/2}).
\end{align*}
However, these two intermediate quantiles are not automatically well defined, since the functions $G_{n,\nu}^{\mathrm{up}}$ and $G_{n,\nu}^{\mathrm{ts}}$ are not genuine distribution functions. The next lemma addresses this issue.

\begin{lemma}[Well-definedness of the intermediate quantiles]\label{lm inter}
For any fixed $B\geq1$ in part~(i) and any fixed $B\geq2$ in part~(ii), for every $\varepsilon\in(0,1/6)$, there exists a constant $n_0(\varepsilon)\in\mathbb N$ such that, for all $n\ge n_0(\varepsilon)$, the following hold.
\begin{itemize}
    \item[(i)] The function $G_{n,\nu}^{\mathrm{up}}$ is strictly increasing on $K^\mathrm{up}_{\varepsilon}$ and satisfies
    \begin{align*}
       \bigl(G_{n,\nu}^{\mathrm{up}}\bigr)'(\mathsf q)\ge \frac12\inf_{\mathsf q\in K^\mathrm{up}_{\varepsilon}}\psiB(\mathsf q).
    \end{align*}
    Hence, for every $\alpha\in I_{3\varepsilon}$, there exists a unique $\widetilde u_{\nu,\alpha}^{\mathrm{up}}\in K_{2\varepsilon}^{\mathrm{up}}$ such that
    \begin{align}\label{eq def interup}
         G_{n,\nu}^{\mathrm{up}}(\widetilde u_{\nu,\alpha}^{\mathrm{up}})=\alpha.
    \end{align}
    Equation~\eqref{eq def interup} therefore defines $\widetilde u_{\nu,\alpha}^{\mathrm{up}}$ for all sufficiently large $n$.
    \item[(ii)] The function $G_{n,\nu}^{\mathrm{ts}}$ is strictly increasing on $K_{\varepsilon}^{\mathrm{ts}}$ and satisfies
    \begin{align*}
    \bigl(G_{n,\nu}^{\mathrm{ts}}\bigr)'(\mathsf q)
    \ge
    \inf_{\mathsf q\in K_{\varepsilon}^{\mathrm{ts}}}\psiB(\mathsf q).
    \end{align*}
    Hence, for every $\alpha\in I_{3\varepsilon}$, there exists a unique $\widetilde u_{\nu,\alpha}^{\mathrm{ts}}\in K_{2\varepsilon}^{\mathrm{ts}}$ such that
    \begin{align}\label{eq:def-implicit-ts-lemma}
    G_{n,\nu}^{\mathrm{ts}}(\widetilde u_{\nu,\alpha}^{\mathrm{ts}})=\alpha.
    \end{align}
    Equation~\eqref{eq:def-implicit-ts-lemma} therefore defines $\widetilde u_{\nu,\alpha}^{\mathrm{ts}}$ for all sufficiently large $n$.
\end{itemize}
\end{lemma}

\begin{proof}
We first prove part (i), which concerns the upper-quantile case. Differentiating \eqref{eq:Gup}, we obtain
\[
\bigl(G_{n,\nu}^{\mathrm{up}}\bigr)'(\mathsf q)
=
\psiB(\mathsf q)+\sum_{k=1}^{\nu} n^{-k/2}\bigl(\eta_k^{\mathrm{up}}\psiB\bigr)'(\mathsf q).
\]
Since $K^\mathrm{up}_\varepsilon$ is compact and each $(\eta_k^{\mathrm{up}}\psiB)'$ is smooth, we have
\begin{align*}
    \sup_{\mathsf q\in K^\mathrm{up}_\varepsilon}
\left|
\sum_{k=1}^{\nu} n^{-k/2}\bigl(\eta_k^{\mathrm{up}}\psiB\bigr)'(\mathsf q)
\right|
\le n^{-1/2}\times
\sup_{\mathsf q\in K^\mathrm{up}_\varepsilon}
\left|
\sum_{k=1}^{\nu}\bigl(\eta_k^{\mathrm{up}}\psiB\bigr)'(\mathsf q)
\right|.
\end{align*}
According to Assumption \ref{assume 100}, the coefficients of $\eta_k^{\mathrm{up}}$ are polynomials in moments up to order $k$. Hence
\begin{align}\label{eq61}
    \sup_{\mathsf q\in K^\mathrm{up}_\varepsilon}
\left|
\sum_{k=1}^{\nu} n^{-k/2}\bigl(\eta_k^{\mathrm{up}}\psiB\bigr)'(\mathsf q)
\right|
\le n^{-1/2}\times C_{[K^\mathrm{up}_\varepsilon,B,\nu,\{\E\|\mathbf X_1\|^k\}_{k=1,\ldots,l}]}.
\end{align}
Since $\inf_{\mathsf q\in K^\mathrm{up}_\varepsilon}\psiB(\mathsf q)>0$, it follows that, for all sufficiently large $n$,
\begin{align*}
    \inf_{\mathsf q\in K^\mathrm{up}_\varepsilon}\bigl(G_{n,\nu}^{\mathrm{up}}\bigr)'(\mathsf q)
    \ge \frac12\inf_{\mathsf q\in K^\mathrm{up}_\varepsilon}\psiB(\mathsf q)
    = C_{[K^\mathrm{up}_\varepsilon,B]}.
\end{align*}
Therefore, $G_{n,\nu}^{\mathrm{up}}$ is strictly increasing on $K^\mathrm{up}_\varepsilon$. It remains to prove existence of a solution to $G_{n,\nu}^{\mathrm{up}}(\mathsf q)=\alpha$ for $\alpha\in I_{3\varepsilon}$. By the same smoothness-on-a-compact-set argument,
\begin{align*}
    G_{n,\nu}^{\mathrm{up}}(\mathsf q)=\PsiB(\mathsf q)+O(n^{-1/2})
\end{align*}
uniformly for $\mathsf q\in K^{\mathrm{up}}_{2\varepsilon}$. Hence, for all sufficiently large $n$, we have
\begin{align*}
    G_{n,\nu}^{\mathrm{up}}\bigl(\PsiB^{-1}(2\varepsilon)\bigr)<3\varepsilon
    \qquad\text{and}\qquad
    G_{n,\nu}^{\mathrm{up}}\bigl(\PsiB^{-1}(1-2\varepsilon)\bigr)>1-3\varepsilon.
\end{align*}
Since $G_{n,\nu}^{\mathrm{up}}$ is strictly increasing on $K^{\mathrm{up}}_{\varepsilon}$, it follows that, for every $\alpha\in I_{3\varepsilon}$, there exists a unique $\widetilde u_{\nu,\alpha}^{\mathrm{up}}\in K^{\mathrm{up}}_{2\varepsilon}$ such that $ G_{n,\nu}^{\mathrm{up}}(\widetilde u_{\nu,\alpha}^{\mathrm{up}})=\alpha$. As a conclusion, there exists $n_0(\varepsilon)\in\mathbb{N}$ both the monotonicity and the existence-uniqueness statements hold for all $n\ge n_0(\varepsilon)$.

For part (ii), the proof for the two-sided case is entirely analogous, up to notational changes. The only point requiring attention is that the derivative lower bound is
\[
\bigl(G_{n,\nu}^{\mathrm{ts}}\bigr)'(\mathsf q)
\ge
\inf_{\mathsf q\in K_{\varepsilon}^{\mathrm{ts}}}\psiB(\mathsf q),
\]
with no factor $\frac12$. This is due to the representation in \eqref{eq:Gts}, whose leading term is $2\PsiB(\mathsf q)-1$ rather than $\PsiB(\mathsf q)$, so the leading derivative is $2\psiB(\mathsf q)$.
\end{proof}
\begin{remark}\label{rm4}
The specific choice $\varepsilon<2\varepsilon<3\varepsilon$ is inessential. More generally, one may replace it by any $0<\varepsilon_1<\varepsilon_2<\varepsilon_3<1/2$. The purpose of this ordering is to ensure that
\begin{align*}
K_{\varepsilon_3}^{.}
\subset
K_{\varepsilon_2}^{.}
\subset
K_{\varepsilon_1}^{.}.
\end{align*}
In other words, $\{\text{admissible quantiles}\}
\subset
\{\text{region where a solution exists}\}
\subset
\{\text{region of strict}$
$\text{monotonicity}\}$.
\end{remark}

\subsection{Proof of Theorem \ref{thm:cf uniform}}\label{sec theo13}
With Lemma \ref{lm inter} and the notation established in Section \ref{sec5.3}, we are in a position to derive a uniform upper bound for the Student's $t$ Cornish--Fisher expansions.

\subsubsection{Proof of the One-Sided Case}\label{sec541}
\begin{proof}
Apply Lemma \ref{lm inter}(i). Then, for all sufficiently large $n$, the function $G_{n,\nu}^{\mathrm{up}}$ is strictly increasing on $K_{\varepsilon}^{\mathrm{up}}$ and satisfies
\begin{align*}
\bigl(G_{n,\nu}^{\mathrm{up}}\bigr)'(\mathsf q)\ge \frac12\inf_{\mathsf q\in K_{\varepsilon}^{\mathrm{up}}}\psiB(\mathsf q)
\qquad\text{for all }\mathsf q\in K_{\varepsilon}^{\mathrm{up}},
\end{align*}
and, for every $\alpha\in I_{3\varepsilon}$, there exists a unique $\widetilde u_{\nu,\alpha}^{\mathrm{up}}\in K_{2\varepsilon}^{\mathrm{up}}$ such that $G_{n,\nu}^{\mathrm{up}}(\widetilde u_{\nu,\alpha}^{\mathrm{up}})=\alpha$.

\medskip
\noindent
\emph{Step 1: Analysis of $\left|u_{\nu,\alpha}^{\mathrm{up}}-\widetilde u_{\nu,\alpha}^{\mathrm{up}}\right|$.}
Recall the definition of $u_{\nu,\alpha}^{\mathrm{up}}$ from \eqref{def cfup}. Apply Proposition \ref{prop:CF-existence} with $K=K_{3\varepsilon}^{\mathrm{up}}$. Since $\alpha\in I_{3\varepsilon}$ is equivalent to $\mathsf q\coloneqq\PsiB^{-1}(\alpha)\in K_{3\varepsilon}^{\mathrm{up}}$, \eqref{eq:formal-up} yields
\begin{align*}
    \sup_{\alpha\in I_{3\varepsilon}}
\left|
G_{n,\nu}^{\mathrm{up}}(u_{\nu,\alpha}^{\mathrm{up}})-\alpha
\right|
\le
\mathsf A_{K_{3\varepsilon}^{\mathrm{up}}}^{\mathrm{up}}\times n^{-(\nu+1)/2}.
\end{align*}
Using $G_{n,\nu}^{\mathrm{up}}(\widetilde u_{\nu,\alpha}^{\mathrm{up}})=\alpha$, we obtain
\begin{align*}
    \sup_{\alpha\in I_{3\varepsilon}}
\left|
G_{n,\nu}^{\mathrm{up}}(u_{\nu,\alpha}^{\mathrm{up}})-G_{n,\nu}^{\mathrm{up}}(\widetilde u_{\nu,\alpha}^{\mathrm{up}})
\right|
\le
\mathsf A_{K_{3\varepsilon}^{\mathrm{up}}}^{\mathrm{up}}\times n^{-(\nu+1)/2}.
\end{align*}
The same smoothness argument on a compact set yields $u_{\nu,\alpha}^{\mathrm{up}}-\PsiB^{-1}(\alpha)=O(n^{-1/2})$
uniformly for $\alpha\in I_{3\varepsilon}$. Since $K_{3\varepsilon}^{\mathrm{up}}\subset \mathrm{int}(K_{2\varepsilon}^{\mathrm{up}})$, it follows that $u_{\nu,\alpha}^{\mathrm{up}}\in K_{2\varepsilon}^{\mathrm{up}}$ for all sufficiently large $n$. Therefore,
\begin{align}
    \sup_{\alpha\in I_{3\varepsilon}}\left|u_{\nu,\alpha}^{\mathrm{up}}-\widetilde u_{\nu,\alpha}^{\mathrm{up}}\right|
\le
\frac{2\mathsf A_{K_{3\varepsilon}^{\mathrm{up}}}^{\mathrm{up}}}{\inf_{\mathsf q\in K_{2\varepsilon}^{\mathrm{up}}}\psiB(\mathsf q)}
\times n^{-(\nu+1)/2}.
\label{eq:utilde-uhat}
\end{align}

\medskip
\noindent
\emph{Step 2: Analysis of $\left|u_\alpha^{\mathrm{up}}-\widetilde u_{\nu,\alpha}^{\mathrm{up}}\right|$.}
Recall the Student's $t$ Edgeworth expansion from \eqref{eq:up-edge},
\begin{align*}
    \sup_{\mathsf q\in K_{\varepsilon}^{\mathrm{up}}}
\left|F_n^{\mathrm{up}}(\mathsf q)-G_{n,\nu}^{\mathrm{up}}(\mathsf q)\right|
\le
\mathsf C_{\mathrm{up}}\times n^{-(\nu+1)/2}.
\end{align*}
Let $\alpha\in I_{3\varepsilon}$ and define
\begin{align*}
\delta_n\coloneqq
\frac{4\mathsf C_{\mathrm{up}}}{\inf_{\mathsf q\in K_{\varepsilon}^{\mathrm{up}}}\psiB(\mathsf q)}
\, n^{-(\nu+1)/2}.
\end{align*}
When $n$ is sufficiently large, $\widetilde u_{\nu,\alpha}^{\mathrm{up}}\in K_{2\varepsilon}^{\mathrm{up}}$ and $\widetilde u_{\nu,\alpha}^{\mathrm{up}}\pm \delta_n\in K_{\varepsilon}^{\mathrm{up}}$ by Lemma \ref{lm inter}(i). Using the definition $G_{n,\nu}^{\mathrm{up}}(\widetilde u_{\nu,\alpha}^{\mathrm{up}})=\alpha$, we have
\begin{align*}
    G_{n,\nu}^{\mathrm{up}}(\widetilde u_{\nu,\alpha}^{\mathrm{up}}+\delta_n)>\alpha+\mathsf C_{\mathrm{up}}\times n^{-(\nu+1)/2}\quad\text{and}\quad G_{n,\nu}^{\mathrm{up}}(\widetilde u_{\nu,\alpha}^{\mathrm{up}}-\delta_n)&<\alpha-\mathsf C_{\mathrm{up}}\times n^{-(\nu+1)/2}.
\end{align*}
Applying the Edgeworth expansion from \eqref{eq:up-edge},
\begin{align*}
F_n^{\mathrm{up}}(\widetilde u_{\nu,\alpha}^{\mathrm{up}}+\delta_n)> \alpha,
\qquad
F_n^{\mathrm{up}}(\widetilde u_{\nu,\alpha}^{\mathrm{up}}-\delta_n)< \alpha.
\end{align*}
The definition $u_\alpha^{\mathrm{up}}=\inf\{\mathsf q:F_n^{\mathrm{up}}(\mathsf q)\ge \alpha\}$ then implies
\begin{align*}
\widetilde u_{\nu,\alpha}^{\mathrm{up}}-\delta_n\le u_\alpha^{\mathrm{up}}\le \widetilde u_{\nu,\alpha}^{\mathrm{up}}+\delta_n.
\end{align*}
for all $\alpha\in I_{3\varepsilon}$. Consequently,
\begin{align}
\sup_{\alpha\in I_{3\varepsilon}}\left|u_\alpha^{\mathrm{up}}-\widetilde u_{\nu,\alpha}^{\mathrm{up}}\right|
\le
\frac{4\mathsf C_{\mathrm{up}}}{\inf_{\mathsf q\in K_{\varepsilon}^{\mathrm{up}}}\psiB(\mathsf q)}\, n^{-(\nu+1)/2}.
\label{eq:u-utilde}
\end{align}

\medskip
\noindent
\emph{Step 3: Conclusion.}
Combining \eqref{eq:utilde-uhat} and \eqref{eq:u-utilde}, we obtain
\begin{align*}
    \sup_{\alpha\in I_{3\varepsilon}}\left|u_{\nu,\alpha}^{\mathrm{up}}-u_\alpha^{\mathrm{up}}\right|
\le \frac{\mathsf 2\mathsf A_{K_{3\varepsilon}^{\mathrm{up}}}^{\mathrm{up}}+4\mathsf C_{\mathrm{up}}}{\inf_{\mathsf q\in K_{\varepsilon}^{\mathrm{up}}}\psiB(\mathsf q)}\times n^{-(\nu+1)/2}.
\end{align*}
Since $\varepsilon$ and $3\varepsilon$ may be replaced by any pair $\varepsilon_1<\varepsilon_3$ (see Remark \ref{rm4}), it follows that
\begin{align*}
    \sup_{\alpha\in I_{\varepsilon}}\left|u_{\nu,\alpha}^{\mathrm{up}}-u_\alpha^{\mathrm{up}}\right|
\le \frac{\mathsf 2 \mathsf A_{K_{3\varepsilon}^{\mathrm{up}}}^{\mathrm{up}}+4\mathsf C_{\mathrm{up}}}{\inf_{\mathsf q\in K_{\varepsilon/2}^{\mathrm{up}}}\psiB(\mathsf q)}\times n^{-(\nu+1)/2},
\end{align*}
Since $\mathsf A_{K_{3\varepsilon}^{\mathrm{up}}}^{\mathrm{up}}
\le
\mathsf A_{K_{\varepsilon/2}^{\mathrm{up}}}^{\mathrm{up}}$, we may choose the constant to be
\begin{align}\label{cf constants1}
    \mathsf C_{\mathrm{c.f.}}^{\mathrm{up}}
    \coloneqq
    \frac{2\mathsf A_{K}^{\mathrm{up}}+4\mathsf C_{\mathrm{up}}}
    {\inf_{\mathsf q\in K}\psiB(\mathsf q)},
    \qquad
    K\coloneqq \left[\PsiB^{-1}\!\left(\frac{\varepsilon}{2}\right),\PsiB^{-1}\!\left(1-\frac{\varepsilon}{2}\right)\right].
\end{align}
This completes the proof.
\end{proof}

\subsubsection{Proof of the Two-Sided Case}
\begin{proof}
The proof is entirely analogous to that of the one-sided case, with Lemma \ref{lm inter}(ii) and \eqref{eq:formal-ts} in place of Lemma \ref{lm inter}(i) and \eqref{eq:formal-up}, so we omit the details. The constant can be chosen as
\begin{align}\label{cf constants2}
    \mathsf C^{\mathrm{ts}}_{\mathrm{c.f.}}\coloneqq \frac{\mathsf A_{K}^{\mathrm{ts}}+2\mathsf C_{\mathrm{ts}}^\frac{\varepsilon}{2}}{\inf_{\mathsf q\in K}\psiB(\mathsf q)},\qquad K\coloneqq
\left[
\PsiB^{-1}\!\left(\frac12+\frac{\varepsilon}{4}\right),
\PsiB^{-1}\!\left(1-\frac{\varepsilon}{4}\right)
\right].
\end{align}
\end{proof}

\subsection{Proof of Theorem \ref{thm:resample cf uniform}}\label{sec theo14}
\begin{proof}
To prove Theorem \ref{thm:resample cf uniform}, we revisit Lemma \ref{lm inter} in the resampled setting, with $\hat G_{n,\nu}^{\mathrm{up}}$ and $\hat G_{n,\nu}^{\mathrm{ts}}$. These are obtained by taking the formal expansions of $G_{n,\nu}^{\mathrm{up}}$ and $G_{n,\nu}^{\mathrm{ts}}$ and replacing, in every coefficient, the population moments by the corresponding sample moments.

The only randomness that may affect the uniform convergence conclusions comes from the coefficients of $\eta_k$. These coefficients are already controlled by sufficiently high moments, as in \eqref{eq61}. Therefore, exactly as in the proof of Theorem \ref{theorem resample uniform version}, it suffices to work on the good event $\mathcal E_{\mathrm m}$ from \eqref{good event emp moment}. Consequently, Lemma \ref{lm inter} remains valid for the resampled expansions on $\mathcal E_{\mathrm m}$, while
\[
\PP\!\left(\mathcal E_{\mathrm m}\right)=1- O(n^{-\lambda}).
\]
In addition, \eqref{eq:formal-up} and \eqref{eq:formal-ts} from Proposition \ref{prop:CF-existence} were used in Step 1 of Section \ref{sec541} not as the definition of $\xi_k$, but as approximation properties of $G_{n,\nu}^{\mathrm{up}}$ and $G_{n,\nu}^{\mathrm{ts}}$, with deterministic constants $\mathsf A_K^{\mathrm{up}}$ and $\mathsf A_K^{\mathrm{ts}}$. In the resampled setting, we need the corresponding analogues to hold  for some constants $\hat{\mathsf A}_K^{\mathrm{up}}$ and $\hat{\mathsf A}_K^{\mathrm{ts}}$,
\begin{align*}
&\sup_{\mathsf q\in K}
\left|
\hat G_{n,\nu}^{\mathrm{up}}\!\left(\mathsf q+\sum_{k=1}^{\nu} n^{-k/2}\hat \xi_k^{\mathrm{up}}(\mathsf q)\right)-\PsiB(\mathsf q)
\right|
\le
\hat{\mathsf A}_K^{\mathrm{up}}\times n^{-(\nu+1)/2},\\
&\sup_{\mathsf q\in K}
\left|
\hat G_{n,\nu}^{\mathrm{ts}}\!\left(\mathsf q+\sum_{\substack{2\le k\le \nu\\ k\text{ even}}} n^{-k/2}\hat \xi_k^{\mathrm{ts}}(\mathsf q)\right)-\bigl(2\PsiB(\mathsf q)-1\bigr)
\right|
\le
\hat{\mathsf A}_K^{\mathrm{ts}}\times n^{-(\nu+1)/2},
\end{align*}
except on an event of probability $O(n^{-\lambda})$.  By checking Step 2 in Sections \ref{sec51} and \ref{sec52}, we see that imposing the same good event $\mathcal E_{\mathrm m}$ is sufficient. Therefore, the proof of Theorem \ref{thm:cf uniform} carries over verbatim, with the constants $\mathsf A$ and $\mathsf C$ in \eqref{cf constants1} and \eqref{cf constants2} replaced by their resampled counterparts $\hat{\mathsf A}$ and $\hat{\mathsf C}$. This proves Theorem \ref{thm:resample cf uniform}.
\end{proof}

\subsection{Proof of Corollary \ref{cor1}}\label{appendixcor5}
\begin{proof}
\emph{One-sided case.}
For the upper one-sided case, it suffices to work with the first-order Cornish--Fisher approximation. Thus let $\nu=1$. Then
\begin{align*}
    \left|\hat u_\alpha^{\mathrm{up}}- u_\alpha^{\mathrm{up}}\right|
    \le
    \left|\hat u_\alpha^{\mathrm{up}}- \hat u_{1,\alpha}^{\mathrm{up}}\right|
    +
    \left|u_{1,\alpha}^{\mathrm{up}}- u_{\alpha}^{\mathrm{up}}\right|
    +
    \left|\hat u_{1,\alpha}^{\mathrm{up}}- u_{1,\alpha}^{\mathrm{up}}\right|.
\end{align*}
According to Theorem \ref{thm:resample cf uniform}, we define the good events
\begin{align}
    \mathcal{E}_\mathrm{rcf}^{\nu,\mathrm{up}}\coloneqq & \left\{\sup_{\varepsilon\leq\alpha\leq 1-\varepsilon}\left|\hat u_{\nu,\alpha}^{\mathrm{up}}-\hat u_\alpha^{\mathrm{up}}\right|\le \hat{\mathsf C}_{\mathrm{c.f.}}^{\mathrm{up}}\times n^{-(\nu+1)/2}\right\} \label{rcf good eventup}\\
    \mathcal{E}_\mathrm{rcf}^{\nu,\mathrm{ts}}\coloneqq & \left\{\sup_{\varepsilon\leq\alpha\leq 1-\varepsilon}\left|\hat u_{\nu,\alpha}^{\mathrm{ts}}-\hat u_\alpha^{\mathrm{ts}}\right|\le \hat{\mathsf C}_{\mathrm{c.f.}}^{\mathrm{ts}}\times n^{-(\nu+1)/2}\right\} \label{rcf good eventts}
\end{align}
Applying both Theorem \ref{thm:cf uniform} and Theorem \ref{thm:resample cf uniform} with $\mathcal{E}_\mathrm{rcf}^{1,\mathrm{up}}$, we have
\begin{align*}
&\sup_{\varepsilon\le \alpha\le 1-\varepsilon}
\left|\hat u_{\alpha}^{\mathrm{up}}-u_{\alpha}^{\mathrm{up}}\right|
\le
\sup_{\varepsilon\le \alpha\le 1-\varepsilon}
\left|\hat u_{1,\alpha}^{\mathrm{up}}-u_{1,\alpha}^{\mathrm{up}}\right|
+\frac{\mathsf C_{\mathrm{c.f.}}^{\mathrm{up}}+\hat{\mathsf C}_{\mathrm{c.f.}}^{\mathrm{up}}}{n}\\
\le~&
n^{-1/2}\times \sup_{\mathsf q\in K}\left|\hat{\xi}_1^{\mathrm{up}}(\mathsf q) -\xi_1^{\mathrm{up}}(\mathsf q)\right|
+ \frac{\mathsf C_{\mathrm{c.f.}}^{\mathrm{up}}+\hat{\mathsf C}_{\mathrm{c.f.}}^{\mathrm{up}}}{n},
\qquad
K
\coloneqq
\bigl[\PsiB^{-1}(\varepsilon),\PsiB^{-1}(1-\varepsilon)\bigr].
\end{align*}
We have $\xi_1^{\mathrm{up}}=-\eta_1^{\mathrm{up}}$, and $\eta_1^{\mathrm{up}}(\mathsf q)$ is linear in $b_{[1,0]}$ and $b_{[1,2]}$. Hence
\begin{align*}
\sup_{\mathsf q\in K}\left|\hat{\xi}_1^{\mathrm{up}}(\mathsf q) -\xi_1^{\mathrm{up}}(\mathsf q)\right|
\le
C_{[B,\varepsilon]}\times \sum_{k=0,2}\abs{b_{[1,k]}- \hat b_{[1,k]}}
\end{align*}
Therefore, on $\mathcal{E}_\mathrm{rcf}^{1,\mathrm{up}}$, we have
\begin{align*}
\sup_{\varepsilon\le \alpha\le 1-\varepsilon}
\left|\hat u_{\alpha}^{\mathrm{up}}-u_{\alpha}^{\mathrm{up}}\right|
\le
Cn^{-1/2}\sum_{k=0,2}\abs{b_{[1,k]}- \hat b_{[1,k]}}
+\frac{\mathsf C_{\mathrm{c.f.}}^{\mathrm{up}}+\hat{\mathsf C}_{\mathrm{c.f.}}^{\mathrm{up}}}{n}= O_p(n^{-1})
\end{align*}
under sufficient moment conditions. Hence $\sup_{\varepsilon\le \alpha\le 1-\varepsilon}
\left|\hat u_{\alpha}^{\mathrm{up}}-u_{\alpha}^{\mathrm{up}}\right|
=O_p(n^{-1})$.

\medskip
\noindent
\emph{Two-sided case.}
For the two-sided case, take $\nu=3$. Then
\begin{align*}
    \left|\hat u_\alpha^{\mathrm{ts}}- u_\alpha^{\mathrm{ts}}\right|
    \le
    \left|\hat u_\alpha^{\mathrm{ts}}- \hat u_{3,\alpha}^{\mathrm{ts}}\right|
    +
    \left|u_{3,\alpha}^{\mathrm{ts}}- u_{\alpha}^{\mathrm{ts}}\right|
    +
    \left|\hat u_{3,\alpha}^{\mathrm{ts}}- u_{3,\alpha}^{\mathrm{ts}}\right|.
\end{align*}
Applying both Theorem \ref{thm:cf uniform} and Theorem \ref{thm:resample cf uniform} with $\mathcal E_{\mathrm{rcf}}^{3,\mathrm{ts}}$, we have
\begin{align*}
&\sup_{\varepsilon\le \alpha\le 1-\varepsilon}
\left|\hat u_{\alpha}^{\mathrm{ts}}-u_{\alpha}^{\mathrm{ts}}\right|
\le
\sup_{\varepsilon\le \alpha\le 1-\varepsilon}
\left|\hat u_{3,\alpha}^{\mathrm{ts}}-u_{3,\alpha}^{\mathrm{ts}}\right|
+\frac{\mathsf C_{\mathrm{c.f.}}^{\mathrm{ts}}+\hat{\mathsf C}_{\mathrm{c.f.}}^{\mathrm{ts}}}{n^2}\\
\le~&
n^{-1}\times \sup_{\mathsf q\in K}\left|\hat{\xi}_2^{\mathrm{ts}}(\mathsf q) -\xi_2^{\mathrm{ts}}(\mathsf q)\right|
+ \frac{\mathsf C_{\mathrm{c.f.}}^{\mathrm{ts}}+\hat{\mathsf C}_{\mathrm{c.f.}}^{\mathrm{ts}}}{n^2},
\qquad
K
\coloneqq
\left[
\PsiB^{-1}\!\left(\frac{1+\varepsilon}{2}\right),
\PsiB^{-1}\!\left(1-\frac{\varepsilon}{2}\right)
\right].
\end{align*}
By the remark after Proposition \ref{prop:CF-existence}, $\xi_2^{\mathrm{ts}}=-\frac12\eta_2^{\mathrm{ts}}$, and by Example 1, $\eta_2^{\mathrm{ts}}(\mathsf q)$ is linear in $a_{[2,k]}$ and $b_{[2,k]}$ for $k\in\{1,3,5\}$. Hence
\begin{align*}
\sup_{\mathsf q\in K}\left|\hat{\xi}_2^{\mathrm{ts}}(\mathsf q)-\xi_2^{\mathrm{ts}}(\mathsf q)\right|
\le
C_{[B,\varepsilon]}\times \sum_{k=1,3,5}\left(\abs{a_{[2,k]}-\hat a_{[2,k]}}+\abs{b_{[2,k]}-\hat b_{[2,k]}}\right).
\end{align*}
Therefore, on $\mathcal E_{\mathrm{rcf}}^{3,\mathrm{ts}}$, we have
\begin{align*}
\sup_{\varepsilon\le \alpha\le 1-\varepsilon}
\left|\hat u_{\alpha}^{\mathrm{ts}}-u_{\alpha}^{\mathrm{ts}}\right|
\le
Cn^{-1}\sum_{k=1,3,5}\left(\abs{a_{[2,k]}-\hat a_{[2,k]}}+\abs{b_{[2,k]}-\hat b_{[2,k]}}\right)
+\frac{\mathsf C_{\mathrm{c.f.}}^{\mathrm{ts}}+\hat{\mathsf C}_{\mathrm{c.f.}}^{\mathrm{ts}}}{n^2}
=
O_p(n^{-3/2})
\end{align*}
under sufficient moment conditions. Hence $\sup_{\varepsilon\le \alpha\le 1-\varepsilon}
\left|\hat u_{\alpha}^{\mathrm{ts}}-u_{\alpha}^{\mathrm{ts}}\right|
=
O_p(n^{-3/2})$. Finally, when $B=1$, the term
\[
n^{-2}\bigl(\mathsf C_{\mathrm{c.f.}}^{\mathrm{ts}}+\hat{\mathsf C}_{\mathrm{c.f.}}^{\mathrm{ts}}\bigr)
\]
is replaced by
\[
n^{-3/2}\bigl(\mathsf C_{\mathrm{c.f.}}^{\mathrm{ts}}+\hat{\mathsf C}_{\mathrm{c.f.}}^{\mathrm{ts}}\bigr),
\]
and hence the conclusion remains unchanged.
\end{proof}

\section{Proof of SCB Coverage Error Results}\label{appendix-D}\label{appendix D}
\subsection{Auxiliary Lemmas for Theorems \ref{theorem err 1side} and \ref{theorem err 2side}}
\begin{lemma}\label{lmHqq}
Under Assumption \ref{assume1}, for any fixed $B\geq1$, there exists a constant $L$ such that for any $\mathsf q_1$, $\mathsf q_2$, and $n$,
\begin{align*}
\left|
\PP(T\le \mathsf q_1)-\PP(T\le \mathsf q_2)
\right|
\le L\times \abs{\mathsf q_1- \mathsf q_2}+
2\mathsf C_{\mathrm{up}}\times n^{-3/2},
\end{align*}
\end{lemma}
\begin{proof}
By Theorem \ref{theorem uniform version}, there exists a constant $\mathsf C_{\mathrm{up}}$ such that
\begin{align*}
\sup_{\mathsf q\in\mathbb R}
\left|
\PP(T\le \mathsf q)-H_n(\mathsf q)
\right|
\le
\mathsf C_{\mathrm{up}}\,n^{-3/2},
\end{align*}
where $H_n(\mathsf q)\coloneqq
\Psi_B(\mathsf q)+n^{-1/2}\zeta_1^{\mathrm{up}}(\mathsf q)+n^{-1}\zeta_2^{\mathrm{up}}(\mathsf q)$. Hence, for every pair of real numbers $\mathsf q_1,\mathsf q_2$,
\begin{align*}
\left|
\PP(T\le \mathsf q_1)-\PP(T\le \mathsf q_2)
\right|
\le
\abs{H_n(\mathsf q_1)-H_n(\mathsf q_2)}+2 \mathsf C_{\mathrm{up}}\,n^{-3/2}.
\end{align*}
By inspecting the expressions of $\zeta_1^{\mathrm{up}}$ and $\zeta_2^{\mathrm{up}}$ in \eqref{eq zeta2 upper}, we see that $H_n(\cdot)$ and its first derivative are uniformly bounded in $\mathsf q$ and $n$. Therefore there exists a finite constant $L$ such that
\begin{align*}
\sup_{n\ge 1}\sup_{\mathsf q\in\mathbb R}\abs{H_n'(\mathsf q)}
\le
L.
\end{align*}
By the mean value theorem, $\abs{H_n(\mathsf q_1)-H_n(\mathsf q_2)}
\le
L\times\abs{\mathsf q_1 - \mathsf q_2}$.
\end{proof}

\begin{lemma}[Two-sided analogue of Lemma \ref{lmHqq}]\label{lmHqqts}
Under Assumption \ref{assume1}, for any fixed $B\geq2$ and any $\varepsilon_0>0$, there exists a constant $L_{\varepsilon_0}$ such that for any $\mathsf q_1,\mathsf q_2\ge \varepsilon_0$, and $n$,
\begin{align*}
\left|
\PP(\abs{T}\le \mathsf q_1)-\PP(\abs{T}\le \mathsf q_2)
\right|
\le
L_{\varepsilon_0}\times \abs{\mathsf q_1-\mathsf q_2}
+
2\mathsf C_{\mathrm{ts}}\times n^{-2}.
\end{align*}
\end{lemma}
\begin{proof}
By Theorem \ref{theorem uniform version}, there exists a constant $\mathsf C_{\mathrm{ts}}$ such that
\begin{align*}
\sup_{\mathsf q\ge \varepsilon_0}
\left|
\PP(\abs{T}\le \mathsf q)-H_n(\mathsf q)
\right|
\le
\mathsf C_{\mathrm{ts}}\,n^{-2},
\end{align*}
where $H_n(\mathsf q)\coloneqq
2\Psi_B(\mathsf q)-1+n^{-1}\zeta_2(\mathsf q)$. Hence, for every pair $\mathsf q_1,\mathsf q_2\ge \varepsilon_0$,
\begin{align*}
\left|
\PP(\abs{T}\le \mathsf q_1)-\PP(\abs{T}\le \mathsf q_2)
\right|
\le
\abs{H_n(\mathsf q_1)-H_n(\mathsf q_2)}+2\mathsf C_{\mathrm{ts}}\,n^{-2}.
\end{align*}
By inspecting the expression of $\zeta_2$ in \eqref{eq zeta2}, we see that $H_n'(\cdot)$ is uniformly bounded on $[\varepsilon_0,\infty)$ over all $n$. Therefore there exists a finite constant $L_{\varepsilon_0}$ such that
\begin{align*}
\sup_{n\ge 1}\sup_{\mathsf q\ge \varepsilon_0}\abs{H_n'(\mathsf q)}
\le
L_{\varepsilon_0}.
\end{align*}
By the mean value theorem, $\abs{H_n(\mathsf q_1)-H_n(\mathsf q_2)}
\le
L_{\varepsilon_0}\times \abs{\mathsf q_1-\mathsf q_2}$.
\end{proof}

\subsection{Proof of Theorem \ref{theorem err 1side}}
\begin{proof}
We define
\begin{align*}
    I_0\coloneqq& \bigl|
\PP(T\le u_\alpha^{\mathrm{up}})-\alpha 
\bigr|\\
I_1\coloneqq & \left|
\PP(T\le u_{2,\alpha}^{\mathrm{up}})
-
\PP(T\le u_\alpha^{\mathrm{up}})
\right|\\
I_2\coloneqq & \left|
\PP(T\le \hat u_\alpha^{\mathrm{up}})
-
\PP(T\le \hat u_{2,\alpha}^{\mathrm{up}})
\right|\\
I_3\coloneqq & \left|
\PP(T\le \hat u_{2,\alpha}^{\mathrm{up}})
-
\PP(T\le u_{2,\alpha}^{\mathrm{up}})
\right|
\end{align*}
so that $\bigl|
\PP(T\le \hat u_\alpha^{\mathrm{up}})-\alpha
\bigr|
\eqqcolon I_0+I_1+ I_2 +I_3$.

\begin{enumerate}[label=(\alph*)]
\item Since $u_\alpha^{\mathrm{up}}=\inf\{\mathsf q:\PP(T\le \mathsf q)\ge \alpha\}$, we only have
\[
0\le \PP(T\le u_\alpha^{\mathrm{up}})-\alpha \le \PP(T=u_\alpha^{\mathrm{up}}).
\]
By Theorem 2.3 in \cite{hall2013bootstrap} and Cram\'er's  condition, the maximal point mass is exponentially small. Hence, for some $\epsilon_0>0$,
\[
\sup_{\varepsilon\le \alpha\le 1-\varepsilon} I_0
\le \sup_x \PP(T=x)
=O(e^{-\epsilon_0 n}).
\]

\item By Lemma \ref{lmHqq},
    $ I_1\leq  L\times \left|u_{2,\alpha}^{\mathrm{up}}-u_{\alpha}^{\mathrm{up}}\right|+
2\mathsf C_{\mathrm{up}}\times n^{-3/2}$.
Using Theorem \ref{thm:cf uniform} with $\nu=2$ and taking the supremum, we obtain
\begin{align*}
    \sup_{\varepsilon\leq\alpha\leq 1-\varepsilon} I_1\leq  \left(L\times  \mathsf C_{\mathrm{c.f.}}^{\mathrm{up}}+
2\mathsf C_{\mathrm{up}}\right)\times n^{-3/2}.
\end{align*}

\item By Theorem \ref{thm:resample cf uniform}, define the event
\begin{align*}
    \mathcal{E}_\mathrm{rcf}^{\nu,\mathrm{up}}\coloneqq & \left\{\sup_{\varepsilon\leq\alpha\leq 1-\varepsilon}\left|\hat u_{\nu,\alpha}^{\mathrm{up}}-\hat u_\alpha^{\mathrm{up}}\right|\le \hat{\mathsf C}_{\mathrm{c.f.}}^{\mathrm{up}}\times n^{-(\nu+1)/2}\right\}.
\end{align*}
Since
\begin{align*}
    \left|
\PP(T\le \hat u_\alpha^{\mathrm{up}})
-
\PP(T\le \hat u_{2,\alpha}^{\mathrm{up}})
\right|\leq \left|
\PP(T\le \hat u_\alpha^{\mathrm{up}})
-
\PP(T\le \hat u_{2,\alpha}^{\mathrm{up}})
\right|\times \ind_{\mathcal{E}_\mathrm{rcf}^{2,\mathrm{up}}} + 1- \ind_{\mathcal{E}_\mathrm{rcf}^{2,\mathrm{up}}},
\end{align*}
taking expectations after taking the supremum, and then applying Lemma \ref{lmHqq}, we obtain
\begin{align*}
\sup_{\varepsilon\leq\alpha\leq 1-\varepsilon} I_2= &\sup_{\varepsilon\leq\alpha\leq 1-\varepsilon}\left| 
\PP(T\le \hat u_\alpha^{\mathrm{up}})
-
\PP(T\le \hat u_{2,\alpha}^{\mathrm{up}})
\right|\\
\le\ &
\E\left[\sup_{\varepsilon\leq\alpha\leq 1-\varepsilon}
\left|
\PP(T\le \hat u_\alpha^{\mathrm{up}})
-
\PP(T\le \hat u_{2,\alpha}^{\mathrm{up}})
\right|\times
{\mathcal{E}_\mathrm{rcf}^{2,\mathrm{up}}}
\right]
+
1-\PP({\mathcal{E}_\mathrm{rcf}^{2,\mathrm{up}}})\\
\le\ &
L\times \E\left[\sup_{\varepsilon\leq\alpha\leq 1-\varepsilon}
\left|\hat u_\alpha^{\mathrm{up}}-\hat u_{2,\alpha}^{\mathrm{up}}\right|\times
{\mathcal{E}_\mathrm{rcf}^{2,\mathrm{up}}}
\right]
+
2\mathsf C_{\mathrm{up}}\times n^{-3/2}
+
O(n^{-\lambda})\\
\le\ &
\left(L\times \hat{\mathsf C}_{\mathrm{c.f.}}^{\mathrm{up}}
+
2\mathsf C_{\mathrm{up}}\right)\times n^{-3/2}
+
O(n^{-\lambda})\\
=\ &
O(n^{-3/2}).
\end{align*}

\item Write $\Delta_{2,\alpha}^{\mathrm{up}}\coloneqq \hat u_{2,\alpha}^{\mathrm{up}}-u_{2,\alpha}^{\mathrm{up}}$. Applying Lemma \ref{lmHqq}, we obtain
\begin{align*}
\sup_{\varepsilon\le \alpha\le 1-\varepsilon} I_3
\le&
\E \sup_{\varepsilon\le \alpha\le 1-\varepsilon} \Bigl|
\PP(T\le u_{2,\alpha}^{\mathrm{up}}+\Delta_{2,\alpha}^{\mathrm{up}})
-
\PP(T\le u_{2,\alpha}^{\mathrm{up}})
\Bigr|\\
\le&
L\times \E\left[\sup_{\varepsilon\le \alpha\le 1-\varepsilon} \abs{\Delta_{2,\alpha}^{\mathrm{up}}}\right]
+
2\mathsf C_{\mathrm{up}}\times n^{-3/2}.
\end{align*}
Here
\begin{align*}
    \E\left[\sup_{\varepsilon\le \alpha\le 1-\varepsilon} \abs{\Delta_{2,\alpha}^{\mathrm{up}}}\right]\leq  \frac{C_1}{\sqrt{n}}\times \sum_{k=0,2}\E\abs{b_{[1,k]}-\hat b_{[1,k]}}+ \frac{C_2}{n}\times \sum_{k=1,3,5}\left(\E\abs{b_{[2,k]}-\hat b_{[2,k]}}+\E\abs{a_{[2,k]}-\hat a_{[2,k]}}\right),
\end{align*}
where $C_1$ and $C_2$ are finite constants obtained by taking suprema of smooth functions over a compact set. Since $a_{[i,k]}$ and $b_{[i,k]}$ are polynomials in moments, under sufficiently high moment assumptions,
\begin{align*}
    \E\abs{b_{[i,k]}-\hat b_{[i,k]}}+\E\abs{a_{[i,k]}-\hat a_{[i,k]}}= O(n^{-0.5}).
\end{align*}
Therefore, $\sup_{\varepsilon\le \alpha\le 1-\varepsilon} I_3
=
O(n^{-1})$.
\end{enumerate}
Combining the bounds for $I_0$, $I_1$, $I_2$, and $I_3$, we obtain
\begin{align*}
\sup_{\varepsilon\le \alpha\le 1-\varepsilon}
\bigl|
\PP(T\le \hat u_\alpha^{\mathrm{up}})-\alpha
\bigr|
=
O(n^{-1}).
\end{align*}
\end{proof}

\subsection{Proof of Theorem \ref{theorem err 2side}}

Before presenting the main proof, we first establish an auxiliary lemma showing that both $u_{3,\alpha}$ and $\hat u_{3,\alpha}$ are eventually bounded away from $0$, uniformly over $\varepsilon\le \alpha\le 1-\varepsilon$. This lower bound will be needed in the proof of the main theorem, because these quantities appear in the denominator in the Taylor expansion argument. The deterministic quantity $u_{3,\alpha}$ is straightforward to handle. By contrast, the random quantity $\hat u_{3,\alpha}$ is more delicate and can only be controlled on a suitable good event. Recall that the event $\mathcal E_{\mathrm m}$, defined in \eqref{good event emp moment}, controls the deviations of sample moments, up to a sufficiently high order, from their population counterparts within the prescribed tolerance $\varepsilon_{\mathrm m}$.

\begin{lemma}[A technical lemma for the Taylor expansion argument]\label{lm7}
For any fixed $B\geq2$, define
\begin{align*}
    \mathsf q_{\mathrm{min}}
    \coloneqq
    \PsiB^{-1}\!\left(\frac{1+\varepsilon}{2}\right)>0.
\end{align*}
Then there exists an $n_0$ such that, for all $\varepsilon\le \alpha\le 1-\varepsilon$ and all $n\ge n_0$, the following hold:
\begin{align*}
    u_{3,\alpha}^{\mathrm{ts}}
    \ge
    \frac{\mathsf q_{\mathrm{min}}}{2},
    \qquad
    \hat u_{3,\alpha}^{\mathrm{ts}}\times \ind_{\mathcal E_{\mathrm m}}
    \ge
    \frac{\mathsf q_{\mathrm{min}}}{2}\times \ind_{\mathcal E_{\mathrm m}},
\end{align*}
where the event $\mathcal E_{\mathrm m}$ is defined in \eqref{good event emp moment} and satisfies $\PP(\mathcal E_{\mathrm m})=1-O(n^{-\lambda})$ for some arbitrary $\lambda$.
\end{lemma}
\begin{proof}
Let $\mathsf q\coloneqq\PsiB^{-1}\!\left(\frac{1+\alpha}{2}\right)$. Since $\varepsilon\le \alpha\le 1-\varepsilon$, we have $\mathsf q\in K$, where
\begin{align*}
K\coloneqq
\left[
\PsiB^{-1}\!\left(\frac{1+\varepsilon}{2}\right),
\PsiB^{-1}\!\left(1-\frac{\varepsilon}{2}\right)
\right].
\end{align*}
In particular, $\mathsf q\ge \mathsf q_{\mathrm{min}}$. Recall that $u_{3,\alpha}^{\mathrm{ts}}=\mathsf q+n^{-1}\xi_2^{\mathrm{ts}}(\mathsf q)$ and $\hat u_{3,\alpha}^{\mathrm{ts}}=\mathsf q+n^{-1}\hat\xi_2^{\mathrm{ts}}(\mathsf q)$. Since $\xi_2^{\mathrm{ts}}\in\QB$ is smooth, the constant $C_1\coloneqq \sup_{\mathsf q\in K}\abs{\xi_2^{\mathrm{ts}}(\mathsf q)}$ is finite, so
\begin{align*}
\sup_{\varepsilon\le \alpha\le 1-\varepsilon}\bigl|u_{3,\alpha}^{\mathrm{ts}}-\mathsf q\bigr|
\le \frac{C_1}{n}.
\end{align*}
On the event $\mathcal E_{\mathrm m}$, all estimated coefficients entering $\hat\xi_2^{\mathrm{ts}}$ stay within $\varepsilon_{\mathrm m}$ of their population counterparts. Hence
\begin{align*}
C_2
\coloneqq
\sup_{\substack{
\|\hat a-a\|_\infty\le \varepsilon_{\mathrm m}\\
\|\hat b-b\|_\infty\le \varepsilon_{\mathrm m}
}}
\sup_{\mathsf q\in K}
\abs{\hat\xi_2^{\mathrm{ts}}(\mathsf q)}
\end{align*}
is well defined, and therefore
\begin{align*}
\sup_{\varepsilon\le \alpha\le 1-\varepsilon}\bigl|\hat u_{3,\alpha}^{\mathrm{ts}}-\mathsf q\bigr|\times\ind_{\mathcal E_{\mathrm m}}
\le \frac{C_2}{n}.
\end{align*}
Choose $n_0$ so large that $\max\{C_1,C_2\}/n\le \mathsf q_{\mathrm{min}}/2$ for all $n\ge n_0$. Then for all $\varepsilon\le \alpha\le 1-\varepsilon$ and $n\ge n_0$,
\begin{align*}
u_{3,\alpha}^{\mathrm{ts}}
\ge
\mathsf q-\frac{C_1}{n}
\ge
\mathsf q_{\mathrm{min}}-\frac{\mathsf q_{\mathrm{min}}}{2}
=
\frac{\mathsf q_{\mathrm{min}}}{2},
\end{align*}
and similarly,
\begin{align*}
\hat u_{3,\alpha}^{\mathrm{ts}}\times\ind_{\mathcal E_{\mathrm m}}
\ge
\left(\mathsf q-\frac{C_2}{n}\right)\times\ind_{\mathcal E_{\mathrm m}}
\ge
\frac{\mathsf q_{\mathrm{min}}}{2}\times\ind_{\mathcal E_{\mathrm m}}.
\end{align*}
This proves the claim.
\end{proof}

\subsubsection{Proof of the Two-Sided Case}
\begin{proof}
Similar to the proof of Theorem \ref{theorem err 1side}, we define
\begin{align*}
    I_0\coloneqq& \bigl|
\PP(\abs{T}\le u_\alpha^{\mathrm{ts}})-\alpha 
\bigr|\\
I_1\coloneqq & \left|
\PP(\abs{T}\le u_{3,\alpha}^{\mathrm{ts}})
-
\PP(\abs{T}\le u_\alpha^{\mathrm{ts}})
\right|\\
I_2\coloneqq & \left|
\PP(\abs{T}\le \hat u_\alpha^{\mathrm{ts}})
-
\PP(\abs{T}\le \hat u_{3,\alpha}^{\mathrm{ts}})
\right|\\
I_3\coloneqq & \left|
\PP(\abs{T}\le \hat u_{3,\alpha}^{\mathrm{ts}})
-
\PP(\abs{T}\le u_{3,\alpha}^{\mathrm{ts}})
\right|
\end{align*}
so that $\bigl|
\PP(\abs{T}\le \hat u_\alpha^{\mathrm{ts}})-\alpha
\bigr|
\le
I_0+I_1+I_2+I_3$. The terms $I_0$, $I_1$, $I_2$, and $I_3$ can be handled in the same way as in the one-sided case, except that we now use the two-sided versions with $\nu=3$ of the corresponding ingredients. The resulting bounds are
\begin{align*}
    \sup_{\varepsilon\le \alpha\le  1-\varepsilon} I_0
\le & O(e^{-\epsilon_0 n})\\
 \sup_{\varepsilon\le \alpha\le 1-\varepsilon} I_1 = & O({n^{-2}})\\
  \sup_{\varepsilon\le \alpha\le 1-\varepsilon} I_2 = & O({n^{-2}})\\
  \sup_{\varepsilon\le \alpha\le 1-\varepsilon} I_3 = & O({n^{-3/2}}).
\end{align*}
Therefore, the coverage error can be shown to be $O(n^{-3/2})$ without requiring $B\geq 2$. However, in order to establish the sharper order $O(n^{-2})$, we must treat $I_3$ more carefully. 

\medskip
\noindent
\emph{Step 1: Write $I_3$ as the difference of $J(x)$ functions with different radii.}
For $I_3$, we have
\begin{align*}
    I_3\coloneqq \left|
\PP(\abs{T}\le \hat u_{3,\alpha}^{\mathrm{ts}})
-
\PP(\abs{T}\le u_{3,\alpha}^{\mathrm{ts}})
\right|.
\end{align*}
Recall the definition
\begin{align*}
    u_{3,\alpha}^{\mathrm{ts}}
=&
\mathsf q+n^{-1}\times \xi_1^{\mathrm{ts}}(\mathsf q),
\quad\text{and}\quad  \hat u_{3,\alpha}^{\mathrm{ts}}
=
\mathsf q+n^{-1}\times \hat \xi_1^{\mathrm{ts}}(\mathsf q), \qquad \text{where }\mathsf q\coloneqq\PsiB^{-1}\!\left(\frac{1+\alpha}{2}\right).
\end{align*}
Recall Step 4 from Section \ref{seca1},
\begin{align*}
    U({\mathcal X})\times \ind_{\hat{\mathcal E}_{\nu=3}}= \left(U_0({\mathcal X})+\frac{1}{n}\times U_1({\mathcal X})+R_{[1]}+R_{[2]}\right)\times \ind_{\hat{\mathcal E}_{\nu=3}}
\end{align*}
with $ \left|(R_{[1]}+R_{[2]}) \times \ind_{\hat{\mathcal E}_{\nu=3}}\right|\leq C_{[\hat{\mathsf{C}}_3,B]} \times \frac{1}{n^2}$. As in Section \ref{seca1}, we define
\begin{align*}
    x= \frac{\sqrt n A_s(\bar{\mathbf{X}})}{\mathsf q}
\end{align*}
so that $U_0({\mathcal X})$ can be simplified as a function $J_1(x)$ and $U_1({\mathcal X})$ as $J_2(x)$, namely
\begin{align*}
    J_1(x)= & \dfrac{\Gamma\!\left(\frac B2,\frac{Bx^2}{2}\right)}{\Gamma\!\left(\frac B2\right)}\\
    J_2(x)= & -\sum_{k=1,3,5}\hat{a}_{[2,k]}\times\frac{B^{\frac{B+k+1}{2}}\Gamma(1+\frac{k}{2})}{2^{\frac{B}{2}-1}\sqrt{\pi} \times\Gamma(\frac{1+B+k}{2})}\times \left( \abs{x}^{{B+k-1}}\times \exp(-\frac{Bx^2}{2})\right).
\end{align*}
Hence the principal term is
\begin{align*}
    J_1(x)+\frac{1}{n}\times  J_2(x).
\end{align*}
Both $J_1(x)$ and $J_2(x)$ can also be viewed as functions of $\mathsf q$. Their derivatives with respect to $\mathsf q$ are
\begin{align*}
\frac{d}{d \mathsf  q}J\!\left(\frac{a}{\mathsf q}\right)
=
-\frac{x}{\mathsf q}J'(x)\qquad\text{and}\qquad \frac{d^2}{d \mathsf q^2}J\!\left(\frac{a}{\mathsf  q}\right)
=
\frac{x^2}{\mathsf  q^2}J''(x)+\frac{2x}{\mathsf q^2}J'(x).
\end{align*}

\medskip
\noindent
\emph{Step 2: The difference of $n^{-1}\times J_2$ is $O(n^{-2.5})$.}
For brevity, we write $\hat{u}$ for $\hat u_{3,\alpha}^{\mathrm{ts}}$ and ${u}$ for $u_{3,\alpha}^{\mathrm{ts}}$. Define
\begin{align*}
    \delta= \hat{u}-u =\frac{1}{n}\left(\hat{\xi}_2^{\mathrm{ts}}(\mathsf q)-\xi_2^{\mathrm{ts}}(\mathsf q)\right).
\end{align*}
Let $\theta\in(0,1)$ be the point arising from the Lagrange remainder, and denote $\tilde u = u+\theta \delta $. Define
\begin{align*}
    r= \frac{\sqrt n A_s(\bar{\mathbf{X}})}{{u}}\qquad\hat r= \frac{\sqrt n A_s(\bar{\mathbf{X}})}{\hat{u}}\qquad   \tilde r= \frac{\sqrt n A_s(\bar{\mathbf{X}})}{\tilde u}
\end{align*}
We have $J_2(\cdot)\in C^2(\mathbb R)$ when $B\geq 2$. Hence the Taylor expansion gives
\begin{align*}
   J_2(\hat{r})-J_2(r) =   -{r}J'_2( r) \times \frac{\delta}{u} + \left({\tilde{r}^2}J''_2(\tilde r)+{2 \tilde r}J'_2(\tilde r)\right)\times \frac{\delta^2}{\tilde{u}^2}
\end{align*}
It is immediate from the exponential decay that both $\abs{{r}J'_2(r)}$ and $\abs{{{r}^2}J''_2(r)+{2r}J'_2(r)}$ are uniformly bounded for all $r\geq 0$. Therefore, we only need to control the orders of
\begin{align*}
    \sup_{\varepsilon\le \alpha\le 1-\varepsilon}\left|\frac{\delta}{u}\right|\qquad \text{and} \qquad \sup_{\varepsilon\le \alpha\le 1-\varepsilon}\left|\frac{\delta}{\tilde u}\right|^2
\end{align*}
Define $\mathsf q_\mathrm{min}\coloneqq\PsiB^{-1}\!\left(\frac{1+\varepsilon}{2}\right)$. Lemma \ref{lm7} implies that there exists an $n_0$ such that both of the following hold for $n\geq n_0$:
\begin{align*}
    &\E \left[\sup_{\varepsilon\le \alpha\le 1-\varepsilon}\left|\frac{\delta}{u}\right| \right] \leq \frac{2}{\mathsf q_\mathrm{min}} \times \E \left[\abs{{\delta}} \right] \leq \frac{2}{n\times \mathsf q_\mathrm{min}}\E \left|\hat{\xi}_2^{\mathrm{ts}}(\mathsf q)-\xi_2^{\mathrm{ts}}(\mathsf q)\right| = O(n^{-3/2})\\
    &\E \left[\sup_{\varepsilon\le \alpha\le 1-\varepsilon}\left|\frac{\delta}{\tilde u}\right|^2 \times\ind_{\mathcal E_{\mathrm m}}\right] \leq \frac{4}{\mathsf q^2_\mathrm{min}} \times \E \left[\abs{{\delta}}^2 \right] = O(n^{-3})
\end{align*}
As a consequence,
\begin{align*}
   \frac{1}{n}\times\sup_{\varepsilon\le \alpha\le 1-\varepsilon} \abs{J_2(\hat{r})-J_2(r)}\times\ind_{\mathcal E_{\mathrm m}} =O(n^{-2.5}).
\end{align*}

\medskip
\noindent
\emph{Step 3: The difference of $J_1$ is $O(n^{-2})$.}
When $B\geq 2$, the Taylor expansion gives
\begin{align*}
   J_1(\hat{r})-J_1(r) =&   -{r}J'_1( r) \times \frac{\delta}{u} + \left({\tilde{r}^2}J''_1(\tilde r)+{2 \tilde r}J'_1(\tilde r)\right)\times \frac{\delta^2}{\tilde{u}^2}\\
   \eqqcolon &  J_1^A + J_1^B
\end{align*}
More specifically,
\begin{align*}
J_1'(r)
&=
- C_B \times r \abs{r}^{B-2}\exp(-Br^2/2),
\\
J_1''(r)
&=
C_B \times \abs{r}^{B-2}\bigl(Br^2-(B-1)\bigr)\exp(-Br^2/2)
\end{align*}
with $C_B \coloneqq \frac{B^{B/2}}{2^{B/2-1}\Gamma(B/2)}$.
By the same argument as above, we can bound the second term on the event ${\mathcal E_{\mathrm m}}$:
\begin{align*}
    \E \left[\sup_{\varepsilon\le \alpha\le 1-\varepsilon}\left|{\tilde{r}^2}J''_1(\tilde r)+{2 \tilde r}J'_1(\tilde r)\right|\times\left|\frac{\delta}{\tilde u}\right|^2 \times\ind_{\mathcal E_{\mathrm m}}\right] = O(n^{-3}).
\end{align*}
The first term
\begin{align*}
J_1^A\coloneqq -{\delta}\times {r}J'_1(r)/{u}    
\end{align*}
is typically of the same order as ${\delta}$, namely $O(n^{-3/2})$. However, we observe that
\begin{align*}
     -{r}J'_1( r) \times \frac{\delta}{u}
    = &C_B \times \abs{r}^{B}\exp(-\frac{Br^2}{2})\times \frac{\delta}{u}\\
    = &\frac{1}{n}\times \frac{C_B}{u} \times g_B(r)\times \left(\hat{\xi}_2^{\mathrm{ts}}(\mathsf q)-\xi_2^{\mathrm{ts}}(\mathsf q)\right)
\end{align*}
with 
\begin{align*}
    g_B(x)= \abs{x}^B \exp(-\frac{B}{2}x^2).
\end{align*}
Here $g_B(x)$ is an even function in $C^2(\mathbb R)$ when $B\geq 2$. It is also straightforward to verify that $g_B(x)$ satisfies the other conditions required by Proposition \ref{theo:expect-cancellation-good}. In addition, $\hat{\xi}_2^{\mathrm{ts}}(\mathsf q)-\xi_2^{\mathrm{ts}}(\mathsf q)$ is the difference between a polynomial in moments and its plug-in sample version, with $\mathsf q$ treated as fixed. Therefore, by applying Proposition \ref{theo:expect-cancellation-good}, we obtain
\begin{equation}\label{eq98}
\begin{split}
      \left|\E \left[{r}J'_1( r) \times \frac{\delta}{u}\right]\right| =& \frac{C_B}{n\times u_{3,\alpha}^{\mathrm{ts}}}\times \left|\E \left[g_B(r)\times \left(\hat{\xi}_2^{\mathrm{ts}}(\mathsf q)-\xi_2^{\mathrm{ts}}(\mathsf q)\right)\right]\right|\\
   \leq &\frac{C_B}{n^2\times u_{3,\alpha}^{\mathrm{ts}}}\times  C_{[\mathsf C_1,\tilde A,M,C_{[g]},\{\E\|\mathbf X_1\|^k\}_{k=1,\ldots,l}]}.
\end{split}
\end{equation}
Note that the dependence on $\mathsf C_1,\tilde A,M,C_{[g]},\{\E\|\mathbf X_1\|^k\}_{k=1,\ldots,l}$ can all be regarded as fixed, except for the monomial $M$, because $\mathsf q$ varies with $\alpha$ and $n$, and determines the coefficients of $M$. However, since $\mathsf q$ is eventually confined to a known compact set, this constant can still be chosen uniformly in $n$. Together with $\liminf_{n\xrightarrow[]{}\infty}u_{3,\alpha}^{\mathrm{ts}}\geq {\mathsf q_\mathrm{min}}/2$, as implied by Lemma \ref{lm7}, we conclude that
\begin{align*}
   \left|\E \left[{r}J'_1( r) \times \frac{\delta}{u}\right]\right| =& O(n^{-2}).
\end{align*}

\medskip
\noindent
\emph{Conclusion.}
\begin{align*}
   &\sup_{\varepsilon\le \alpha\le 1-\varepsilon} I_3\\
   = & \sup_{\varepsilon\le \alpha\le 1-\varepsilon} \left|
\PP(\abs{T}\le \hat u_{3,\alpha}^{\mathrm{ts}})
-
\PP(\abs{T}\le u_{3,\alpha}^{\mathrm{ts}})
\right|\\
= & \sup_{\varepsilon\le \alpha\le 1-\varepsilon} \left|
\E\left[\PP(\abs{T}\le \hat u_{3,\alpha}^{\mathrm{ts}}\mid \mathcal X)
-
\PP(\abs{T}\le u_{3,\alpha}^{\mathrm{ts}}\mid \mathcal X)\right]
\right|\\
\leq  & \sup_{\varepsilon\le \alpha\le 1-\varepsilon} \left|
\E\left[\left(\PP(\abs{T}\le \hat u_{3,\alpha}^{\mathrm{ts}}\mid \mathcal X)
-
\PP(\abs{T}\le u_{3,\alpha}^{\mathrm{ts}}\mid \mathcal X)\right)\times \ind_{{\hat{\mathcal E}_{\nu=3}}\cap\mathcal E_{\mathrm m}}\right]
\right|+O(n^{-\lambda})\\
\leq  & \sup_{\varepsilon\le \alpha\le 1-\varepsilon} \left|
\E\left[\left(J_1(\hat r)-J_1(r)+\frac{1}{n}\times  \left(J_2(\hat r)- J_2(r)\right)\right)\times \ind_{{\hat{\mathcal E}_{\nu=3}}\cap\mathcal E_{\mathrm m}}\right]
\right|+O(n^{-\lambda})+ C_{[\hat{\mathsf{C}}_3,B]} \times {n^{-2}}
\end{align*}
The first term satisfies
\begin{align*}
    &\sup_{\varepsilon\le \alpha\le 1-\varepsilon} \left|
\E\left[\left(J_1(\hat r)-J_1(r)+\frac{1}{n}\times  \left(J_2(\hat r)- J_2(r)\right)\right)\times \ind_{{\hat{\mathcal E}_{\nu=3}}\cap\mathcal E_{\mathrm m}}\right]
\right|\\
=&\sup_{\varepsilon\le \alpha\le 1-\varepsilon} \left|
\E\left[\left(J_1(\hat r)-J_1(r)\right)\times \ind_{{\hat{\mathcal E}_{\nu=3}}\cap\mathcal E_{\mathrm m}}\right]
\right|+\frac{1}{n}\times\sup_{\varepsilon\le \alpha\le 1-\varepsilon} \left|
\E\left[\left(  J_2(\hat r)- J_2(r)\right)\times \ind_{{\hat{\mathcal E}_{\nu=3}}\cap\mathcal E_{\mathrm m}}\right]
\right|\\
\leq &\sup_{\varepsilon\le \alpha\le 1-\varepsilon} \left|
\E\left[\left(J_1(\hat r)-J_1(r)\right)\times \ind_{{\hat{\mathcal E}_{\nu=3}}\cap\mathcal E_{\mathrm m}}\right]
\right|+\frac{1}{n}\E \left[\sup_{\varepsilon\le \alpha\le 1-\varepsilon} \left|
\left(  J_2(\hat r)- J_2(r)\right)\times \ind_{{\hat{\mathcal E}_{\nu=3}}\cap\mathcal E_{\mathrm m}}
\right|\right]\\
\leq &\sup_{\varepsilon\le \alpha\le 1-\varepsilon} \left|
\E\left[\left(J_1(\hat r)-J_1(r)\right)\times \ind_{{\hat{\mathcal E}_{\nu=3}}\cap\mathcal E_{\mathrm m}}\right]
\right|+O(n^{-2.5})
\end{align*}
The first term is
\begin{align*}
    &\sup_{\varepsilon\le \alpha\le 1-\varepsilon} \left|
\E\left[\left(J_1^A+J_1^B\right)\times \ind_{{\hat{\mathcal E}_{\nu=3}}\cap\mathcal E_{\mathrm m}}\right]
\right|\\
\leq & \sup_{\varepsilon\le \alpha\le 1-\varepsilon} \left|
\E\left[J_1^A\times \ind_{{\hat{\mathcal E}_{\nu=3}}\cap\mathcal E_{\mathrm m}}\right]
\right|+\sup_{\varepsilon\le \alpha\le 1-\varepsilon} \left|
\E\left[J_1^B\times \ind_{{\hat{\mathcal E}_{\nu=3}}\cap\mathcal E_{\mathrm m}}\right]
\right|\\
\leq & \sup_{\varepsilon\le \alpha\le 1-\varepsilon} \left|
\E\left[J_1^A\right]
\right|+O(n^{-\lambda})+\E\left[\sup_{\varepsilon\le \alpha\le 1-\varepsilon} \left|
J_1^B\times \ind_{{\hat{\mathcal E}_{\nu=3}}\cap\mathcal E_{\mathrm m}}\right|\right]\\
\leq & \sup_{\varepsilon\le \alpha\le 1-\varepsilon} \left|
\E\left[J_1^A\right]
\right|+O(n^{-\lambda})+O(n^{-3}).
\end{align*}
Using \eqref{eq98} with some universal $C_0$, for $n\geq n_0$ where $n_0$ is from Lemma \ref{lm7} we have
\begin{align*}
     \sup_{\varepsilon\le \alpha\le 1-\varepsilon} \left|
\E\left[J_1^A\right]
\right|\leq   \sup_{\varepsilon\le \alpha\le 1-\varepsilon} \frac{C_B\times C_0}{n^2\times u_{3,\alpha}^{\mathrm{ts}}}  \leq  \frac{2C_B C_0}{n^2\times \mathsf q_\text{min}} = O(n^{-2})
\end{align*}
This completes the proof.
\end{proof}

\end{appendix}

\bibliographystyle{imsart-number}
\bibliography{bootstrap}

\end{document}